\documentclass[a4paper]{article}

\usepackage[latin1]{inputenc}
\usepackage{amsmath}
\usepackage{amssymb}
\usepackage{amsthm}
\usepackage{latexsym}
\usepackage[dvips]{graphicx,color}
\usepackage{float}
\usepackage{array}
\usepackage{xcolor}
\usepackage{enumerate} 
\usepackage[shortlabels]{enumitem}
\usepackage[a4paper,margin=1.5in]{geometry}
\usepackage{authblk}
\usepackage{url}

\usepackage{amsthm}
\usepackage{float}
\usepackage{microtype}
\usepackage[shortlabels]{enumitem}
\graphicspath{{./figs/}}

\title{A Finitary Analogue of the Downward L\"owenheim-Skolem Property}
\author{Abhisekh Sankaran}
\affil{Department of Computer Science and Engineering,\\ Indian
  Institute of Technology (IIT) Bombay, India\\ \texttt{abhisekh@cse.iitb.ac.in}}

\date{}

\newcommand{\tbf}[1]{\textbf{#1}}
\newcommand{\mc}[1]{\mathcal{#1}}
\newcommand{\und}[1]{\underline{#1}}
\newcommand{\mf}[1]{\mathfrak{#1}}

\newcommand{\fequiv}[1]{\ensuremath{\equiv_{#1, \fo}}}
\newcommand{\mequiv}[1]{\ensuremath{\equiv_{#1, \mso}}}
\newcommand{\lequiv}[1]{\ensuremath{\equiv_{#1, \mc{L}}}}

\newcommand{\lth}[2]{\ensuremath{\mathsf{Th}_{#1,
      \mc{L}}(#2)}}

\newcommand{\lef}{\ensuremath{\mc{L}\text{-}\mathsf{EF}}}
\newcommand{\fef}{\ensuremath{\fo\text{-}\mathsf{EF}}}
\newcommand{\mef}{\ensuremath{\mso\text{-}\mathsf{EF}}}

\newcommand{\fo}{\ensuremath{\text{FO}}}
\newcommand{\mso}{\ensuremath{\text{MSO}}}

\newcommand{\lt}{{{\L}o{\'s}-Tarski}}

\newcommand{\glt}[1]{\ensuremath{\mathsf{GLT}({#1})}}

\newcommand{\lebsp}[1]{\ensuremath{\mc{L}\text{-}\mathsf{EBSP}({#1})}}
\newcommand{\febsp}[1]{\ensuremath{\fo\text{-}\mathsf{EBSP}({#1})}}
\newcommand{\mebsp}[1]{\ensuremath{\mso\text{-}\mathsf{EBSP}({#1})}}
\newcommand{\lebspcond}{\ensuremath{\mc{L}\text{-}\mathsf{EBSP}\text{-}\mathsf{condition}}}
\newcommand{\febspcond}{\ensuremath{\fo\text{-}\mathsf{EBSP}\text{-}\mathsf{condition}}}

\newcommand{\reflebsp}{\ensuremath{\mc{L}\text{-}\mathsf{EBSP}}}
\newcommand{\reffebsp}{\ensuremath{\fo\text{-}\mathsf{EBSP}}}
\newcommand{\refmebsp}{\ensuremath{\mso\text{-}\mathsf{EBSP}}}

\newcommand{\lwitfn}[1]{\ensuremath{\theta_{(#1, \mc{L})}}}
\newcommand{\fwitfn}[1]{\ensuremath{\theta_{(#1, \fo)}}}
\newcommand{\mwitfn}[1]{\ensuremath{\theta_{(#1, \mso)}}}

\newcommand{\dls}{downward L\"owenheim-Skolem}
\newcommand{\dlsfull}{downward L\"owenheim-Skolem theorem}

\newcommand{\ldelta}[2]{\ensuremath{\Delta_{#1, \mc{L}, #2}}}

\newcommand{\cl}[1]{\ensuremath{\mathcal{#1}}}

\newcommand{\words}{\mathsf{Words}}

\newcommand{\eat}[1]{}

\newcommand{\tree}[1]{\ensuremath{\mathsf{#1}}}

\newcommand{\nesword}[1]{\ensuremath{\mathsf{#1}}}
\newcommand{\nestedwords}[1]{\mathsf{Nested}\text{-}\mathsf{words}(#1)}

\theoremstyle{plain}
\newtheorem{theorem}{Theorem}[section]

\newtheorem{prop}[theorem]{Proposition}
\newtheorem{corollary}[theorem]{Corollary}

\newtheorem{lemma}[theorem]{Lemma}

\theoremstyle{definition}
\newtheorem{definition}[theorem]{Definition}

\theoremstyle{remark}
\newtheorem*{discussion}{Discussion}

\begin{document}

\maketitle

\begin{abstract}
We present a model-theoretic property of finite structures, that can
be seen to be a finitary analogue of the well-studied downward
L\"owenheim-Skolem property from classical model theory.  We call this
property as the \emph{$\mc{L}$-equivalent bounded substructure
  property}, denoted $\reflebsp$, where $\mc{L}$ is either $\fo$ or
$\mso$. Intuitively $\reflebsp$ states that a large finite structure
contains a small ``logically similar'' substructure, where logical
similarity means indistinguishability with respect to sentences of
$\mc{L}$ having a given quantifier nesting depth. It turns out that
this simply stated property is enjoyed by a variety of classes of
interest in computer science: examples include various classes of
posets, such as regular languages of words, trees (unordered, ordered
or ranked) and nested words, and various classes of graphs, such as
cographs, graph classes of bounded tree-depth, those of bounded
shrub-depth and $n$-partite cographs. Further, $\reflebsp$ remains
preserved in the classes generated from the above by operations that
are implementable using quantifier-free translation schemes.  We show
that for natural tree representations for structures that all the
aforementioned classes admit, the small and logically similar
substructure of a large structure can be computed in time linear in
the size of the representation, giving linear time fixed parameter
tractable (f.p.t.) algorithms for checking $\mc{L}$ definable
properties of the large structure. We conclude by presenting a
strengthening of $\reflebsp$, that asserts ``logical self-similarity
at all scales'' for a suitable notion of scale. We call this the
\emph{logical fractal} property and show that most of the classes
mentioned above are indeed, logical fractals.
\end{abstract}

\section{Introduction}

The downward L\"owenheim-Skolem theorem is one of the earliest results
of classical model theory. This theorem, first proved by L\"owenheim
in 1915~\cite{lowenheim}, states that if a first order (henceforth,
FO) theory over a countable vocabulary has an infinite model, then it
has a countable model. In the mid-1920s, Skolem came up with a more
general statement: any structure $\mf{A}$ over a countable vocabulary
has a countable ``\fo-similar'' substructure. Here, ``\fo-similarity''
of two given structures means that the structures agree on all
properties than can be expressed in $\fo$. This result of Skolem was
further generalized by Mal'tsev in 1936~\cite{maltsev}, to what is
considered as the modern statement of the {\dls} theorem: for any
infinite cardinal $\kappa$, any structure $\mf{A}$ over a countable
vocabulary has an elementary substructure (an \fo-similar substructure
having additional properties) that has size at most $\kappa$.  The
{\dlsfull} is one of the most important results of classical model
theory, and indeed as Lindstr\"om showed in 1969~\cite{lindstrom},
$\fo$ is the only logic (having certain well-defined and reasonable
closure properties) that satisfies this theorem, along with the
(countable) compactness theorem.

The {\dls} theorem is a statement intrinsically of infinite
structures, and hence does not make sense in the finite when taken as
is. While preservation and interpolation theorems from classical model
theory have been actively studied over finite
structures~\cite{gurevich-ajtai, ajtai-gurevich94, rosen,
  gradel-rosen, gurevich84, rosen-thesis, rosen-weinstein,
  gurevich-alechina, stolboushkin, dawar-hom, rossman-hom,
  dawar-pres-under-ext, nicole-lmcs-15}, there is very little study of
the {\dls} theorem (or adaptations of it) in the finite
(\cite{grohe-dls, vaananen-dls} seem to be the only studies of this
theorem in the contexts of finite and pseudo-finite structures
respectively). In this paper, we take a step towards addressing this
issue. Specifically, we formulate a finitary analogue of the
model-theoretic property contained in the {\dlsfull}, and show that
classes of finite structures satisfying this analogue indeed abound in
computer science.  We call this analogue the \emph{$\mc{L}$-equivalent
  bounded substructure property}, denoted $\lebsp{\cl{S}}$, where
$\mc{L}$ is one of the logics $\fo$ or $\mso$, and $\cl{S}$ is a class
of finite structures (Definition~\ref{definition:lebsp}). Intuitively,
this property states that over $\cl{S}$, for each $m$, every structure
$\mf{A}$ contains a small substructure $\mf{B}$ that is
``$\mc{L}\left[m\right]$-similar'' to $\mf{A}$, where
$\mc{L}\left[m\right]$ is the class of all sentences of $\mc{L}$ that
have quantifier nesting depth at most $m$. In other words, $\mf{B}$
and $\mf{A}$ agree on all properties that can be described in
$\mc{L}[m]$. The bound on the size of $\mf{B}$ is given by a
``witness'' function that depends only on $m$ (when $\mc{L}$ and
$\cl{S}$ are fixed). It is easily seen that $\lebsp{\cl{S}}$ has
strong resemblance to the model-theoretic property contained in the
{\dlsfull}, and can very well be seen as a finitary analogue of a
version of the {\dlsfull} that is ``intermediate'' between the
versions of this theorem by Skolem and Mal'tsev.

The motivation to define $\lebsp{\cl{S}}$ came from our investigations
over finite structures, of a generalization of the classical {\lt}
preservation theorem from model theory, that was proved
in~\cite{abhisekh-apal}. This generalization, called the
\emph{generalized {\lt} theorem at level $k$}, denoted $\glt{k}$,
gives a semantic characterization, over arbitrary structures, of
sentences in prenex normal form, whose quantifier prefixes are of the
form $\exists^k \forall^*$, i.e. a sequence of $k$ existential
quantifiers followed by zero or more universal quantifiers. The {\lt}
theorem is a special case of $\glt{k}$ when $k$ equals
0. Unfortunately, $\glt{k}$ fails over all finite structures for all
$k \ge 0$ (like most preservation theorems do~\cite{rosen}), and worse
still, also fails for all $k \ge 2$, over the special classes of
finite structures that are acyclic, of bounded degree, or of bounded
tree-width, which were identified by Atserias, Dawar and
Grohe~\cite{dawar-pres-under-ext} to satisfy the {\lt} theorem. This
motivated the search for new (and possibly abstract) structural
properties of classes of finite structures, that admit $\glt{k}$ for
each $k$. It is in this context that a version of $\lebsp{\cl{S}}$ was
first studied in~\cite{abhisekh-mfcs}. The present paper takes that
study much ahead. (Most of the results of this paper are contained in
the author's Ph.D. thesis~\cite{arxiv-self-thesis}.)

The contributions of this paper are as described below.

\vspace{2pt}
\emph{1. A variety of classes of interest in computer science satisfy
  $\reflebsp$}: Our property presents a unified framework, via logic,
for studying a variety of classes of finite structures that are of
interest in computer science. The classes that we consider are broadly
of two kinds: special kinds of labeled posets and special kinds of
graphs. For the case of labeled posets, we show $\reflebsp$ holds for
words, trees (of various kinds such as unordered, ordered, ranked, or
``partially'' ranked), and nested words over a finite alphabet, and
all regular subclasses of these
(Theorem~\ref{theorem:words-and-trees-and-nested-words-satisfy-lebsp}). For
each of these classes, we also show that $\reflebsp$ holds with
computable witness functions.  While words and trees have had a long
history of studies in the literature, nested words are much
recent~\cite{alur-madhu}, and have attracted a lot of attention as
they admit a seamless generalization of the theory of regular
languages and are also closely connected with visibly pushdown
languages. For the case of graphs, we show $\reflebsp$ holds for a
very general, and again very recently defined, class of graphs called
\emph{$n$-partite cographs}, and all hereditary subclasses of this
class (Theorem~\ref{theorem:n-partite-cographs-satisfy-lebsp}). This
class of graphs, introduced in~\cite{shrub-depth}, jointly generalizes
the classes of cographs (which includes several interesting graph
classes such as complete $r$-partite graphs, Turan graphs, cluster
graphs, threshold graphs, etc.), graph classes of bounded tree-depth
and those of bounded shrub-depth. Cographs have been well studied
since the '80s~\cite{cograph-1981-paper} and have been shown to admit
fast algorithms for many decision and optimization problems that are
hard in general. Graph classes of bounded tree-depth and bounded
shrub-depth are much more recently defined~\cite{tree-depth,
  shrub-depth} and have become particularly prominent in the context
of investigating fixed parameter tractable (f.p.t.) algorithms for
$\mso$ model checking, that have \emph{elementary dependence} on the
size of the $\mso$ sentence (which is the
parameter)~\cite{shrub-depth-FO-equals-MSO, shrub-depth}. This line of
work seeks to identify classes of structures for which Courcelle-style
\emph{algorithmic meta-theorems}~\cite{model-theoretic-methods} hold,
but with better dependence on the parameter than in the case of
Courcelle's theorem (which is unavoidably
non-elementary~\cite{frick-grohe}). A different and important line of
work shows that {\fo} and {\mso} are equal in their expressive powers
over graph classes of bounded tree-depth/shrub-depth
~\cite{shrub-depth-FO-equals-MSO,tree-depth-FO-equals-MSO}. Since each
of the graph classes mentioned above is a hereditary subclass of the
class of $n$-partite cographs for some $n$, each of these satisfies
$\reflebsp$, further with computable witness functions, and further
still, even elementary witness functions in many cases.

We give methods to construct new classes of structures satisfying
$\reflebsp$ from classes known to satisfy $\reflebsp$. Specifically,
we show that $\reflebsp$ remains preserved under a wide range of
operations on structures, that have been well-studied in the
literature: unary operations like complementation, transpose and the
line graph operation, binary ``sum-like'' operations~\cite{makowsky}
such as disjoint union and join, and binary ``product-like''
operations that include various kinds of products like the Cartesian,
tensor, lexicographic and strong products. All of these are examples
of operations that can be implemented using, what are called,
\emph{quantifier-free translation
  schemes}~\cite{makowsky,model-theoretic-methods}.  We show that
$\reffebsp$ is always closed under such operations, and $\refmebsp$ is
closed under such operations, provided that they are unary or sum-like
(Theorem~\ref{theorem:lebsp-closure-under-1-step-operations}). In both
cases, the computability/elementariness of witness functions is
preserved under the operations.

\vspace{2pt} \emph{2. Linear time f.p.t. algorithms for deciding
  $\mc{L}$ properties of structures}: For each of the classes
mentioned above (including those generated using the various
operations) and for natural representations of structures in these
classes, we give linear time f.p.t. algorithms for deciding properties
of structures, that can be defined in $\mc{L}$. The structures in the
above classes have natural \emph{tree representations} in which the
leaf nodes of the tree represent simple substructures and the internal
nodes represent operations that produce new structures upon being fed
with input structures. Given such a tree representation for a
structure, we perform appropriate ``prunings'' of, and ``graftings''
within, the tree, such that the resultant tree represents an
$\mc{L}\left[m\right]$-similar proper substructure of the original
structure.  Two key technical elements that are employed to perform
these prunings and graftings are the finiteness of the index of the
$\mc{L}\left[m\right]$-similarity relation (which is an equivalence
relation) and a Feferman-Vaught kind \emph{composition property} of
the operations used in the tree representations. The latter means that
the ``$\mc{L}\left[m\right]$-similarity class'' of the structure
produced by an operation is \emph{determined} by the multi-set of the
$\mc{L}\left[m\right]$-similarity classes of the structures that are
input to the operation, and further (in the case of operations having
arbitrary finite arity), determined only by a threshold number of
appearances of each $\mc{L}\left[m\right]$-similarity class in the
multi-set, with the threshold depending solely on $m$.  These
technical features enable \emph{generating} the ``composition
functions'' uniformly for any operation for any given $m$, and the
composition functions thus generated, in turn, enable doing the
compositions in time linear in the arity of the operation. Using
these, we get linear time f.p.t. algorithms that when given an
$\mc{L}$ sentence of quantifier nesting depth $m$ (the parameter) and
a tree $\tree{t}$ as inputs, perform the aforementioned prunings and
graftings in $\tree{t}$ iteratively to produce a small subtree that
represents a small $\mc{L}\left[m\right]$-similar substructure (the
``kernel'', in the f.p.t. parlance) of the structure represented by
$\tree{t}$.  The techniques mentioned above have been incorporated int
a single abstract result concerning tree representations
(Theorem~\ref{theorem:good-tree-rep-implies-lebsp-and-f.p.t.-algorithm}). Given
that this result gives unified explanations for the good computational
properties of many interesting classes, we believe it might be of
independent interest.

\vspace{2pt}
\emph{3. A strengthening of $\reflebsp$ and connections with
  fractals}: Fractals are classes of mathematical structures that
exhibit self-similarity at all scales. That is, every structure in the
class contains a similar (in some technical sense) substructure at
every scale of sizes less than the size of the structure. Well-known
examples of fractals in mathematics include the Mandelbrot set, the
Menger Sponge and the Koch snowflake.  Remarkably, fractals are not
limited to only mathematics, but in fact abound nearly everywhere in
nature. Tree branching, cloud structures, galaxy clustering, fern
shapes, and crystal growth patterns are some of a wide range of
natural phenomena that exhibit
self-similarity~\cite{fractal-2}.

In the light of fractals, we observe that the $\reflebsp$ property
indeed asserts ``logical self-similarity'' at ``small scales''. We
formulate a strengthening of the $\reflebsp$ property, that asserts
logical self-similarity at all scales, for a suitable notion of scale
(Definition~\ref{definition:logical-fractal}). We call this the
\emph{logical fractal} property, and call a class satisfying this
property as a \emph{logical fractal}. Remarkably, it turns out that
the aforementioned posets and graph classes, including those
constructed using many of the aforementioned operations, are all
logical fractals (Proposition~\ref{prop:greatness-and-fractals}).  The
classical {\dlsfull} indeed shows that the class of all infinite
structures satisfies an ``infinitary'' variant of the logical fractal
property. We believe these observations constitute the initial
investigations into a potentially rich theory of logical fractals.

The paper is organized as follows. In
Section~\ref{section:background}, we introduce notation and
terminology, and recall relevant notions from the literature used in
the paper. In Section~\ref{section:lebsp}, we define the $\reflebsp$
property and show that it holds for the class of ``partially'' ranked
trees, which are trees in which some subset of nodes are constrained
to have degrees given by a ranking function. We use this special class
as a setting to illustrate our techniques, that we lift to tree
representations of structures in
Section~\ref{section:abstract-tree-result}. In
Section~\ref{section:classes-satisfying-lebsp}, we give applications
of our abstract results to obtain the $\reflebsp$ property and linear
time f.p.t. algorithms for model checking $\mc{L}$ sentences, in
various concrete settings, specifically those of posets and graphs
mentioned earlier, and also classes that are constructed using various
well-studied operations. We present the notion of logical fractals in
Section~\ref{section:logical-fractal}, and conclude with open
questions in Section~\ref{section:conclusion}.

\section{Terminology and preliminaries}\label{section:background}

\tbf{1. $\mc{L}$ formulae}: We assume familiarity with standard
notation and notions of first order logic (FO) and monadic second
order logic (MSO)~\cite{libkin}. By $\mc{L}$, we mean either FO or
MSO.  We consider only finite vocabularies, represented by $\tau$ or
$\nu$, that contain only \emph{predicate symbols} (and no constant or
function symbols), unless explicitly stated otherwise. All predicate
symbols are assumed to have \emph{positive} arity. We denote by
$\mc{L}(\tau)$ the set of all $\mc{L}$ formulae over $\tau$ (and refer
to these simply as $\mc{L}$ formulae, when $\tau$ is clear from
context).  A sequence $(x_1, \ldots, x_k)$ of variables is written as
$\bar{x}$.  A formula $\varphi$ whose free variables are among
$\bar{x}$, is denoted as $\varphi(\bar{x})$.  Free variables are
always \emph{first order}. A formula with no free variables is called
a \emph{sentence}.  The \emph{rank} of an $\mc{L}$ formula is the
maximum number of quantifiers (first order as well as second order)
that appears along any path from the root to the leaf in the parse
tree of the formula.  Finally, \emph{a notion or result stated for
$\mc{L}$ means that the notion or result is stated for both FO and
MSO}.

\tbf{2. Structures}: Standard notions of $\tau$-structures (denoted
$\mf{A}, \mf{B}$ etc.; we refer to these simply as structures when
$\tau$ is clear from context), substructures (denoted $\mf{A}
\subseteq \mf{B}$) and extensions are used throughout the paper (see
\cite{libkin}). We assume all structures to be \emph{finite}. As
in~\cite{libkin}, by substructures, we always mean \emph{induced}
substructures.  Given a structure $\mathfrak{A}$, we use
$\mathsf{U}_\mathfrak{A}$ to denote the universe of $\mathfrak{A}$,
and $|\mf{A}|$ to denote its cardinality.  We denote by $\mf{A} \cong
\mf{B}$ that $\mf{A}$ is isomorphic to $\mf{B}$, and by $\mathfrak{A}
\hookrightarrow \mathfrak{B}$ that $\mf{A}$ is isomorphically
embeddable in $\mf{B}$.  For an $\mc{L}$ sentence $\varphi$, we denote
by $\mathfrak{A} \models \varphi$ that $\mf{A}$ is a model of
$\varphi$.  We denote classes of structures by $\cl{S}$ possibly with
subscripts, and assume these to be \emph{closed under isomorphisms}.

\tbf{3. The $\lequiv{m}$ relation}: ~Let $\mathbb{N}$ and
$\mathbb{N}_{+}$ denote the natural numbers including zero and
excluding zero respectively. Given $m \in \mathbb{N}$ and a
$\tau$-structure $\mf{A}$, denote by $\lth{m}{\mf{A}}$ the set of all
$\mc{L}(\tau)$ sentences of rank at most $m$, that are true in
$\mf{A}$.  Given a $\tau$-structure $\mathfrak{B}$, we say that
$\mf{A}$ and $\mf{B}$ are \emph{$\mc{L}[m]$-equivalent}, denoted
$\mathfrak{A} \lequiv{m} \mathfrak{B}$ if $\lth{m}{\mf{A}} =
\lth{m}{\mf{B}}$. Given a class $\cl{S}$ of structures and $m \in
\mathbb{N}$, we let $\ldelta{\cl{S}}{m}$ denote the set of all
equivalence classes of the $\lequiv{m}$ relation over $\cl{S}$.  We
denote by $\Lambda_{\cl{S}, \mc{L}}: \mathbb{N} \rightarrow
\mathbb{N}$ a fixed computable function with the property that
$\Lambda_{\cl{S}, \mc{L}}(m) \ge |\ldelta{\cl{S}}{m}|$. It is known
that $\Lambda_{\cl{S}, \mc{L}}$ always exists (see Proposition 7.5
in~\cite{libkin}).  The notion of $\lequiv{m}$ has a characterization
using \emph{Ehrenfeucht-Fr{\"a}iss{\'e}} ($\mathsf{EF}$) games for
$\mc{L}$.  We point the reader to Chapters 3 and 7 of~\cite{libkin}
for results concerning these games.

\tbf{4. Translation schemes}: ~We recall the notion of translation
schemes from the literature~\cite{makowsky} (known in the literature
by different names, like \emph{interpretations, transductions}, etc).
Let $\tau$ and $\nu$ be given vocabularies, and $t \ge 1$ be a natural
number. Let $\bar{x}_0$ be a fixed $t$-tuple of first order variables,
and for each relation $R \in \nu$ of arity $\#R$, let $\bar{x}_R$ be a
fixed $(t \times \#R)$-tuple of first order variables.  A \emph{$(t,
  \tau, \nu, \mc{L})$-translation scheme} $\Xi = (\xi, (\xi_R)_{R \in
  \nu})$ is a sequence of formulas of $\mc{L}(\tau)$ such that the
free variables of $\xi$ are among those in $\bar{x}_0$, and for $R \in
\nu$, the free variables of $\xi_R$ are among those in $\bar{x}_R$.
When $t, \nu$ and $\tau$ are clear from context, we call $\Xi$ simply
as a translation scheme. We call $t$ as the \emph{dimension} of $\Xi$.
One can associate with a $(t, \tau, \nu, \mc{L})$-translation scheme
$\Xi$, two partial maps: (i) $\Xi^*$ from $\tau$-structures to
$\nu$-structures (ii) $\Xi^\sharp$ from $\mc{L}(\nu)$ formulae to
$\mc{L}(\tau)$ formulae. See~\cite{makowsky} for the definitions of
these.  For the ease of readability, we abuse notation slightly and
use $\Xi$ to denote both $\Xi^*$ and $\Xi^\sharp$. 

\tbf{5. Fixed parameter tractability}: ~We say that the model checking
problem for $\mc{L}$ over a given class $\cl{S}$, denoted
$\mathsf{MC}(\mc{L}, \cl{S})$, is \emph{fixed parameter tractable}, in
short \emph{f.p.t.}, if there exists an algorithm $\mathsf{Alg}$ that
when given as input an $\mc{L}$ sentence $\varphi$ of rank $m$, and a
structure $\mf{A} \in \cl{S}$, decides if $\mf{A} \models \varphi$, in
time $f(m) \cdot |\mf{A}|^c$, where $f:\mathbb{N} \rightarrow
\mathbb{N}$ is some computable function and $c$ is a constant. In this
case, we say $\mathsf{Alg}$ is an \emph{f.p.t. algorithm} for
$\mathsf{MC}(\mc{L}, \cl{S})$. We say $\mathsf{Alg}$ is a \emph{linear
  time} f.p.t. algorithm for $\mathsf{MC}(\mc{L}, \cl{S})$ if it is
f.p.t. for $\mathsf{MC}(\mc{L}, \cl{S})$ and runs in time $f(k) \cdot
|\mf{A}|$ where as before, $f$ is a computable function.

\tbf{6. Miscellaneous}: ~The \emph{$k$-fold} exponential function
$\mathsf{exp}(n, k)$ is the function given inductively as:
$\mathsf{exp}(n, 0) = n$ and $\mathsf{exp}(n, l) = 2^{\mathsf{exp}(n,
  l-1)}$ for $0 \leq l \leq k$.  We call a function $f: \mathbb{N}
\rightarrow \mathbb{N}$ as \emph{elementary} if there exists $k$ such
that $f(n) = O(\mathsf{exp}(n, k))$, and call it \emph{non-elementary}
if it is not elementary.  Finally, we use standard abbreviations of
English phrases that commonly appear in mathematical
literature. Specifically, `w.l.o.g' stands for `without loss of
generality', `iff' stands for `if and only if', and `resp.'  stands
for `respectively'.

\section[The $\mc{L}$-Equivalent Bounded Substructure Property -- $\lebsp{\cl{S}}$]{The $\mc{L}$-Equivalent Bounded Substructure Property -- $\lebsp{\cl{S}}$}\label{section:lebsp}

\begin{definition}[$\lebsp{\cl{S}}$]\label{definition:lebsp}
Let $\cl{S}$ be a class of structures and $\mc{L}$ be either $\fo$ or
$\mso$.  We say that $\cl{S}$ satisfies the \emph{$\mc{L}$-equivalent
  bounded substructure property}, abbreviated \emph{$\lebsp{\cl{S}}$
  is true} (alternatively, \emph{$\lebsp{\cl{S}}$ holds}), if there
exists a monotonic function $\lwitfn{\cl{S}}: \mathbb{N} \rightarrow
\mathbb{N}$ such that for each $m \in \mathbb{N}$ and each structure
$\mf{A}$ of $\cl{S}$, there exists a structure $\mf{B}$ such that (i)
$\mf{B} \in \cl{S}$, (ii) $\mf{B} \subseteq \mf{A}$, (iii) $|\mf{B}|
\leq \lwitfn{\cl{S}}(m)$, and (iv) $\mf{B} \lequiv{m} \mf{A}$.  The
conjunction of these four conditions is denoted as
\emph{$\lebspcond(\cl{S}, \mf{A}, \mf{B}, m, \lwitfn{\cl{S}})$}.  We
call $\lwitfn{\cl{S}}$ a \emph{witness function} of $\lebsp{\cl{S}}$.
\end{definition}

We present below two simple examples of classes satisfying $\reflebsp$. 
\vspace{3pt}\begin{enumerate}[nosep]
\item
  Let $\cl{S}$ be the class of all $\tau$-structures, where all
  predicates in $\tau$ are unary. By a simple $\fef$ game argument, we
  see that $\febsp{\cl{S}}$ holds with $\fwitfn{\cl{S}}(m) = m \cdot
  2^{|\tau|}$. In more detail: given $\mf{A} \in \cl{S}$, associate
  exactly one of $2^{|\tau|}$ colors with each element $a$ of
  $\mf{A}$, where the colour gives the valuation of all predicates of
  $\tau$ for $a$ in $\mf{A}$.  Then consider $\mf{B} \subseteq \mf{A}$
  such that for each colour $c$, if $A_c = \{a \mid a \in
  \mathsf{U}_{\mf{A}}, ~a~\text{has colour}~c~\text{in}~\mf{A}\}$,
  then $A_c \subseteq \mathsf{U}_{\mf{B}}$ if $|A_c| < m$, else $|A_c
  \cap \mathsf{U}_{\mf{B}}| = m$. It is easy to see that
  $\febspcond(\cl{S}, \mf{A}, \mf{B}, m, \fwitfn{\cl{S}})$ holds.  By
  a similar $\mef$ game argument, one can show that $\mebsp{\cl{S}}$
  holds with a witness function given by $\mwitfn{\cl{S}}(m) = m\cdot
  2^{(|\tau| + m)}$.

\item Let $\cl{S}$ be the class of disjoint unions of undirected
  paths. It is known that for any $m$, any two paths of length $\ge p
  = 3^m$ are $\fo[m]$-equivalent. Let $\mf{A} = \bigsqcup_{n \ge 0}
  i_n\cdot P_n$ where $P_n$ denotes the path of length $n$, $i_n\cdot
  P_n$ denotes the disjoint union of $i_n$ copies of $P_n$, and
  $\bigsqcup$ denotes disjoint union. For $n < p$, let $j_n$ be such
  that $j_n = i_n$ if $i_n < m$ and $j_n = m$ if $i_n \ge m$. For $n =
  p$, let $j_n = h = \sum_{r \ge p} i_r$ if $h < m$, else $j_n =
  m$. One can then see using an $\fef$ game argument that if $\mc{B} =
  \bigsqcup_{n =1}^{n = p} j_n\cdot P_n$, then $\mf{B}$ satisfies
  $\febspcond(\cl{S}, \mf{A}, \mf{B}, m, \fwitfn{\cl{S}})$ where
  $\fwitfn{\cl{S}}(m) = \sum_{n = 0}^{n = p} m \cdot n$.
  
\end{enumerate}

\subsection{Partially ranked trees satisfy $\reflebsp$}\label{section:partially-ranked-trees}

\newcommand{\troot}[1]{\ensuremath{\mathsf{root}(#1)}}
\newcommand{\sigmaint}{\ensuremath{\Sigma_{\text{int}}}}
\newcommand{\sigmaunrank}{\ensuremath{\Sigma_{\text{unrank}}}}
\newcommand{\sigmarank}{\ensuremath{\Sigma_{\text{rank}}}}
\newcommand{\sigmairrepl}{\ensuremath{\Sigma_{\text{irrepl}}}}
\newcommand{\sigmaleaf}{\ensuremath{{\Sigma}_{\text{leaf}}}}
\newcommand{\str}[1]{\ensuremath{\mathsf{Str}(#1)}}
\newcommand{\stri}[2]{\ensuremath{\mathsf{Str}_{#1}(#2)}}
\newcommand{\rftrees}{\ensuremath{\mathsf{RF}\text{-}\mathsf{trees}}}
\newcommand{\mylabel}[1]{\ensuremath{\mathsf{#1}}}
\newcommand{\disjun}{\ensuremath{\dot\cup}}
\newcommand{\liequiv}[2]{\ensuremath{\equiv_{#1, #2}}}
\newcommand{\rfse}{\ensuremath{\mc{L}\text{-RFSE}}}
\newcommand{\mrfse}{\ensuremath{\mso\text{-RFSE}}}
\newcommand{\prankedtrees}{\ensuremath{\mathsf{Partially\text{-}ranked\text{-}trees}}}
\newcommand{\fsigmai}[2]{\ensuremath{f_{\sigma, #1, #2}}}
\newcommand{\fsigma}[1]{\ensuremath{f_{\sigma, #1}}}

In this subsection, we show that the class of ordered ``partially''
ranked trees satisfies $\reflebsp$ with computable witness functions,
as well as admits a linear time f.p.t. algorithm for model checking
$\mc{L}$ sentences. This setting illustrates our reasoning and
techniques that we lift in Section~\ref{section:abstract-tree-result}
to the more abstract setting of tree representations of
structures.

An \emph{unlabeled unordered tree} is a finite poset $P = (A, \leq)$
with a unique minimal element (called ``root''), such that for each $c
\in A$, the set $\{b \mid b \leq c\}$ is totally ordered by $\leq$.
Informally speaking, the Hasse diagram of $P$ is an inverted
(graph-theoretic) tree.  We call $A$ as the set of \emph{nodes} of
$P$.  We use the standard notions of leaf, internal node, ancestor,
descendent, parent, child, degree, height, and subtree in connection
with trees. (We clarify that by height, we mean the maximum distance
between the root and any leaf of the tree, as against the ``number of
levels'' in the tree.)  An \emph{unlabeled ordered} tree is a pair $O
= (P, \lesssim)$ where $P$ is an unlabeled unordered tree and
$\lesssim$ is a binary relation that imposes a linear order on the
children of any internal node of $P$. Unless explicitly stated
otherwise, we always consider our trees to be \emph{ordered}. It is
clear that the above mentioned notions in connection with unordered
trees can be adapted for ordered trees.  Given a countable alphabet
$\Sigma$, a \emph{tree over $\Sigma$}, also called a
\emph{$\Sigma$-tree}, or simply \emph{tree} when $\Sigma$ is clear
from context, is a pair $(O, \lambda)$ where $O$ is an unlabeled tree
and $\lambda: A \rightarrow \Sigma$ is a labeling function, where $A$
is the set of nodes of $O$. We denote $\Sigma$-trees by $\tree{s},
\tree{t}, \tree{x}, \tree{y}, \tree{u}, \tree{v}$ or $\tree{z}$,
possibly with numbers as subscripts.  Given a tree $\tree{t}$, we
denote the root of $\tree{t}$ as $\troot{\tree{t}}$. For a node $a$ of
$\tree{t}$, we denote the subtree of $\tree{t}$ rooted at $a$ as
$\tree{t}_{\ge a}$, and the subtree of $\tree{t}$ obtained by deleting
$\tree{t}_{\ge a}$ from $\tree{t}$, as $\tree{t} - \tree{t}_{\ge a}$.
Given a tree $\tree{s}$ and a non-root node $a$ of $\tree{t}$, the
\emph{replacement of $\tree{t}_{\ge a}$ with $\tree{s}$ in
  $\tree{t}$}, denoted $\tree{t} \left[ \tree{t}_{\ge a} \mapsto
  \tree{s} \right]$, is a tree defined as follows. Assume
w.l.o.g. that $\tree{s}$ and $\tree{t}$ have disjoint sets of
nodes. Let $c$ be the parent of $a$ in $\tree{t}$.  Then $\tree{t}
\left[ \tree{t}_{\ge a} \mapsto \tree{s} \right]$ is defined as the
tree obtained by deleting $\tree{t}_{\ge a}$ from $\tree{t}$ to get a
tree $\tree{t}'$, and inserting (the root of) $\tree{s}$ at the same
position among the children of $c$ in $\tree{t}'$, as the position of
$a$ among the children of $c$ in $\tree{t}$.  For $\tree{s}$ and
$\tree{t}$ as just mentioned, suppose the roots of both these trees
have the same label. Then the \emph{merge of $\tree{s}$ with
  $\tree{t}$}, denoted $\tree{t} \odot \tree{s}$, is defined as the
tree obtained by deleting $\troot{\tree{s}}$ from $\tree{s}$ and
concatenating the sequence of subtrees hanging at $\troot{\tree{s}}$
in $\tree{s}$, to the sequence of subtrees hanging at
$\troot{\tree{t}}$ in $\tree{t}$. Thus the children of
$\troot{\tree{s}}$ in $\tree{s}$ are the ``new'' children of
$\troot{\tree{t}}$, and appear ``after'' the ``old'' children of
$\troot{\tree{t}}$, and in the order they appear in $\tree{s}$.

Fix a finite alphabet $\Sigma$, and let $\sigmarank \subseteq
\Sigma$. Let $\rho: \sigmarank \rightarrow \mathbb{N}_{+}$ be a fixed
function. We say a $\Sigma$-tree $\tree{t} = (O, \lambda)$ is
\emph{partially ranked by $(\sigmarank, \rho)$} if for any node $a$ of
$\tree{t}$, if $\lambda(a) \in \sigmarank$, then the number of
children of $a$ in $\tree{t}$ is exactly $\rho(\lambda(a))$. Observe
that the case of $\sigmarank = \Sigma$ corresponds to the notion of
ranked trees that are well-studied in the literature~\cite{tata}. Let
$\prankedtrees(\Sigma, \sigmarank, \rho)$ be the class of all ordered
$\Sigma$-trees partially ranked by $(\sigmarank, \rho)$. The central
result of this section is now as stated below.

\vspace{5pt}
\begin{prop}\label{prop:result-for-partially-ranked-trees}
Given $\Sigma$, a subset $\sigmarank$ of $\Sigma$ and $\rho:
\sigmarank \rightarrow \mathbb{N}_+$, let $\cl{S}$ be the class
$\prankedtrees(\Sigma, \sigmarank, \rho)$. Then the following are
true:
\vspace{3pt}\begin{enumerate}[nosep]
\item $\lebsp{\cl{S}}$ holds with a computable witness
  function. Further, any witness function is necessarily
  non-elementary.
\item There is a linear time f.p.t. algorithm for $\mathsf{MC}(\mc{L},
  \cl{S})$.
\end{enumerate}
\end{prop}

We prove the two parts of the above result separately. In the
remainder of this section, we fix $\mc{L}$, and also fix $\cl{S}$ to
be the class $\prankedtrees(\Sigma, \sigmarank, \rho)$. Given these
fixings, we denote $\ldelta{\cl{S}}{m}$ (the set of equivalence
classes of the $\lequiv{m}$ relation over $\cl{S}$) simply as
$\Delta_m$, and denote $\Lambda_{\cl{S}, \cl{L}}(m)$ (see point 3 in
Section~\ref{section:background} for the definition of
$\Lambda_{\cl{S}, \cl{L}}(m)$) simply as $\Lambda(m)$. All trees will
be assumed to be from $\cl{S}$.

Towards the proof of
Proposition~\ref{prop:result-for-partially-ranked-trees}, we first
present a Feferman-Vaught style $\mc{L}$-composition lemma for ordered
trees. Composition results of this kind were first studied by Feferman
and Vaught, and subsequently by many others (see~\cite{makowsky}).  To
state the composition lemma, we introduce some terminology. For a
finite alphabet $\Omega$, given ordered $\Omega$-trees $\tree{t},
\tree{s}$ having disjoint sets of nodes (w.l.o.g.) and a non-root node
$a$ of $\tree{t}$, the \emph{join of $\tree{s}$ to $\tree{t}$ to the
  right of $a$}, denoted $\tree{t} \cdot^{\rightarrow}_a \tree{s}$, is
defined as the tree obtained by making $\tree{s}$ as a new child
subtree of the parent of $a$ in $\tree{t}$, at the successor position
of the position of $a$ among the children of the parent of $a$ in
$\tree{t}$.  We can similarly define the \emph{join of $\tree{s}$ to
  $\tree{t}$ to the left of $a$}, denoted $\tree{t}
\cdot^{\leftarrow}_a \tree{s}$. Likewise, for $\tree{t}$ and
$\tree{s}$ as above, if $a$ is a leaf node of $\tree{t}$, we can
define the \emph{join of $\tree{s}$ to $\tree{t}$ below $a$}, denoted
$\tree{t} \cdot^{\uparrow}_a \tree{s}$, as the tree obtained upto
isomorphism by making the root of $\tree{s}$ as a child of $a$. The
$\mc{L}$ composition lemma for ordered trees can now be stated as
follows. The proof is similar to the proof of the known
$\mc{L}$-composition lemma for words.  We skip presenting the proof
here, but point the interested reader to
Appendix~\ref{section:appendix:composition-lemma-for-partially-ranked-trees}
for the detailed proof.

\begin{lemma}[Composition lemma for ordered trees]\label{lemma:mso-composition-lemma-for-ordered-trees}
For a finite alphabet $\Omega$, let ${\tree{t}}_i, \tree{s}_i$ be
non-empty ordered $\Omega$-trees, and let $a_i$ be a non-root node of
$\tree{t}_i$, for each $i \in \{1, 2\}$. Let $m \ge 2$ and suppose
that $({\tree{t}}_1, a_1) \lequiv{m} ({\tree{t}}_2, a_2)$ and
${\tree{s}}_1
\lequiv{m} {\tree{s}}_2$. Then each of the following hold.
\begin{enumerate}[nosep]
\item
$(({\tree{t}}_1 \cdot^{\rightarrow}_{a_1} {\tree{s}}_1), a_1) \lequiv{m}
(({\tree{t}}_2 \cdot^{\rightarrow}_{a_2} {\tree{s}}_2), a_2)$\label{lemma:mso-composition-lemma-for-ordered-trees:part-1}
\item
$(({\tree{t}}_1 \cdot^{\leftarrow}_{a_1} {\tree{s}}_1), a_1) \lequiv{m}
(({\tree{t}}_2 \cdot^{\leftarrow}_{a_2} {\tree{s}}_2), a_2)$
\item
$(({\tree{t}}_1 \cdot^{\uparrow}_{a_1} {\tree{s}}_1), a_1) \lequiv{m}
  (({\tree{t}}_2 \cdot^{\uparrow}_{a_2} {\tree{s}}_2), a_2)$ if $a_1,
  a_2$ are leaf nodes of $\tree{t}_1, \tree{t}_2$ resp.
\end{enumerate}
\end{lemma}

A useful corollary of this lemma is as below.
\begin{corollary}\label{corollary:tree-composition}
  The following are true for $m \ge 3$.
  \vspace{3pt}\begin{enumerate}[nosep]
    \item Given trees $\tree{s}, \tree{t}$ and a non-root node $a$ of
      $\tree{t}$, let $\tree{z} = \tree{t}[\tree{t}_{\ge a} \mapsto
      \tree{s}]$. If $\tree{s} \lequiv{m} \tree{t}_{\ge a}$, then
      $\tree{z} \lequiv{m} \tree{t}$.\label{corollary:tree-composition:1}
    \item Let $\tree{s}_1, \tree{s}_2, \tree{t}$ be given trees such
      that the labels of their roots are the same, and belong to
      $\Sigma \setminus \sigmarank$.  Suppose $\tree{z}_i = \tree{s}_i
      \odot \tree{t}$ for $i \in \{1, 2\}$. If $\tree{s}_1 \lequiv{m}
      \tree{s}_2$, then $\tree{z}_1 \lequiv{m}
      \tree{z}_2$.\label{corollary:tree-composition:2}
    \item Let $\tree{s}_1, \tree{s}_2$ be given trees such that the
      labels of their roots are the same, and belong to $\Sigma
      \setminus \sigmarank$. For $i \in \{1, 2\}$, given $\tree{t}_i$,
      let $\tree{z}_i$ be the tree obtained from $\tree{s}_i$ by
      adding $\tree{t}_i$ as the (new) ``last'' child subtree of the
      root of $\tree{s}_i$. If $\tree{s}_1 \lequiv{m} \tree{s}_2$ and
      $\tree{t}_1 \lequiv{m} \tree{t}_2$, then $\tree{z}_1 \lequiv{m}
      \tree{z}_2$.\label{corollary:tree-composition:3}
  \end{enumerate}
\end{corollary}
\begin{proof}
  (\ref{corollary:tree-composition:1}): Let $\tree{v} = \tree{t} -
  \tree{t}_{\ge a}$. There are 3 possibilities:
  \begin{enumerate}[nosep]
    \item[(i)] The node $a$ has a ``predecessor'' sibling in
      $\tree{t}$, call it $b$. Then $\tree{t} = \tree{v}
      \cdot^{\rightarrow}_{b} \tree{t}_{\ge a}$. Then since
      $\tree{t}_{\ge a} \lequiv{m} \tree{s}$, we have by
      Lemma~\ref{lemma:mso-composition-lemma-for-ordered-trees}, that
      $\tree{z} \lequiv{m} \tree{t}$ since $\tree{z} = (\tree{v}
      \cdot^{\rightarrow}_{b} \tree{s})$.
    \item[(ii)] The node $a$ has a ``successor'' sibling in $\tree{t}$,
      call it $b$. Then $\tree{t} = \tree{v} \cdot^{\leftarrow}_{b}
      \tree{t}_{\ge a}$. Again since $\tree{t}_{\ge a} \lequiv{m}
      \tree{s}$, we have by
      Lemma~\ref{lemma:mso-composition-lemma-for-ordered-trees}, that
      $\tree{z} \lequiv{m} \tree{t}$ since $\tree{z} = (\tree{v}
      \cdot^{\leftarrow}_{b} \tree{s})$.
    \item[(iii)] The node $a$ is the sole child of its parent $b$ in
      $\tree{t}$. Then $\tree{t} = \tree{v} \cdot^{\uparrow}_{b}
      \tree{t}_{\ge a}$. Then again by
      Lemma~\ref{lemma:mso-composition-lemma-for-ordered-trees}, we
      have that $\tree{z} \lequiv{m} \tree{t}$ since $\tree{z} =
      (\tree{v} \cdot^{\uparrow}_{b} \tree{s})$.
  \end{enumerate}
  
  (\ref{corollary:tree-composition:2}): We prove this part assuming
  part \ref{corollary:tree-composition:3}. For $i \in \{1, 2\}$, let
  $a_i$ be the last child of the root of $\tree{s}_i$ (under the
  linear order on the children of the root). Let $b_1, \ldots, b_n$ be
  the children (and in that order) of the root of $\tree{t}$. Let
  $\tree{u}_j = \tree{t}_{\ge b_j}$ for $j \in \{1, \ldots, n\}$. For
  $i \in \{1, 2\}$, let $\tree{x}^i_1 = \tree{s}_i
  \cdot^{\rightarrow}_{a_i} \tree{u}_1$ and $\tree{x}^i_{j+1} =
  \tree{x}^i_j \cdot^{\rightarrow}_{b_j} \tree{u}_{j+1}$ for $j \in
  \{1, \ldots, n-1\}$. Since $\tree{s}_1 \lequiv{m} \tree{s}_2$, we
  have by part \ref{corollary:tree-composition:3} of this lemma, that
  $\tree{x}^1_1 \lequiv{m} \tree{x}^2_1$. Whereby, $\tree{x}^1_j
  \lequiv{m} \tree{x}^2_j$ for $j \in \{1, \ldots, n\}$. Since
  $\tree{x}^i_n = \tree{z}_i$ for $i \in \{1, 2\}$, we have
  $\tree{z}_1 \lequiv{m} \tree{z}_2$.

  (\ref{corollary:tree-composition:3}): For $i \in \{1, 2\}$, let
  $a_i$ be the last child of the root of $\tree{s}_i$ (under the
  linear order on the children of the root). It is easy to verify
  given that $\tree{s}_1 \lequiv{m} \tree{s}_2$ and $m \ge 3$, that
  there exists a winning strategy for the duplicator in the $m$ round
  $\lef$ game between $\tree{s}_1$ and $\tree{s}_2$ such that in any
  round, if the spoiler chooses $a_1$ from $\tree{s}_1$ (resp. $a_2$
  from $\tree{s}_2$), then the duplicator chooses $a_2$ from
  $\tree{s}_2$ (resp. $a_1$ from $\tree{s}_1$) according to the
  winning strategy. Whereby, $(\tree{s}_1, a_1) \lequiv{m}
  (\tree{s}_2, a_2)$. Then by
  Lemma~\ref{lemma:mso-composition-lemma-for-ordered-trees},
  $\tree{z}_1 = (\tree{s}_1 \cdot^{\rightarrow}_{a_1} \tree{t}_1)
  \lequiv{m} (\tree{s}_2 \cdot^{\rightarrow}_{a_2} \tree{t}_2) =
  \tree{z}_2$.

\end{proof}

We use the above results to obtain a ``functional'' form of a
composition lemma for partially ranked trees, as given by the lemma
below. This lemma plays a crucial role in the proof of
Proposition~\ref{prop:result-for-partially-ranked-trees}. Recall that
$\cl{S} = \prankedtrees(\Sigma, \sigmarank, \rho)$.

\begin{lemma}[Composition lemma for partially ranked trees]\label{lemma:composition-lemma-for-partially-ranked-trees}
For each $\sigma \in \Sigma$ and $m \ge 3$, there exists a function
$\fsigma{m}: (\Delta_m)^{\rho(\sigma)} \rightarrow \Delta_m$ if
$\sigma \in \sigmarank$, and functions $\fsigmai{m}{i}: (\Delta_m)^i
\rightarrow \Delta_m$ for $i \in \{1, 2\}$ if $\sigma \in \Sigma
\setminus \sigmarank$, with the following properties: Let $\tree{t} =
(O, \lambda) \in \cl{S}$ and $a$ be an internal node of $\tree{t}$
such that $\lambda(a) = \sigma$, and the children of $a$ in $\tree{t}$
are $b_1, \ldots, b_n$. Let $\delta_i$ be the $\lequiv{m}$ class of
$\tree{t}_{\ge b_i}$ for $i \in \{1, \ldots, n\}$, and let $\delta$ be
the $\lequiv{m}$ class of $\tree{t}_{\ge a}$.
\vspace{3pt}\begin{enumerate}[nosep]
  \item If $\sigma \in \sigmarank$ (whereby $n = \rho(\sigma)$), then
    $\delta = \fsigma{m}(\delta_1, \ldots, \delta_n)$.
  \item If $\sigma \in \Sigma \setminus \sigmarank$, then $\delta$ is
    given as follows: For $k \in \{1, \ldots, n-1\}$, let $\chi_{k+1}
    = \fsigmai{m}{2}(\chi_k, \delta_{k+1})$ where $\chi_1 =
    \fsigmai{m}{1}(\delta_1)$. Then $\delta = \chi_n$.
\end{enumerate}
\end{lemma}
\begin{proof}
  We define functions $\fsigma{m}$ and $\fsigmai{m}{i}$ as follows:

  \begin{enumerate}[nosep]
    \item $\fsigma{m}$: Let $\delta_i \in \Delta_m$ for $i \in \{1,
      \ldots, n\}$ be given, where $n = \rho(\sigma)$ and $\sigma \in
      \sigmarank$. If any of the $\delta_i$'s is not realized in
      $\cl{S}$ (i.e. there is no tree in $\cl{S}$ whose $\lequiv{m}$
      class is $\delta_i$), then define $\fsigma{m}(\delta_1, \ldots,
      \delta_n) = \delta_{\text{default}}$ where
      $\delta_{\text{default}}$ is some fixed element of $\Delta_m$.

      Else, let $\tree{t}_i \in \cl{S}$ be a tree such that the
      $\lequiv{m}$ class of $\tree{t}_i$ is $\delta_i$ for $1 \leq i
      \leq n$. Let $\tree{s}_{\delta_1, \ldots, \delta_n}$ be the tree
      obtained by making $\tree{t}_1, \ldots, \tree{t}_n$ as the child
      subtrees (and in that sequence) of a new root node labeled with
      $\sigma$. Let $\delta$ be the $\lequiv{m}$ class of
      $\tree{s}_{\delta_1, \ldots, \delta_n}$. Define
      $\fsigma{m}(\delta_1, \ldots, \delta_n) = \delta$.
    \item $\fsigmai{m}{i}$: The case when $i = 1$ can be done similarly
      as above. We consider the case of $i = 2$. Let $\delta_1,
      \delta_2 \in \Delta_m$. Note that $\sigma \in \Sigma \setminus
      \sigmarank$. For $i \in \{1, 2\}$, find trees $\tree{t}_i$ such
      that the $\lequiv{m}$ class of $\tree{t}_i$ is $\delta_i$ and
      further such that the root of $\tree{t}_1$ is labeled with
      $\sigma$. If either $\tree{t}_1$ or $\tree{t}_2$ is not found,
      then define $\fsigmai{m}{2}(\delta_1, \delta_2) =
      \delta_{\text{default}}$. Else, let $\tree{v}_{\delta_1,
        \delta_2}$ be the tree obtained adding $\tree{t}_2$ as the
      (new) ``last'' child subtree of the root of $\tree{t}_1$. Let
      $\delta$ be the $\lequiv{m}$ class of $\tree{v}_{\delta_1,
        \delta_2}$. Define $\fsigmai{m}{2}(\delta_1, \delta_2) =
      \delta$.
  \end{enumerate}

   We claim that $\fsigma{m}$ and $\fsigmai{m}{i}$ indeed satisfy the
   properties mentioned in the statement of this lemma. Let $\tree{t}
   = (O, \lambda) \in \cl{S}$ and $a$ be an internal node of
   $\tree{t}$ such that $\lambda(a) = \sigma$, and the children of $a$
   in $\tree{t}$ are $b_1, \ldots, b_n$. Let $\delta_i$ be the
   $\lequiv{m}$ class of $\tree{t}_{\ge b_i}$ for $i \in \{1, \ldots,
   n\}$, and let $\delta$ be the $\lequiv{m}$ class of $\tree{t}_{\ge
     a}$.
   \begin{itemize}[nosep]
     \item $\fsigma{m}$: Since $\tree{t}_{\ge b_i}$ has $\lequiv{m}$
       class $\delta_i$ for $i \in \{1, \ldots, n\}$, we see that the
       tree $\tree{z} = \tree{s}_{\delta_1, \ldots, \delta_n}$, as
       referred to earlier, exists. Let $d_1, \ldots, d_n$ be the
       children of the root of $\tree{z}$; then for $i \in \{1,
       \ldots, n\}$, the $\lequiv{m}$ class of $\tree{z}_{\ge d_i}$ is
       $\delta_i$, and hence $\tree{z}_{\ge d_i} \lequiv{m}
       \tree{t}_{\ge b_i}$. Since $\tree{t}_{\ge a} =
       \tree{z}\left[\tree{z}_{\ge d_1} \mapsto \tree{t}_{\ge
           b_1}\right] \cdots \left[\tree{z}_{\ge d_n} \mapsto
         \tree{t}_{\ge b_n}\right]$, we see by
       Corollary~\ref{corollary:tree-composition}(\ref{corollary:tree-composition:1})
       that $\tree{t}_{\ge a} \lequiv{m} \tree{z}$, whereby the
       $\lequiv{m}$ class of $\tree{t}_{\ge a}$ equals the
       $\lequiv{m}$ class of $\tree{z}$. The latter in turn is the
       same as $\fsigma{m}(\delta_1, \ldots, \delta_n)$ by
       construction.

     \item $\fsigmai{m}{i}$: The reasoning for $i = 1$ is just as done
       above for $\fsigma{m}$. We hence consider the case of $i = 2$.
       We illustrate our reasoning for the example of $n = 3$. The
       reasoning for general $n$ can be done likewise. Let $\tree{u} =
       \tree{t}_{\ge a}$; the root of $\tree{u}$ has 3 children $b_1,
       b_2, b_3$ such that the $\lequiv{m}$ class of $b_i$ is
       $\delta_i$ for $i \in \{1, 2, 3\}$. Consider the subtrees
       $\tree{x}$ and $\tree{y}$ of $\tree{u}$ defined as $\tree{x} =
       \tree{u} - \tree{u}_{\ge b_3}$ and $\tree{y} = \tree{x} -
       \tree{x}_{\ge b_2}$. Let $\delta_4$ and $\delta_5$ be resp. the
       $\lequiv{m}$ classes of $\tree{x}$ and $\tree{y}$. Now consider
       the trees $\tree{v}_{\delta_5, \delta_2}$ and
       $\tree{v}_{\delta_4, \delta_3}$ which are guaranteed to exist
       ($\tree{v}_{\delta_5, \delta_2}$ exists since $\tree{y}$
       \emph{is} a tree whose root is labeled with $\sigma$ and whose
       $\lequiv{m}$ class is $\delta_5$, while $\tree{u}_{\ge b_2}$
       \emph{is} a tree whose $\lequiv{m}$ class is $\delta_2$). Since
       $\sigma \in \Sigma \setminus \sigmarank$, we have by
       Corollary~\ref{corollary:tree-composition}(\ref{corollary:tree-composition:3}),
       that $\tree{x} \lequiv{m} \tree{v}_{\delta_5, \delta_2}$ and
       $\tree{u} \lequiv{m} \tree{v}_{\delta_4, \delta_3}$. Whereby,
       the $\lequiv{m}$ class of $\tree{x}$ is $\delta_4 =
       \fsigmai{m}{2}(\delta_5, \delta_2)$ and that of $\tree{u}$ is
       $\delta = \fsigmai{m}{2}(\delta_4, \delta_3)$. Observe that
       $\delta_5$ is indeed $\fsigmai{m}{1}(\delta_1)$.

   \end{itemize}
   
\end{proof}

\tbf{Proof of part (1) of
  Proposition~\ref{prop:result-for-partially-ranked-trees}}: The proof
of this part has at its core, the following ``reduction'' lemma that
shows that the degree and height of a tree can always be reduced to
under a threshold, preserving $\mc{L}[m]$ equivalence.

  \begin{lemma}\label{lemma:degree-and-height-reduction-for-trees}
    There exist computable functions $\eta_1, \eta_2: \mathbb{N}
    \rightarrow \mathbb{N}$ such that for each $\tree{t} \in \cl{S}$
    and $m \in \mathbb{N}$, the following hold:
    \vspace{3pt}\begin{enumerate}[nosep]
    \item \emph{(Degree reduction)} There exists a subtree
      $\tree{s}_1$ of $\tree{t}$ in $\cl{S}$, of degree $\leq
      \eta_1(m)$, such that (i) the roots of $\tree{s}_1$ and $\tree{t}$
      are the same, and (ii) $\tree{s}_1 \lequiv{m}
      \tree{t}$.\label{lemma:degree-reduction-for-trees}
    \item \emph{(Height reduction)} There exists a subtree
      $\tree{s}_2$ of $\tree{t}$ in $\cl{S}$, of height $\leq
      \eta_2(m)$, such that (i) the roots of $\tree{s}_2$ and $\tree{t}$
      are the same, and (ii) $\tree{s}_2 \lequiv{m}
      \tree{t}$.\label{lemma:height-reduction-for-trees}
    \end{enumerate}
  \end{lemma}

  \begin{proof}[Proof sketch]
    For a finite subset $X$ of $\mathbb{N}$, let $\text{max}(X)$
    denote the maximum element of $X$.

    \vspace{3pt}(\ref{lemma:degree-reduction-for-trees}): For $n \ge
    3$, define $\eta_1(n) = \text{max}(\{\rho(\sigma) \mid \sigma \in
    \sigmarank\} \cup \{3\}) \times \Lambda(n)$. For $n < 3$, define
    $\eta_1(n) = \eta_1(3)$. We prove this part for $m \ge 3$; then it
    follows that this part is also true for $m < 3$ (by taking
    $\tree{s}_1$ for the $m = 3$ case as $\tree{s}_1$ for the $m < 3$
    case).

    Given $m \ge 3$, let $p = \eta_1(m)$. If $\tree{t}$ has degree
    $\leq p$, then putting $\tree{s}_1 = \tree{t}$, we are done.
    Else, some node $a$ of $\tree{t}$ has degree $n > p$. Clearly then
    $\lambda(a) \notin \sigmarank$. Let $\tree{z} = \tree{t}_{\ge a}$
    and let $a_1, \ldots, a_n$ be the (ascending) sequence of children
    of $\troot{\tree{z}}$ in $\tree{z}$. For $1 \leq j \leq n$, let
    $\tree{x}_{1, j}$, resp. $\tree{y}_{j+1, n}$, be the subtree of
    $\tree{z}$ obtained from $\tree{z}$ by deleting the subtrees
    rooted at $a_{j+1}, \ldots, a_n$, resp.  deleting the subtrees
    rooted at $a_1, a_2, \ldots, a_{j}$. Then $\tree{z} = \tree{x}_{1,
      n} = \tree{x}_{1, j} \odot \tree{y}_{j+1, n}$ for $1 \leq j <
    n$.  Let $g: \{1, \ldots, n\} \rightarrow \Delta_m$ be such that
    $g(j)$ is the $\lequiv{m}$ class of $\tree{x}_{1, j}$. Since $n >
    p$, there exist $j, k \in \{1, \ldots, n\}$ such that $j < k$ and
    $g(j) = g(k)$, i.e.  $\tree{x}_{1, j} \lequiv{m} \tree{x}_{1, k}$.
    If $k < n$, then let $\tree{z}_1 = \tree{x}_{1, j} \odot
    \tree{y}_{k+1, n}$, else let $\tree{z}_1 = \tree{x}_{1, j}$. Then
    by Corollary~\ref{corollary:tree-composition}, $\tree{z}_1
    \lequiv{m} \tree{z}$.  Let $\tree{t}_1$ be the subtree of
    $\tree{t}$ in $\cl{S}$ given by $\tree{t}_1 = \tree{t}
    \left[\tree{z} \mapsto \tree{z}_1 \right]$. By
    Corollary~\ref{corollary:tree-composition} again, $\tree{t}_1
    \lequiv{m} \tree{t}$. Observe that $\tree{t}_1$ has strictly
    lesser size than $\tree{t}$. Recursing on $\tree{t}_1$, we are
    eventually done.

    \vspace{3pt}(\ref{lemma:height-reduction-for-trees}): For $n \ge
    3$, define $\eta_2(n) = \Lambda(n) + 1$. For $n < 3$, define
    $\eta_2(n) = \eta_2(3)$. As before, it suffices to prove this part
    for $m \ge 3$.

    Given $m \ge 3$, let $p = \eta_2(m)$.  If $\tree{t}$ has height
    $\leq p$, then putting $\tree{s}_2 = \tree{t}$, we are done. Else,
    there is a path from the root of $\tree{t}$ to some leaf of
    $\tree{t}$, whose length is $> p$. Let $A$ be the set of nodes
    appearing along this path. Let $h : A \rightarrow \Delta_m$ be
    such that for each $a \in A$, $h(a)$ is the $\lequiv{m}$ class of
    $\tree{t}_{\ge a}$. Since $|A| > p$, there exist distinct nodes
    $a, b \in A$ such that $a$ is an ancestor of $b$ in $\tree{t}$, $a
    \neq \troot{\tree{t}}$, and $h(a) = h(b)$. Let $\tree{t}_2 =
    \tree{t}\left[\tree{t}_{\ge a} \mapsto \tree{t}_{\ge b} \right]$;
    then $\tree{t}_2$ is a subtree of $\tree{t}$ in $\cl{S}$. Since
    $h(a) = h(b)$, $\tree{t}_{\ge a} \lequiv{m} \tree{t}_{\ge b}$. By
    Corollary~\ref{corollary:tree-composition}, we get $\tree{t}_2
    \lequiv{m} \tree{t}$.  Note that $\tree{t}_2$ has strictly lesser
    size than $\tree{t}$. Recursing on $\tree{t}_2$, we are eventually
    done.
  \end{proof}

  \begin{proof}[Proof of  Proposition~\ref{prop:result-for-partially-ranked-trees}(1)]
    Let $\tree{t} \in \cl{S}$ and $m \in \mathbb{N}$ be given. By
    Lemma~\ref{lemma:degree-and-height-reduction-for-trees}, there
    exists a subtree $\tree{s}$ of $\tree{t}$ in $\cl{S}$, of degree
    $\leq \eta_1(m)$ and height $\leq \eta_2(m)$, and hence of size
    $\leq \eta_1(m)^{(\eta_2(m)+ 1)}$, such that $\tree{s} \lequiv{m}
    \tree{t}$. Then $\lebspcond(\cl{S}, \tree{t}, \tree{s}, m,
    \lwitfn{\cl{S}})$ is true where $\lwitfn{\cl{S}}(m) =
    \eta_1(m)^{(\eta_2(m)+ 1)}$. Since $\tree{t} \in \cl{S}$ and $m
    \in \mathbb{N}$ are arbitrary, it follows that $\lebsp{\cl{S}}$ is
    true.

    As for the non-elementariness of witness functions for
    $\lebsp{\cl{S}}$, observe that if there exists an elementary
    witness function $\theta$ for $\lebsp{\cl{S}}$, then every tree
    $\tree{t}$ in $\cl{S}$ is $\mc{L}[m]$-equivalent to a tree
    $\tree{s}$ in $\cl{S}$ such that $|\tree{s}| \leq
    \theta(m)$. Whereby the index of the $\lequiv{m}$ relation over
    $\cl{S}$ is bounded by the number of trees in $\cl{S}$ whose size
    is $\leq \theta(m)$. Clearly then, this number, and hence the
    index, is bounded by an elementary function of $m$ if $\theta$ is
    elementary. However, even over words, we know that the index of
    the $\lequiv{m}$ relation is non-elementary~\cite{frick-grohe}.
  \end{proof}
  
  \newcommand{\enumcompfun}{\ensuremath{\mathsf{Generate\text{-}functions}}}
\newcommand{\gsigmai}[2]{\ensuremath{g_{\sigma, #1, #2}}}
\newcommand{\gsigma}[1]{\ensuremath{g_{\sigma, #1}}}

\newcommand{\reducedeg}{\ensuremath{\mathsf{Reduce}\text{-}\mathsf{degree}}}
\newcommand{\redheight}{\ensuremath{\mathsf{Reduce}\text{-}\mathsf{height}}}
\newcommand{\rainbow}{\ensuremath{\mathsf{Rainbow\text{-}subtree}}}
\newcommand{\lowest}{\ensuremath{\mathsf{Lowest\text{-}subtree}}}

\newcommand{\colour}{\ensuremath{\mathsf{Colour}}}
\newcommand{\degred}{\ensuremath{\mathsf{Reduce\text{-}degree\text{-}of\text{-}node}}}
\newcommand{\listofmclasses}{\ensuremath{\mathsf{\mc{L}[m]\text{-}\mathsf{classes}}}}

\vspace{3pt}\tbf{Proof of part (2) of
  Proposition~\ref{prop:result-for-partially-ranked-trees}}: The
following result contains the core argument for the proof of this part
of Proposition~\ref{prop:result-for-partially-ranked-trees}. The first
part of Lemma~\ref{lemma:core-lemma-for-fpt-over-trees} gives an
algorithm to generate the ``composition'' functions of
Lemma~\ref{corollary:tree-composition}, uniformly for $m \ge 3$. This
algorithm is in turn used in the second part of
Lemma~\ref{lemma:core-lemma-for-fpt-over-trees} to get a ``linear
time'' version of
Lemma~\ref{lemma:degree-and-height-reduction-for-trees}.

\begin{lemma}\label{lemma:core-lemma-for-fpt-over-trees}
There exist computable functions $\eta_3, \eta_4, \eta_5: \mathbb{N}
\rightarrow \mathbb{N}$ and algorithms\linebreak $\enumcompfun(m)$,
$\reducedeg(\tree{t}, m)$ and $\redheight(\tree{t}, m)$ such that for
$m \ge 3$,

\begin{enumerate}[nosep]
\item \vspace{2pt}$\enumcompfun(m)$ generates in time $\eta_3(m)$, the
  functions $\fsigma{m}$ if $\sigma \in \sigmarank$ and
  $\fsigmai{m}{i}$ for $i \in \{1, 2\}$ if $\sigma \in \Sigma
  \setminus \sigmarank$, that satisfy the properties mentioned in
  Lemma~\ref{lemma:composition-lemma-for-partially-ranked-trees}.
  \label{lemma:core-lemma-generate-functions}

\item For $\tree{t} \in \cl{S}$, $\reducedeg(\tree{t}, m)$ computes
  the subtree $\tree{s}_1$ of $\tree{t}$ as given by
  Lemma~\ref{lemma:degree-and-height-reduction-for-trees}, in time
  $\eta_4(m) \cdot |\tree{t}|$. Likewise, $\redheight(\tree{t}, m)$
  computes the subtree $\tree{s}_2$ of $\tree{t}$ as given by
  Lemma~\ref{lemma:degree-and-height-reduction-for-trees}, in time
  $\eta_5(m) \cdot |\tree{t}|$.
  \label{lemma:core-lemma-reduce-degree-and-height}
\end{enumerate}

\end{lemma}

Using this lemma, part (2) of
Proposition~\ref{prop:result-for-partially-ranked-trees} can be proved
as follows.

\begin{proof}[Proof of Proposition~\ref{prop:result-for-partially-ranked-trees}(2)]
We describe a simple algorithm $\mathsf{Evaluate}(\tree{t}, \varphi)$
that when given a tree $\tree{t} \in \cl{S}$ and an $\mc{L}$ sentence
$\varphi$ of rank $m$, as inputs, decides if $\tree{t} \models
\varphi$ in time $f(m) \cdot |\tree{t}|$ for some computable function
$f:\mathbb{N} \rightarrow \mathbb{N}$.

\vspace{3pt}
\underline{$\mathsf{Evaluate}(\tree{t}, \varphi)$:}

\begin{enumerate}
\item Let $m_1 = \text{max}\{m, 3\}$.
\item Compute a subtree $\tree{s}$ of $\tree{t}$ in $\cl{S}$ by
  invoking $\redheight(\reducedeg(\tree{t}, m_1), m_1)$.
\item Evaluate $\varphi$ on $\tree{s}$.
\item If $\tree{s} \models \varphi$, return $\mathsf{True}$, else
  return $\mathsf{False}$.
\end{enumerate}

Analysis:

\begin{itemize}
  \item Correctness: For functions $\eta_1, \eta_2$ as mentioned in
    Lemma~\ref{lemma:degree-and-height-reduction-for-trees}, the
    subtree $\tree{s}$ in the algorithm above is such that $|\tree{s}|
    \leq \eta_1(m_1)^{(\eta_2(m_1) + 1)}$ and $\tree{s} \lequiv{m_1}
    \tree{t}$ -- this follows from
    Lemma~\ref{lemma:core-lemma-for-fpt-over-trees}(\ref{lemma:core-lemma-reduce-degree-and-height}). Since
    $m_1 \ge m$, we have $\tree{s} \lequiv{m} \tree{t}$; then
    $\tree{t} \models \varphi$ iff $\tree{s} \models \varphi$, proving
    that the above algorithm is indeed correct.

  \item Running time: By
    Lemma~\ref{lemma:core-lemma-for-fpt-over-trees}(\ref{lemma:core-lemma-reduce-degree-and-height}),
    the time taken for computing $\tree{s}$ is at most $\eta_4(m_1)
    \cdot |\tree{t}| + \eta_5(m_1) \cdot |\tree{t}|$. The time taken
    to evaluate $\varphi$ on $\tree{s}$ is $\eta_6(m_1)$ for some
    computable function $\eta_6:\mathbb{N} \rightarrow
    \mathbb{N}$. Then the total running time of
    $\mathsf{Evaluate}(\tree{t}, \varphi)$ is at most $f(m) \cdot
    |\tree{t}|$, where $f(m) = \eta_4(m_1) + \eta_5(m_1) +
    \eta_6(m_1)$ and $m_1 = \text{max}\{m, 3\}$.
\end{itemize}
\end{proof}

We now provide a proof sketch for
Lemma~\ref{lemma:core-lemma-for-fpt-over-trees} to complete this
section.

\begin{proof}[Proof sketch for Lemma~\ref{lemma:core-lemma-for-fpt-over-trees}] (Part \ref{lemma:core-lemma-generate-functions}): 
  For the algorithm, we observe that the $\mc{L}$-$\mathsf{SAT}$
  problem is decidable over $\cl{S}$ -- since $\lebsp{\cl{S}}$ holds
  with a computable witness function (by
  Proposition~\ref{prop:result-for-partially-ranked-trees}(1)), if an
  $\mc{L}$ sentence has a model in $\cl{S}$, it also has a model of
  size bounded by a computable function of its rank.

  \vspace{3pt}\und{$\enumcompfun(m)$}:
  
  \vspace{2pt}
  \begin{enumerate}
  \item Create a list $\listofmclasses$ of the $\lequiv{m}$ classes
    over $\cl{S}$. This is done as follows:
    \begin{enumerate}
      \item Given the inductive definition of $\mc{L}[m]$, there is an
        algorithm $\mc{P}(m)$ which enumerates $\mc{L}[m]$ sentences
        $\varphi_1, \varphi_2, \ldots, \varphi_n$ such that every
        sentence $\varphi_i$ captures some equivalence class of the
        $\lequiv{m}$ relation over all finite structures, and
        conversely, every equivalence class of the $\lequiv{m}$
        relation over all finite structures, is captured by some
        $\varphi_i$. First invoke $\mc{P}(m)$ to get the $\varphi_i$s.
      \item For each $i \in \{1, \ldots, n\}$, if $\varphi_i$ is
        satisfiable over $\cl{S}$ (whereby it represents some
        equivalence class of the $\lequiv{m}$ relation over $\cl{S})$,
        then put it in $\listofmclasses$, else discard it.  (We
        interchangeably regard $\listofmclasses$ as a list of
        $\mc{L}[m]$ sentences or a list of $\lequiv{m}$ classes.)
    \end{enumerate}
  \item For $\sigma \in \sigmarank$ and $d = \rho(\sigma)$, generate
    $\gsigma{m}: (\listofmclasses)^d \rightarrow \listofmclasses$ as
    follows. Given $\xi_i \in \listofmclasses$ for $i \in \{1, \ldots,
    d\}$, find models $\tree{s}_i$ for $\xi_i$ in $\cl{S}$. Let
    $\tree{s}$ be the tree obtained by making $\tree{s}_1, \ldots,
    \tree{s}_n$ as the child subtrees (and in that sequence) of a new
    root node labeled with $\sigma$. Find out $\xi \in
    \listofmclasses$ of which $\tree{s}$ is a model. Then define
    $\gsigma{m}(\xi_1, \ldots, \xi_d) = \xi$. Generate
    $\gsigmai{m}{1}: \listofmclasses \rightarrow \listofmclasses$
    similarly.
  \item For $\sigma \in \Sigma \setminus \sigmarank$, generate
    $\gsigmai{m}{2}: (\listofmclasses)^2 \rightarrow \listofmclasses$
    as follows. For $\xi_1, \xi_2 \in \listofmclasses$, find models
    $\tree{s}_1$ and $\tree{s}_2$ resp. in $\cl{S}$. such that the
    root of $\tree{s}_1$ is labeled with $\sigma$ (this condition on
    the root can be captured by an FO sentence). If no $\tree{s}_1$ is
    found, then define $\gsigmai{m}{2}(\xi_1, \xi_2) =
    \xi_{\text{default}}$ where the latter is some fixed element of
    $\listofmclasses$. Else, let $\tree{v}_{\xi_1, \xi_2}$ be the tree
    obtained adding $\tree{s}_2$ as the (new) ``last'' child subtree
    of the root of $\tree{s}_1$. Find out $\xi \in \listofmclasses$ of
    which $\tree{v}_{\xi_1, \xi_2}$ is a model. Define
    $\gsigmai{m}{2}(\xi_1, \xi_2) = \xi$.
  \end{enumerate}

  It is clear that there exists a computable function $\eta_3:
  \mathbb{N} \rightarrow \mathbb{N}$ such that the running time of
  $\enumcompfun(m)$ is at most $\eta_3(m)$.  We now claim that
  $\gsigma{m}$ and $\gsigmai{m}{i}$ generated by $\enumcompfun(m)$
  indeed satisfy the composition properties of
  Lemma~\ref{lemma:composition-lemma-for-partially-ranked-trees},
  whereby they can be indeed taken as $\fsigma{m}$ and
  $\fsigmai{m}{i}$ appearing in the latter lemma. That $\gsigma{m}$
  and $\gsigmai{m}{1}$ satisfy the composition properties is easy to
  see using Corollary~\ref{corollary:tree-composition}. To reason for
  $\gsigmai{m}{2}$, consider a tree $\tree{t}$ whose root is labeled
  with $\sigma$, and which has say $3$ children $a_1, \ldots, a_3$
  (and in that sequence) such that the $\lequiv{m}$ class of
  $\tree{t}_{\ge a_i}$ is $\delta_i$ for $1 \leq i \leq 3$. Consider
  the subtrees $\tree{x}$ and $\tree{y}$ of $\tree{t}$ defined as
  $\tree{x} = \tree{t} - \tree{t}_{\ge a_3}$ and $\tree{y} = \tree{x}
  - \tree{x}_{\ge a_2}$. Let $\delta_4$ and $\delta_5$ be resp. the
  $\lequiv{m}$ classes of $\tree{x}$ and $\tree{y}$. Now consider the
  trees $\tree{v}_{\delta_5, \delta_2}$ and $\tree{v}_{\delta_4,
    \delta_3}$ which are guaranteed to be found (since indeed
  $\tree{x}$ and $\tree{y}$ \emph{are} trees each of whose roots is
  labeled with $\sigma$). By
  Corollary~\ref{corollary:tree-composition}, $\tree{x} \lequiv{m}
  \tree{v}_{\delta_5, \delta_2}$ and $\tree{t} \lequiv{m}
  \tree{v}_{\delta_4, \delta_3}$. Whereby, the $\lequiv{m}$ class of
  $\tree{x}$ is $\delta_4 = \gsigmai{m}{2}(\delta_5, \delta_2)$ and
  that of $\tree{t}$ is $\delta = \gsigmai{m}{2}(\delta_4,
  \delta_3)$. Observe that $\delta_5$ is indeed
  $\gsigmai{m}{1}(\delta_1)$.

\vspace{5pt} (Part \ref{lemma:core-lemma-reduce-degree-and-height}):
\und{$\reducedeg(\tree{t}, m)$:}
  \begin{enumerate}
    \item Call $\enumcompfun(m)$ that returns the ``composition''
      functions $\fsigma{m}$ and $\fsigmai{m}{i}$, and also gives the
      list $\listofmclasses$ as described above.
    \item Using the composition functions, construct bottom-up in
      $\tree{t}$, the function $\colour: \text{Nodes}(\tree{t})
      \rightarrow \listofmclasses$ such that for each node $a$ of
      $\tree{t}$, $\colour(a)$ is the $\lequiv{m}$ class of
      $\tree{t}_{\ge a}$.
    \item For $\eta_1$ as given by
      Lemma~\ref{lemma:degree-and-height-reduction-for-trees}, if the
      degree of $\tree{t}$ is $\leq \eta_1(m)$, then return
      $\tree{t}$.
    \item Else, let $a$ be a node of $\tree{t}$ of degree $n >
      \eta_1(m)$. Let $\tree{x} = \tree{t}_{\ge a}$.
    \item For each $\delta \in \listofmclasses$, do the following:
      \begin{enumerate}
        \item Let $a_1, \ldots, a_n$ be the children of $a$ in
          $\tree{x}$. For $k \in \{1, \ldots, n\}$, let $\tree{x}_{1,
          k}$ be the subtree of $\tree{x}$ obtained by deleting the
          subtrees rooted at $a_{k+1}, \ldots, a_n$. Let $g:\{1,
          \ldots, n\} \rightarrow \listofmclasses$ be such that $g(i)$
          is the $\lequiv{m}$ class of $\tree{x}_{1, k}$.
          \item If $\delta$ appears in the range of $g$, then let $i,
            j$ be resp. the least and greatest indices in $\{1,
            \ldots, n\}$ such that $g(i) = g(j) = \delta$. Let
            $\tree{y}$ be the subtree of $\tree{x}$ obtained by
            deleting the subtrees rooted at $a_{i+1}, \ldots,
            a_j$. Set $\tree{x}:= \tree{y}$.
      \end{enumerate}
      \item Set $\tree{t}:= \tree{t}[\tree{t}_{\ge a} \mapsto
        \tree{x}]$ and go to step 3.
  \end{enumerate}

    Reasoning similarly as in the proof of
    Lemma~\ref{lemma:degree-and-height-reduction-for-trees}(\ref{lemma:degree-reduction-for-trees}),
    we can verify that $\reducedeg(\tree{t}, m)$ indeed returns the
    desired subtree $\tree{s}_1$ of $\tree{t}$. The time taken to
    compute $\colour$ is linear in $|\tree{t}|$, while that for
    computing $g$ is linear in the degree of $a$, whereby the time
    taken to reduce the degree of a node $a$ in any iteration of the
    loop, is $O(\Lambda(m) \cdot \text{degree}(a))$. Then, the total
    time taken by $\reducedeg(\tree{t}, m)$ is $O(\alpha(m) +
    \Lambda(m) \cdot |\tree{t}|)$ for some computable function
    $\alpha: \mathbb{N} \rightarrow \mathbb{N}$.

 \vspace{3pt} \und{$\redheight(\tree{t}, m)$:}
  \begin{enumerate}
    \item Generate $\listofmclasses$ and the function $\colour$ as in
      the previous part.
    \item Construct bottom up in $\tree{t}$, the function $\lowest:
      \text{Nodes}(\tree{t}) \times \listofmclasses \rightarrow
      \text{Nodes}(\tree{t})$ such that for any node $a$ of $\tree{t}$
      and $\delta \in \listofmclasses$, $\lowest(a, \delta)$ gives a
      lowest (i.e. closest to a leaf) node $b$ in $\tree{t}_{\ge a}$
      such that $\colour(b) = \delta$. In other words, $b$ is the only
      node in $\tree{t}_{\ge b}$ such that $\colour(b) = \delta$.
    \item Let $a_1, \ldots, a_n$ be the children of
      $\troot{\tree{t}}$.  Let $\tree{x}_i =$ $ \rainbow(\tree{t}_{\ge
      a_i})$ for $i \in \{1, \ldots, n\}$, where $\rainbow(\tree{x})$
      is described below.
    \item Return $\tree{t}[\tree{t}_{\ge a_1} \mapsto
      \tree{x}_1]\ldots[\tree{t}_{\ge a_n} \mapsto \tree{x}_n]$.
        
  \end{enumerate}

  \und{$\rainbow(\tree{x})$:}
  \begin{enumerate}
    \item Let $a = \troot{\tree{x}}$.
    \item If $b = \lowest(a, \colour(a)) \neq a$, then return
      $\rainbow(\tree{x}_{\ge b})$.
    \item Else, let $b_1, \ldots, b_n$ be the children of
      $\troot{\tree{x}}$. For $i \in \{1, \ldots, n\}$, let
      $\tree{y}_i = \rainbow(\tree{x}_{\ge b_i})$.
    \item Return $\tree{x}[\tree{x}_{\ge b_1} \mapsto
      \tree{y}_1]\ldots[\tree{x}_{\ge b_n} \mapsto \tree{y}_n]$.
   \end{enumerate}

   Using similar reasoning as in the proof of
   Lemma~\ref{lemma:degree-and-height-reduction-for-trees}(\ref{lemma:height-reduction-for-trees}),
   we can verify that algorithm $\rainbow(\tree{x})$, that takes a
   subtree $\tree{x}$ of $\tree{t}$ as input, outputs a subtree
   $\tree{y}$ of $\tree{x}$ such that (i) $\tree{y} \lequiv{m}
   \tree{x}$ and (ii) no path from the root to the leaf of $\tree{y}$
   contains two distinct nodes $a$ and $b$ such that $\colour(a) =
   \colour(b)$. Further, $\rainbow(\tree{x})$ also satisfies the
   following ``colour preservation'' property. Let for a subtree
   $\tree{s}$ of $\tree{t}$, obtained from $\tree{t}$ by removal of
   rooted subtrees and replacements with rooted subtrees,
   $\mc{Q}(\tree{s})$ be a predicate that denotes that the function
   $\mathsf{Colour}$ computed for $\tree{t}$, when restricted to the
   nodes of $\tree{s}$, is such that for any node $a$ of $\tree{s}$,
   $\mathsf{Colour}(a)$ gives the $\lequiv{m}$ class of $\tree{s}_{\ge
     a}$. Then the ``colour preservation'' property says that if the
   input $\tree{x}$ to $\rainbow$ satisfies $\mc{Q}(\cdot)$, then so
   does the output $\tree{y}$ of $\rainbow$.

   From the preceding features of $\rainbow$, we see that the height
   of the output $\tree{y}$ of $\rainbow(\tree{x})$ is at most
   $\Lambda(m)$. The number of ``top level'' recursive calls made by
   $\rainbow(\tree{x})$ is linear in the degree of $\troot{\tree{x}}$,
   whereby the total time taken by $\rainbow(\tree{x})$ is linear in
   $|\tree{x}|$. The time taken to compute
   $\mathsf{Lowest}\text{-}\mathsf{subtree}$ is easily seen to be
   $O(\Lambda(m) \cdot |\tree{t}|)$.  Then the time taken by
   $\redheight(\tree{t}, m)$ is $O(\eta_3(m) + \Lambda(m) \cdot
   |\tree{t}|)$. One can verify that $\redheight(\tree{t}, m)$ indeed
   returns the desired subtree $\tree{s}_2$ of $\tree{t}$.
\end{proof}

\section{Lifting to tree representations}\label{section:abstract-tree-result}

We now consider the more abstract setting of tree representations of
structures, in which the internal nodes are labeled with operations
coming from a finite set and the leaf nodes represent structures from
a given class of structures. We show that under suitable assumptions
on the tree representations (that a variety of classes of structures
satisfy as seen in the forthcoming sections), we can lift the
techniques seen in the previous section to show the $\reflebsp$
property for classes of structures that admit the aforesaid
representations.

Fix finite alphabets $\sigmaint$ and $\sigmaleaf$ (where the two
alphabets are allowed to be overlapping). Let $\sigmarank \subseteq
\sigmaint$. Let $\rho: \sigmaint \rightarrow \mathbb{N}_{+}$ be a
fixed function. We say a class $\mc{T}$ of $(\sigmaint \cup
\sigmaleaf)$-trees is \emph{representation-feasible for $(\sigmarank,
  \rho)$} if $\mc{T}$ is closed under (label-preserving) isomorphisms,
and for all trees $\tree{t} = (O, \lambda) \in \mc{T}$ and nodes $a$
of $\tree{t}$, the following conditions hold:
\vspace{1pt}
\begin{enumerate}[nosep]
\item Labeling condition: If $a$ is a leaf node, resp. internal node,
  then the label $\lambda(a)$ belongs to $\sigmaleaf$,
  resp. $\sigmaint$.
\item Ranking by $\rho$: If $a$ is an internal node and $\lambda(a)$
  is in $\sigmarank$, then the number of children of $a$ in $\tree{t}$
  is exactly $\rho(\lambda(a))$.
\item Closure under rooted subtrees: The subtree $\tree{t}_{\ge a}$ is
  in $\mc{T}$.
\item Closure under removal of rooted subtrees respecting
  $\sigmarank$: If $a$ is an internal node, $b$ is a child of $a$ in
  $\tree{t}$ and $\lambda(a) \notin \sigmarank$, then the subtree
  $(\tree{t} - \tree{t}_{\ge b})$ is in $\mc{T}$.
\item Closure under replacements with rooted subtrees: If $a$ is an
  internal node, then for every descendent $b$ of $a$ in $\tree{t}$,
  the subtree $\tree{t} \left[ \tree{t}_{\ge a} \mapsto \tree{t}_{\ge
      b} \right]$ is in $\mc{T}$.
\end{enumerate}

We say $\mc{T}$ is \emph{representation-feasible} if there exist
alphabets $\sigmaleaf, \sigmaint$ and $\sigmarank$ and function $\rho:
\sigmaint \rightarrow \mathbb{N}_{+}$ such that $\mc{T}$ is a class of
$(\sigmaint \cup \sigmaleaf)$-trees that is representation feasible
for $(\sigmarank, \rho)$.  Given such a class $\mc{T}$ of trees and a
class $\cl{S}$ of structures, let $\mathsf{Str}: \mc{T} \rightarrow
\cl{S}$ be a map that associates with each tree in $\mc{T}$, a
structure in $\cl{S}$. We call $\mathsf{Str}$ a \emph{representation
  map}. For a tree $\tree{t} \in \mc{T}$, if $\mf{A} =
\str{\tree{t}}$, then we say $\tree{t}$ is a \emph{tree
  representation} of $\mf{A}$ under $\mathsf{Str}$.  For the purposes
of our result, we consider ``good'' maps that would allow tree
reductions of the kind seen in the previous section. We formally
define these below:
\begin{definition}\label{definition:L-good-map}
Given a class $\cl{S}$ of structures and a representation-feasible
class $\mc{T}$ of trees, a representation map $\mathsf{Str}: \mc{T}
\rightarrow \cl{S}$ is said to be \emph{$\mc{L}$-good for $\cl{S}$} if
it has the following properties:
\begin{enumerate}[nosep]
\item Isomorphism preservation: $\mathsf{Str}$ maps isomorphic
  (labeled) trees to isomorphic structures.
\item Surjectivity: Each structure in $\cl{S}$ has an isomorphic
  structure in the range of $\mathsf{Str}$.
\item Monotonicity: Let $\tree{t} \in \mc{T}$ be a tree of size $\ge
  2$, and $a$ be a node of $\tree{t}$.
  \begin{enumerate}[nosep]
  \item \label{A.1} If $\tree{s} = \tree{t}_{\ge a}$, then
    $\str{\tree{s}} \hookrightarrow \str{\tree{t}}$
  \item \label{A.2} If $b$ is a child of $a$ in $\tree{t}$,
    $\lambda(a) \notin \sigmarank$ and $\tree{z} = (\tree{t} -
    \tree{t}_{\ge b})$, then $\str{\tree{z}} \hookrightarrow
    \str{\tree{t}}$.
  \item \label{A.3} If $b$ is a descendent of $a$ in $\tree{t}$ and
    $\tree{z} = \tree{t} \left[ \tree{t}_{\ge a} \mapsto \tree{t}_{\ge
      b} \right]$, then $\str{\tree{z}} \hookrightarrow
    \str{\tree{t}}$.
  \end{enumerate}
\item \label{B} Composition: There exists $m_0 \in \mathbb{N}$ such
  that for every $m \ge m_0$ and for every $\sigma \in \sigmaint$,
  there exists a function $\fsigma{m}:
  (\ldelta{\cl{S}}{m})^{\rho(\sigma)} \rightarrow \ldelta{\cl{S}}{m}$
  if $\sigma \in \sigmarank$, and functions $\fsigmai{m}{i}:
  (\ldelta{\cl{S}}{m})^i \rightarrow \ldelta{\cl{S}}{m}$ for $i \in
  \{1, \ldots, \rho(\sigma)\}$ if $\sigma \in \sigmaint \setminus
  \sigmarank$, with the following properties: Let $\tree{t} = (O,
  \lambda) \in \mc{T}$ and $a$ be an internal node of $\tree{t}$ such
  that $\lambda(a) = \sigma$ and the children of $a$ in $\tree{t}$ are
  $b_1, \ldots, b_n$. Let $\delta_i$ be the $\lequiv{m}$ class of
  $\mathsf{Str}(\tree{t}_{\ge b_i})$ for $i \in \{1, \ldots, n\}$, and
  let $\delta$ be the $\lequiv{m}$ class of $\str{\tree{t}_{\ge a}}$.
  \begin{itemize}[nosep]
  \item If $\sigma \in \sigmarank$ (whereby $n = \rho(\sigma)$), then
    $\delta = \fsigma{m}(\delta_1, \ldots, \delta_n)$.
  \item If $\sigma \in \sigmaint \setminus \sigmarank$, then $\delta$
    is given as follows: Let $d = \rho(\sigma)$ and $n = r + q \cdot
    (d - 1)$ where $1 \leq r < d$. Let $I = \{r + j \cdot (d - 1) \mid
    0 \leq j \leq q\}$ and for $k \in I, k \neq n$, let $\chi_{k + (d
      - 1)} = \fsigmai{m}{d}(\chi_k, \delta_{k+1}, \ldots, \delta_{k +
      (d -1)})$ where $\chi_r = \fsigmai{m}{r}(\delta_1, \ldots,
    \delta_r)$. Then $\delta = \chi_n$.
    
  \end{itemize}
\end{enumerate}
\end{definition}

We say $\cl{S}$ \emph{admits an $\mc{L}$-good tree representation} if
there exists some representation map $\mathsf{Str}$ that is
$\mc{L}$-good for $\cl{S}$. We say an $\mc{L}$-good tree
representation $\mathsf{Str}: \mc{T} \rightarrow \cl{S}$ is
\emph{effective (resp. elementary)} if (i) $\mc{T}$ is recursive and
(ii) there is an algorithm that, given $\tree{t} \in \mc{T}$ as input,
computes $\str{\tree{t}}$ (resp. computes $\str{\tree{t}}$ in time
which is bounded by an elementary function of $|\tree{t}|$).  We now
present the central result of this section, which is a lifting of
Proposition~\ref{prop:result-for-partially-ranked-trees} to tree
representations. The proof involves an abstraction of all the ideas
presented in proof of
Proposition~\ref{prop:result-for-partially-ranked-trees}.

\begin{theorem}\label{theorem:good-tree-rep-implies-lebsp-and-f.p.t.-algorithm}
Let $\cl{S}$ be a class of structures that admits an $\mc{L}$-good
tree representation $\mathsf{Str}: \mc{T} \rightarrow \cl{S}$. Then
the following are true:
\vspace{3pt}\begin{enumerate}[nosep]
\item $\lebsp{\cl{S}}$
  holds.\label{theorem:good-tree-rep-implies-lebsp-and-f.p.t.-algorithm:point-1}
\item If \,$\mathsf{Str}$ is effective, then there exists a computable
  witness function for $\lebsp{\cl{S}}$. Further, there exists a
  linear time f.p.t. algorithm for $\mathsf{MC}(\mc{L}, \cl{S})$ that
  decides, for every $\mc{L}$ sentence $\varphi$ (the parameter), if a
  given structure $\mf{A}$ in $\cl{S}$ satisfies $\varphi$, provided
  that a tree representation of $\mf{A}$ under $\mathsf{Str}$ is
  given.
  \label{theorem:good-tree-rep-implies-lebsp-and-f.p.t.-algorithm:point-2}
\item If \,$\mathsf{Str}$ is elementary, then there exists an elementary
  witness function for $\lebsp{\cl{S}}$ iff the index of the
  $\lequiv{m}$ relation over $\cl{S}$ has an elementary dependence on
  $m$.
  \label{theorem:good-tree-rep-implies-lebsp-and-f.p.t.-algorithm:point-3}
\end{enumerate}
\end{theorem}

The rest of this section is entirely devoted to proving the above result.

\vspace{5pt}We prove
Theorem~\ref{theorem:good-tree-rep-implies-lebsp-and-f.p.t.-algorithm}
analogous to
Proposition~\ref{prop:result-for-partially-ranked-trees}. Specifically,
we show the following two results which resp. are abstract versions of
Lemma~\ref{lemma:degree-and-height-reduction-for-trees} and
Lemma~\ref{lemma:core-lemma-for-fpt-over-trees}.

\begin{lemma}\label{lemma:abstract-tree-lemma}
For a class $\cl{S}$ of structures, and a representation-feasible
class $\mc{T}$ of trees, let $\mathsf{Str}: \mc{T} \rightarrow \cl{S}$
be a representation map that is $\mc{L}$-good for $\cl{S}$. Then there
exist computable functions $\eta_1, \eta_2: \mathbb{N} \rightarrow
\mathbb{N}$ such that for each $\tree{t} \in \mc{T}$ and $m \in
\mathbb{N}$, we have the following:
\begin{enumerate}[nosep]
\item \emph{(Degree reduction)} There exists a subtree $\tree{s}_1$ of
  $\tree{t}$ in $\mc{T}$, of degree $\leq \eta_1(m)$, such that (i)
  the roots of $\tree{s}_1$ and $\tree{t}$ are the same, (ii)
  $\str{\tree{s}_1} \hookrightarrow \str{\tree{t}}$, and (iii)
  $\str{\tree{s}_1} \lequiv{m} \str{\tree{t}}$.
  \label{lemma:abstract-tree-lemma-degree-reduction}
\item \emph{(Height reduction)} There exists a subtree $\tree{s}_2$ of
  $\tree{t}$ in $\mc{T}$, of height $\leq \eta_2(m)$, such that (i)
  the roots of $\tree{s}_2$ and $\tree{t}$ are the same, (ii)
  $\str{\tree{s}_2} \hookrightarrow \str{\tree{t}}$, and (ii)
  $\str{\tree{s}_2} \lequiv{m} \str{\tree{t}}$.
\label{lemma:abstract-tree-lemma-height-reduction}
\end{enumerate}
Above, it additionally holds that if the index of the $\lequiv{m}$
relation over $\cl{S}$ is an elementary function of $m$, then each of
$\eta_1$ and $\eta_2$ is elementary as well.
\end{lemma}

\begin{lemma}\label{lemma:core-lemma-for-fpt-over-trees-abstract}
For a class $\cl{S}$ of structures, and a representation-feasible
class $\mc{T}$ of trees, let $\mathsf{Str}: \mc{T} \rightarrow \cl{S}$
be a representation map that is $\mc{L}$-good for $\cl{S}$ and
effective. Let $m_0$ witness the composition property of
$\mathsf{Str}$, as mentioned in Definition~\ref{definition:L-good-map}.
There exist computable functions $\eta_3, \eta_4, \eta_5: \mathbb{N}
\rightarrow \mathbb{N}$ and algorithms $\enumcompfun(m)$,
$\reducedeg(\tree{t}, m)$ and $\redheight(\tree{t}, m)$ such that for
$m \ge m_0$,

\begin{enumerate}[nosep]
\item $\enumcompfun(m)$ generates in time $\eta_3(m)$, the functions
  $\fsigma{m}$ if $\sigma \in \sigmarank$ and $\fsigmai{m}{i}$ for $i
  \in \{1, \ldots, \rho(\sigma)\}$ if $\sigma \in \sigmaint \setminus
  \sigmarank$, that satisfy the properties mentioned in
  Definition~\ref{definition:L-good-map}.\label{lemma:core-lemma-for-fpt-enum}

\item For $\tree{t} \in \cl{T}$, $\reducedeg(\tree{t}, m)$ computes
  the subtree $\tree{s}_1$ of $\tree{t}$ as given by
  Lemma~\ref{lemma:abstract-tree-lemma}, in time $\eta_4(m) \cdot
  |\tree{t}|$. Likewise, $\redheight(\tree{t}, m)$ computes the
  subtree $\tree{s}_2$ of $\tree{t}$ as given by
  Lemma~\ref{lemma:degree-and-height-reduction-for-trees}, in time
  $\eta_5(m) \cdot |\tree{t}|$.\label{lemma:core-lemma-for-fpt-reduce}
\end{enumerate}

\end{lemma}

\begin{proof}[Proof of Theorem~\ref{theorem:good-tree-rep-implies-lebsp-and-f.p.t.-algorithm}]
  (1): Let $\mf{A} \in \cl{S}$. Let $\tree{t}$ be such that
  $\str{\tree{t}} = \mf{A}$. By Lemma~\ref{lemma:abstract-tree-lemma},
  there exists a subtree $\tree{s}$ of $\tree{t}$ in $\mc{T}$, of
  degree $\leq \eta_1(m)$ and height $\leq \eta_2(m)$, and hence of
  size $\leq p = \eta_1(m)^{(\eta_2(m) + 1)}$, such that (i)
  $\str{\tree{s}} \hookrightarrow \str{\tree{t}}$ and (ii)
  $\str{\tree{s}} \lequiv{m} \str{\tree{t}}$. Define
  $\lwitfn{\cl{S}}(m) = \text{max}\{|\mf{C}| \mid \mf{C} \in \cl{S},
  \text{~there exists}~\tree{z}~\text{in}~\mc{T}~\text{such
    that}~\str{\tree{z}} = \mf{C}~\text{and}~|\tree{z}| \leq p\}$. It
  is then easy to see taking $\mf{B}$ to be the isomorphic copy of
  $\str{\tree{s}}$, that is a substructure of $\mf{A}$, that
  $\lebspcond(\cl{S}, \mf{A}, \mf{B}, m, \lwitfn{\cl{S}})$ holds.

  \vspace{2pt}(2): It is clear that if $\mathsf{Str}$ is effective,
  then $\lwitfn{\cl{S}}$ defined above is computable too. For the
  f.p.t. part, let $\mc{A}$ be the following algorithm. Let
  $\mf{A} \in \cl{S}$ be given as input to $\mc{A}$, in the form of
  the tree representation $\tree{t}$ of $\mf{A}$ under
  $\mathsf{Str}$. Let $\varphi$ be an input $\mc{L}$ sentence. Then
  $\mc{A}$ determines the rank $m$ of $\varphi$, computes $m_1
  = \text{max}\{m, m_0\}$, and calls $\redheight(\reducedeg(\tree{t},
  m_1), m_1)$. By
  Lemma~\ref{lemma:core-lemma-for-fpt-over-trees-abstract}, the
  aforesaid call returns, in time $(\eta_4(m_1) + \eta_5(m_1)) \cdot
  |\tree{t}|$, a tree $\tree{s}$ in $\mc{T}$, of degree
  $\leq \eta_1(m_1)$ and height $\leq \eta_2(m_1)$, and hence of size
  $\leq \eta_1(m_1)^{(\eta_2(m_1) + 1)}$, such that
  $\str{\tree{s}} \lequiv{m_1} \str{\tree{t}} = \mf{A}$. Since
  $m_1 \ge m$, we have $\str{\tree{s}} \lequiv{m} \mf{A}$. Checking if
  $\mf{A} \models \varphi$ is then equivalent to checking if
  $\str{\tree{s}} \models \varphi$, and the latter can be done in time
  $g(m_1)$ for some computable function
  $g: \mathbb{N} \rightarrow \mathbb{N}$ of $m_1$, since the size of
  $\tree{s}$ is bounded by a computable function of $m_1$.  It follows
  that $\mc{A}$ is f.p.t. for $\mathsf{MC}(\mc{L}, \cl{S})$.

  \vspace{2pt}(3): It is easy to see that if there exists an
  elementary witness function $\lwitfn{\cl{S}}$ for $\lebsp{\cl{S}}$,
  then every structure $\mf{A}$ in $\cl{S}$ is $\mc{L}[m]$-equivalent
  to a structure $\mf{B}$ in $\cl{S}$ such that
  $|\mf{B}| \leq \lwitfn{\cl{S}}(m)$. Whereby the index of the
  $\lequiv{m}$ relation over $\cl{S}$ is bounded by the number of
  structures in $\cl{S}$ whose size (of the universe) is
  $\leq \lwitfn{\cl{S}}(m)$. Clearly then, this number, and hence the
  index, is bounded by an elementary function of $m$, if
  $\lwitfn{\cl{S}}$ is elementary.

  Suppose the index of the $\lequiv{m}$ relation over $\cl{S}$ is an
  elementary function of $m$. Then by
  Lemma~\ref{lemma:abstract-tree-lemma}, $\eta_1$ and $\eta_2$ are
  elementary too. Whereby if $\mathsf{Str}$ is also elementary, then
  $\lwitfn{\cl{S}}$ as defined in part (1) above, is also elementary.
\end{proof}

We now prove Lemma~\ref{lemma:abstract-tree-lemma} and
Lemma~\ref{lemma:core-lemma-for-fpt-over-trees-abstract}.  We recall
from Section~\ref{section:background} that for a class $\cl{S}$ of
structures, $\ldelta{\cl{S}}{m}$ denotes the set of all equivalence
classes of the $\lequiv{m}$ relation restricted to the structures in
$\cl{S}$, and $\Lambda_{\cl{S}, \mc{L}}: \mathbb{N} \rightarrow
\mathbb{N}$ is a fixed computable function with the property that
$\Lambda_{\cl{S}, \mc{L}}(m) \ge |\ldelta{\cl{S}}{m}|$.

\subsection{Proof of Lemma~\ref{lemma:abstract-tree-lemma}}
The following facts are easy to verify given that $\mathsf{Str}$
satisfies the composition properties of
Definition~\ref{definition:L-good-map}. The proofs of these use
similar ideas as in the proof of
Corollary~\ref{corollary:tree-composition}, and are hence
skipped. Below, $m_0$ witnesses the composition properties of
$\mathsf{Str}$ as given by Definition~\ref{definition:L-good-map}.

\begin{lemma}\label{lemma:helper-height-reduction} 
  Let $\tree{s}, \tree{t} \in \mc{T}$ and let $a$ be a node of
  $\tree{t}$. Suppose $\tree{z} = \tree{t}[\tree{t}_{\ge
  a} \mapsto \tree{s}] \in \mc{T}$. Then for $m \ge m_0$, if
  $\str{\tree{s}} \lequiv{m} \str{\tree{t}_{\ge a}}$, then
  $\str{\tree{z}} \lequiv{m} \str{\tree{t}}$.
\end{lemma}

\begin{lemma}\label{lemma:helper-degree-reduction}
Let $\tree{s}_1, \tree{s}_2, \tree{t} \in \mc{T}$ be such that the
label of the root of each of these trees is $\sigma \in \sigmaint
\setminus \sigmarank$.  Suppose $\tree{z}_i = \tree{s}_i \odot
\tree{t}$ is such that $\tree{z}_i \in \mc{T}$ for $i \in \{1,
2\}$. Suppose further that the number of children of the root of
$\tree{t}$ is a multiple of $(\rho(\sigma) - 1)$. Then for $m \ge
m_0$, if $\str{\tree{s}_1} \lequiv{m} \str{\tree{s}_2}$, then
$\str{\tree{z}_1} \lequiv{m} \str{\tree{z}_2}$.
\end{lemma}

\begin{proof}[Proof of Lemma~\ref{lemma:abstract-tree-lemma}]

  (Part~\ref{lemma:abstract-tree-lemma-degree-reduction}): Let $m_0
  \in \mathbb{N}$ be a witness to the composition property of
  $\mathsf{Str}$, as mentioned in
  Definition~\ref{definition:L-good-map}.  Define $\eta_1:\mathbb{N}
  \rightarrow \mathbb{N}$ as follows: for $l \in \mathbb{N}$,
  $\eta_1(l) = \text{max}\{\rho(\sigma) \mid \sigma \in \sigmaint \}
  \times \Lambda_{\cl{S}, \mc{L}}(\text{max}\{l, m_0\})$.  Then
  $\eta_1$ is computable.

  Given $m \in \mathbb{N}$, let $p = \eta_1(m)$.  If $\tree{t}$ has
  degree $\leq p$, then putting $\tree{s}_1 = \tree{t}$ we are
  done. Else, some node of $\tree{t}$, say $a$, has degree $n >
  p$. Let $\sigma$ be the label of $a$; clearly $\sigma \in \sigmaint
  \setminus \sigmarank$. Let $\tree{z} = \tree{t}_{\ge a}$; then
  $\tree{z} \in \mc{T}$. Let $a_1, \ldots, a_n$ be the (ascending)
  sequence of children of $\troot{\tree{z}}$ in $\tree{z}$. For $d =
  \rho(\sigma)$, let $n = r + q \cdot (d - 1)$ for $1 \leq r < d$ and
  $q > 1$.

  For $k \in I = \{r + l \cdot (d - 1) \mid 0 \leq l \leq q\}$, let
  $\tree{x}_{1, k}$, resp. $\tree{y}_{k+1, n}$, be the subtree of
  $\tree{z}$ obtained from $\tree{z}$ by deleting the subtrees rooted
  at $a_{k+1}, \ldots, a_n$, resp.  deleting the subtrees rooted at
  $a_1, a_2, \ldots, a_k$. Then $\tree{z} = \tree{x}_{1, n}
  = \tree{x}_{1, k} \odot \tree{y}_{k+1, n}$ for all $k \in I$.  Let
  $m_1 = \text{max}\{m_0, m\}$. Define $g:
  I \rightarrow \ldelta{\cl{S}}{m_1}$ such that $g(k)$ is the
  $\lequiv{m_1}$ class of $\str{\tree{x}_{1, k}}$ for $k \in I$.

  Since $n > p$, there exist $i, j \in I$ such that $i < j$ and $g(i)
  = g(j)$, i.e. $\str{\tree{x}_{1, i}} \lequiv{m_1} \str{\tree{x}_{1,
      j}}$.  If $\tree{z}_1 = \tree{x}_{1, i} \odot \tree{y}_{j+1,
    n}$, then since $\mc{T}$ is closed under removal of rooted
  subtrees respecting $\sigmarank$, we have $\tree{z}_1 \in
  \mc{T}$. Observe that $\tree{z} = \tree{x}_{1, j} \odot
  \tree{y}_{j+1, n}$. Then by
  Lemma~\ref{lemma:helper-degree-reduction} and the monotonicity
  properties of $\mathsf{Str}$ as mentioned in
  Definition~\ref{definition:L-good-map}, we have $\str{\tree{z}_1}
  \hookrightarrow \str{\tree{z}}$ and $\str{\tree{z}_1} \lequiv{m}
  \str{\tree{z}}$. Then by Lemma~\ref{lemma:helper-height-reduction}
  and the monotonicity properties of $\mathsf{Str}$, we see that if
  $\tree{t}_1 = \tree{t} \left[\tree{z} \mapsto \tree{z}_1 \right]$,
  then $\tree{t}_1 \in \mc{T}$, $\str{\tree{t}_1} \hookrightarrow
  \str{\tree{t}}$ and $\str{\tree{t}_1} \lequiv{m_1} \str{\tree{t}}$.
  Observe that $\tree{t}_1$ has strictly lesser size than $\tree{t}$
  (since $\tree{z}_1$ has strictly lesser size than $\tree{z}$), and
  that the roots of $\tree{t}_1$ and $\tree{t}$ are the
  same. Recursing on $\tree{t}_1$, we eventually get a subtree
  $\tree{s}_1$ of $\tree{t}$ in $\mc{T}$, of degree at most $p$, such
  that (i) the roots of $\tree{s}_1$ and $\tree{t}$ are the same, (ii)
  $\str{\tree{s}_1} \hookrightarrow \str{\tree{t}}$, and (iii)
  $\str{\tree{s}_1} \lequiv{m_1} \str{\tree{t}}$. Since $m_1 =
  \text{max}\{m_0, m\} \ge m$, we have $\str{\tree{s}_1} \lequiv{m}
  \str{\tree{t}}$.

  \vspace{4pt}(Part~\ref{lemma:abstract-tree-lemma-height-reduction}):
  As in the previous part, let $m_0 \in \mathbb{N}$ be a witness to
  the composition property of $\mathsf{Str}$, as mentioned in
  Definition~\ref{definition:L-good-map}.  Define
  $\eta_2:\mathbb{N} \rightarrow \mathbb{N}$ as follows: for
  $l \in \mathbb{N}$, $\eta_2(l) = 1
  + \Lambda_{\cl{S}, \mc{L}}(\text{max}\{l, m_0\})$.  Then $\eta_2$ is
  computable.

  Given $m \in \mathbb{N}$, let $p = \eta_2(m)$.  If $\tree{t}$ has
  height $\leq p$, then putting $\tree{s}_2 = \tree{t}$ we are
  done. Else, there is a path from the root of $\tree{t}$ to some leaf
  of $\tree{t}$, whose length is $> p$. Let $A$ be the set of nodes
  appearing along this path. Let $m_2 = \text{max}(m_0, m)$.  Consider
  the function $h : A \rightarrow \ldelta{\cl{S}}{m_2}$ such that for
  each $a \in A$, $h(a) = \delta$ where $\delta$ is the $\lequiv{m_2}$
  class of $\str{\tree{t}_{\ge a}}$. Since $|A| > p$, there exist
  distinct nodes $a, b \in A$ such that $a$ is an ancestor of $b$ in
  $\tree{t}$ and $h(a) = h(b)$ and $a$ is not the root of
  $\tree{t}$. Let $\tree{t}_2 = \tree{t}\left[\tree{t}_{\ge a} \mapsto
    \tree{t}_{\ge b} \right]$.  Since $\mc{T}$ is closed under rooted
  subtrees and under replacements with rooted subtrees, we have that
  $\tree{t}_2$ is a subtree of $\tree{t}$ in $\mc{T}$.  By the
  monotonicity properties mentioned in
  Definition~\ref{definition:L-good-map} that $\mathsf{Str}$
  satisfies, $\str{\tree{t}_2} \hookrightarrow \str{\tree{t}}$. Also
  since $h(a) = h(b)$, we have $\str{\tree{t}_{\ge b}} \lequiv{m_2}
  \str{\tree{t}_{\ge a}}$, whereby using
  Lemma~\ref{lemma:helper-height-reduction}, we get that
  $\str{\tree{t}_2} \lequiv{m_2} \str{\tree{t}}$. Observe that
  $\tree{t}_2$ has strictly less size than $\tree{t}$, and that the
  roots of $\tree{t}_2$ and $\tree{t}$ are the same.  Recursing on
  $\tree{t}_2$, we eventually get a subtree $\tree{s}_2$ of
  $\tree{t}$, of height at most $p$, such that (i) the roots of
  $\tree{s}_2$ and $\tree{t}$ are the same, (ii) $\str{\tree{s}_2}
  \hookrightarrow \str{\tree{t}}$, and (iii) $\str{\tree{s}_2}
  \lequiv{m_2} \str{\tree{t}}$. Since $m_2 = \text{max}\{m_0, m\} \ge
  m$, we have $\str{\tree{s}_2} \lequiv{m} \str{\tree{t}}$.\\

  \vspace{3pt}It is clear from the definitions of $\eta_1$ and
  $\eta_2$ above, that if the index of the $\lequiv{m}$ relation over
  $\cl{S}$ is an elementary function of $m$, then so are $\eta_1$ and
  $\eta_2$.
\end{proof}

\subsection{Proof of Lemma~\ref{lemma:core-lemma-for-fpt-over-trees-abstract}}

We now give the proofs of
Lemma~\ref{lemma:core-lemma-for-fpt-over-trees-abstract}(\ref{lemma:core-lemma-for-fpt-enum})
and
Lemma~\ref{lemma:core-lemma-for-fpt-over-trees-abstract}(\ref{lemma:core-lemma-for-fpt-reduce})
in Section~\ref{subsubsection:generate-functions} and
Section~\ref{subsubsection:computational-reductions} respectively.

\subsubsection{Proof of  Lemma~\ref{lemma:core-lemma-for-fpt-over-trees-abstract}(\ref{lemma:core-lemma-for-fpt-enum})}\label{subsubsection:generate-functions}

Before we present the proof, we need some auxiliary lemmas that we
describe below.

Let $\mathsf{All}$ denote the class of all finite structures.
\begin{lemma}[Enumerability of the equivalence classes of  $\ldelta{\mathsf{All}}{m}$]\label{lemma:equivalence-class-enumeration}
There exists a computable function
$h: \mathbb{N} \rightarrow \mathbb{N}$ and a procedure $\mc{P}$ such
that $\mc{P}$ takes as input a natural number $m$ and enumerates
$\mc{L}[m]$ sentences $\varphi_1, \varphi_2, \ldots, \varphi_n$ for $n
= h(m)$ with the property that $\varphi_i$ captures some equivalence
class $\delta$ of $\ldelta{\mathsf{All}}{m}$ (i.e. the class of finite
models of $\varphi_i$ is exactly $\delta$) for each
$i \in \{1, \ldots, n\}$ and conversely, for every equivalence class
$\delta$ of $\ldelta{\mathsf{All}}{m}$, there exists some $i \in \{1,
\ldots, n\}$ such that $\varphi_i$ captures $\delta$ .
\end{lemma}
\begin{proof}
Follows from the inductive definition of $\mc{L}[m]$, and the proofs
of Lemma 3.13 and Proposition 7.5 in~\cite{libkin}.
\end{proof}

Let $\mc{L}$-$\mathsf{SAT}$ denote the problem of checking if a given
$\mc{L}$ sentence is satisfiable.

\begin{lemma}\label{lemma:decidability-of-L-SAT}
If $\lebsp{\cl{S}}$ is true with a computable witness function, then
 $\mc{L}$-$\mathsf{SAT}$ is decidable over $\cl{S}$.
\end{lemma}
\begin{proof}
Since for any structure in $\cl{S}$ and $m \in \mathbb{N}$, there is
an $\mc{L}[m]$-equivalent substructure of size bounded by a
computable function of $m$, it follows that $\mc{L}$ possesses the
``computable'' small model property over $\cl{S}$. The decidability of
$\mc{L}$-$\mathsf{SAT}$ over $\cl{S}$ then follows.
\end{proof}

Let as usual, $m_0$ witness the composition properties of
$\mathsf{Str}$ as mentioned in Definition~\ref{definition:L-good-map}.

\begin{lemma}\label{lemma:small-tree-given-input-equivalence-classes}
There exists a computable function $\eta: \mathbb{N} \rightarrow
\mathbb{N}$ with the following property: Let $\tree{t} \in \mc{T}$ of
size $\ge 2$ and $a_1, \ldots, a_n$ be the children of
$\troot{\tree{t}}$. For each $m \ge m_0$, there exists a subtree
$\tree{s}$ of $\tree{t}$ in $\mc{T}$ such that
\begin{enumerate}[nosep]
\item the roots of $\tree{s}$ and $\tree{t}$ are the same
\item the size of $\tree{s}$ is at most $\eta(m)$
\item
  \begin{enumerate}
    \item If $\sigma \in \sigmarank$ (whereby $n = \rho(\sigma)$) or
      $n < \rho(\sigma)$, then the root of $\tree{s}$ has exactly $n$
      children $b_1, \ldots, b_n$ satisfying $\str{\tree{s}_{\ge b_i}}
      \lequiv{m} \str{\tree{t}_{\ge a_i}}$ for each $i \in \{1,
      \ldots, n\}$.
    \item Else, $\tree{s} = \tree{x} \odot \tree{y}$ where
      \begin{itemize}
      \item $\tree{x}$ is such that $\str{\tree{x}} \lequiv{m}
        \str{\tree{z}}$ and $\tree{z}$ is the tree obtained from
        $\tree{t}$ by removing the subtrees rooted at $a_{n-d+2},
        \ldots, a_n$.
      \item $\tree{y}$ is such that the root of $\tree{y}$ has exactly
        $d -1$ children $b_{n - d + 2}, \ldots, b_n$ for $d =
        \rho(\sigma)$, satisfying $\str{\tree{s}_{\ge b_i}} \lequiv{m}
        \str{\tree{t}_{\ge a_i}}$ for each $i \in \{n-d+2, \ldots,
        n\}$.
      \end{itemize}
  \end{enumerate}
\end{enumerate}
\end{lemma}
\begin{proof}
Let $k = \text{max}\{\rho(\sigma) \mid \sigma \in \sigmaint\}$.
Define $\eta(m) = 1 + k \times (\eta_1(m))^{\eta_2(m)+1}$ where
$\eta_1, \eta_2$ are as given by
Lemma~\ref{lemma:abstract-tree-lemma}.
   
Consider the case when $\sigma \in \sigmarank$ (whereby $n =
\rho(\sigma)$) or $n < \rho(\sigma)$. Consider the subtrees
$\tree{x}_i = \tree{t}_{\ge a_i}$ for $i \in \{1, \ldots, n\}$; each
of these belongs to $\mc{T}$ since $\mc{T}$ is
representation-feasible. By parts
(\ref{lemma:abstract-tree-lemma-degree-reduction}) and
(\ref{lemma:abstract-tree-lemma-height-reduction}) of
Lemma~\ref{lemma:abstract-tree-lemma}, it follows that for each $i \in
\{1, \ldots, n\}$, there exists a subtree $\tree{y}_i$ of
$\tree{x}_i$, of degree $\leq \eta_1(m)$ and height $\leq \eta_2(m)$,
and hence of size $\leq (\eta_1(m))^{\eta_2(m)+1} $, such that
$\str{\tree{y}_i} \lequiv{m} \str{\tree{x}_i}$. We observe from the
proofs of parts (\ref{lemma:abstract-tree-lemma-degree-reduction}) and
(\ref{lemma:abstract-tree-lemma-height-reduction}) of
Lemma~\ref{lemma:abstract-tree-lemma}, that $\tree{y}_i$ is obtained
from $\tree{x}_i$ by removal of rooted subtrees in $\tree{x}_i$
respecting $\sigmarank$, and by replacements with rooted subtrees in
$\tree{x}_i$. Whereby, since $\mc{T}$ is representation-feasible, we
tree $\tree{s} = \tree{t}[\tree{x}_1 \mapsto \tree{y}_1][\tree{x}_2
  \mapsto \tree{y}_2] \ldots [\tree{x}_n \mapsto \tree{y}_n]$ obtained
by replacing $\tree{x}_i$ in $\tree{t}$ with $\tree{y}_i$, is indeed a
subtree of $\tree{t}$ in $\mc{T}$, having the properties as mentioned
in the statement of this lemma. Observe that the size of $\tree{s}$ is
at most $1 + n \times (\eta_1(m))^{\eta_2(m)+1}$.

Consider now the case when $\sigma \in \sigmaint \setminus \sigmarank$
and $n \ge \rho(\sigma)$. Let $\tree{t} = \tree{z} \odot \tree{v}$
where $\tree{z}$, resp. $\tree{v}$, is the subtree of $\tree{t}$
obtained by deleting the subtrees rooted at $a_{n - d+ 2}, \ldots,
a_n$, resp. $a_1, \ldots, a_{n - d +1 }$. By using the reasoning
above, there exists a subtree $\tree{y}$ of $\tree{v}$ in $\mc{T}$
such that (i) the roots of $\tree{y}$ and $\tree{v}$ are the same (and
hence $\troot{\tree{y}}$ is labeled with $\sigma$) (ii) the size of
$\tree{y}$ is at most $1 + (d - 1) \times (\eta_1(m))^{\eta_2(m)+1}$
and (iii) the root of $\tree{y}$ has $d-1$ children $b_{n-d+2},
\ldots, b_n$ such that $\str{\tree{s}_{\ge b_i}} \lequiv{m}
\str{\tree{t}_{\ge a_i}}$ for each $i \in \{n-d+2, \ldots, n\}$. Now
consider $\tree{z}$. Again, by Lemma~\ref{lemma:abstract-tree-lemma},
it follows that there exists a subtree $\tree{x}$ of $\tree{z}$ in
$\mc{T}$, of size $\leq (\eta_1(m))^{\eta_2(m)+1} $, such that
$\str{\tree{x}} \lequiv{m} \str{\tree{z}}$ and the roots of $\tree{x}$
and $\tree{z}$ are the same (and hence the label of $\troot{\tree{x}}$
is $\sigma$). Let $\tree{s} = \tree{x} \odot \tree{y}$; then the size
of $\tree{s}$ is at most $(\eta_1(m))^{\eta_2(m)+1} + (d - 1) \times
(\eta_1(m))^{\eta_2(m)+1}$ which in turn is at most $\eta(m)$. We
check that $\tree{s}$ is indeed as desired.
\end{proof}

\begin{proof}[Proof of Lemma~\ref{lemma:core-lemma-for-fpt-over-trees-abstract}(\ref{lemma:core-lemma-for-fpt-enum})]
  
  The procedure $\enumcompfun(m)$ operates in three stages that we
  describe below.

  \vspace{2pt} \und{Stage I:} In this stage, $\enumcompfun(m)$ creates
  a list $\listofmclasses$ of $\mc{L}[m]$ sentences such that every
  sentence of $\listofmclasses$ captures over $\cl{S}$, some
  equivalence class of $\ldelta{\cl{S}}{m}$, and conversely, every
  equivalence classes of $\ldelta{\cl{S}}{m}$ is captured over
  $\cl{S}$, by some sentence of $\listofmclasses$.  This is done as
  follows.

  Let $\eta$ and $\mc{P}$ be as given by
  Lemma~\ref{lemma:equivalence-class-enumeration}.  For each
  $\mc{L}[m]$ sentence $\varphi_i$ for $i \in \{1, \ldots, \eta(m)\}$
  that $\mc{P}$ enumerates, $\enumcompfun(m)$ first checks if
  $\varphi_i$ is satisfiable over $\cl{S}$ (in other words, whether
  $\varphi_i$ indeed represents an equivalence class of
  $\ldelta{\cl{S}}{m}$). This is decidable because $\lebsp{\cl{S}}$
  holds with a computable witness function (using
  part \ref{theorem:good-tree-rep-implies-lebsp-and-f.p.t.-algorithm:point-1}
  and the first part of
  part \ref{theorem:good-tree-rep-implies-lebsp-and-f.p.t.-algorithm:point-2}
  of
  Theorem~\ref{theorem:good-tree-rep-implies-lebsp-and-f.p.t.-algorithm},
  and the assumption that $\mathsf{Str}$ is effective), whereby
  $\mc{L}$-$\mathsf{SAT}$ is decidable over $\cl{S}$ by
  Lemma~\ref{lemma:decidability-of-L-SAT}. If $\varphi_i$ is
  satisfiable, then $\varphi_i$ is put into $\listofmclasses$, else it
  is discarded.

  It follows that at the end of this process, $\listofmclasses$ gets
  created as desired.

  (Since $\listofmclasses$ is a list of sentences that represent
  equivalence classes, we shall henceforth treat $\listofmclasses$
  interchangeably as a list of sentences or a list of equivalence
  classes, depending on what is easier to understand in a given
  context.)

  \vspace{2pt}
  \und{Stage II}: In this stage, the following trees from $\mc{T}$ are
  generated by $\enumcompfun(m)$, if they exist in $\mc{T}$:
  \begin{enumerate}
    \item $\tree{s}_{\sigma, \delta_1, \ldots, \delta_{\rho(\sigma)}}$
      for $\sigma \in \sigmarank$ and $\delta_i \in \listofmclasses$
      for $i \in \{1, \ldots, \rho(\sigma)\}$
    \item $\tree{u}_{\sigma, \delta_1, \ldots, \delta_i}$ for $\sigma
      \in \sigmaint \setminus \sigmarank$, $\delta_j \in
      \listofmclasses$ for $j \in \{1, \ldots, i\}$ and $i \in \{1,
      \ldots, \rho(\sigma) - 1\}$
    \item $\tree{v}_{\sigma, \delta_1, \ldots, \delta_{\rho(\sigma)}}$
      for $\sigma \in \sigmaint \setminus \sigmarank$ and $\delta_i
      \in \listofmclasses$ for $i \in \{1, \ldots, \rho(\sigma)\}$
  \end{enumerate}
  with the following properties:
  \begin{enumerate}
    \item The tree $\tree{z} = \tree{s}_{\sigma, \delta_1, \ldots,
      \delta_{\rho(\sigma)}}$ satisfies the following: (i) the label
      of the root of $\tree{z}$ is $\sigma$, (ii) the root of
      $\tree{z}$ has exactly $\rho(\sigma)$ children $b_1, \ldots,
      b_{\rho(\sigma)}$, and (iii) $\lequiv{m}$ class of
      $\str{\tree{z}_{\ge b_i}}$ is $\delta_i$ for $i \in \{1, \ldots,
      \rho(\sigma)\}$.
    \item The tree $\tree{z} = \tree{u}_{\sigma, \delta_1, \ldots,
      \delta_i}$ satisfies the following: (i) the label of the root of
      $\tree{z}$ is $\sigma$, (ii) the root of $\tree{z}$ has exactly
      $i$ children $b_1, \ldots, b_i$, and (iii) $\lequiv{m}$ class of
      $\str{\tree{z}_{\ge b_j}}$ is $\delta_j$ for $j \in \{1, \ldots,
      i\}$.
    \item The tree $\tree{z} = \tree{v}_{\sigma, \delta_1, \ldots,
      \delta_{\rho(\sigma)}}$ satisfies the following: (i) the label
      of the root of $\tree{z}$ is $\sigma$, (ii) $\tree{z} = \tree{x}
      \odot \tree{y}$ where the root of $\tree{y}$ has exactly $d - 1$
      children $b_2, \ldots, b_{d}$ for $d = \rho(\sigma)$, and (iii)
      the $\lequiv{m}$ class of $\tree{x}$ is $\delta_1$, while the
      $\lequiv{m}$ class of $\str{\tree{y}_{\ge b_j}}$ is $\delta_j$
      for $j \in \{2, \ldots, d\}$.
  \end{enumerate}

  This is done as follows. We show this for the cases of
  $\tree{s}_{\sigma, \delta_1, \ldots, \delta_{\rho(\sigma)}}$ and
  $\tree{v}_{\sigma, \delta_1, \ldots, \delta_{\rho(\sigma)}}$; the
  case of $\tree{u}_{\sigma, \delta_1, \ldots, \delta_i}$ can be done
  similarly. First, using $\eta$ as given by
  Lemma~\ref{lemma:small-tree-given-input-equivalence-classes},
  $\enumcompfun(m)$ computes $p = \eta(m)$. Since the trees in
  $\mc{T}$ are over the finite alphabet $\sigmaint \cup \sigmaleaf$
  and since $\mc{T}$ is recursive, $\enumcompfun(m)$ enumerates out
  those trees in $\mc{T}$, whose roots are labeled with $\sigma$, and
  whose size is $\leq p$. For a tree $\tree{t}$ enumerated thus by
  $\enumcompfun(m)$, let $b_1, \ldots, b_n$ be the children of the
  root of $\tree{t}$.

  \begin{enumerate}
     \item the case of
     $\tree{s}_{\sigma, \delta_1, \ldots, \delta_{\rho(\sigma)}}$:
     Here $\sigma \in \sigmarank$. Then $\enumcompfun(m)$ checks if $n
     = \rho(\sigma)$. If not, then it discards $\tree{t}$. Else,
     $\enumcompfun(m)$ computes $\mf{A}_i = \str{\tree{t}_{\ge b_i}}$
     for $i \in \{1, \ldots, \rho(\sigma)\}$. Observe that since
     $\mc{T}$ is closed under rooted subtrees and since $\mathsf{Str}$
     is computable, $\mf{A}_i$ can be computed too. Finally,
     $\enumcompfun(m)$ checks if the $\lequiv{m}$ class of $\mf{A}_i$
     is $\delta_i$ -- this is done by checking if the formula
     $\varphi$ representing $\delta_i$ in $\listofmclasses$, is true
     in $\mf{A}_i$. (Checking if an $\mc{L}$ sentence is true in a
     finite structure is decidable.) If the tree $\tree{t}$ above
     passes this last check, then $\enumcompfun(m)$ stores $\tree{t}$
     as $\tree{s}_{\sigma, \delta_1, \ldots, \delta_{\rho(\sigma)}}$.
     If none of the trees enumerated by $\enumcompfun(m)$ pass the
     last check, then $\enumcompfun(m)$ stores $\mathsf{null}$ for
     $\sigma, \delta_1, \ldots, \delta_{\rho(\sigma)}$.

     \item the case of
     $\tree{v}_{\sigma, \delta_1, \ldots, \delta_{\rho(\sigma)}}$:
     Here $\sigma \in \sigmaint \setminus \sigmarank$. Let $d
     = \rho(\sigma)$ and let $\tree{t} = \tree{x} \odot \tree{y}$
     where $\tree{x}$, resp. $\tree{y}$, is the subtree of $\tree{t}$
     obtained by deleting the subtrees rooted at $b_{n - d +
     2}, \ldots, b_n$, resp. $b_1, \ldots, b_{n - d +1}$. Then
     $\enumcompfun(m)$ computes $\mf{A}_1 = \str{\tree{x}}$ and
     $\mf{A}_i = \str{\tree{t}_{\ge b_{n - d + i}}}$ for
     $i \in \{2, \ldots, d\}$.  Observe once again that $\mf{A}_i$ can
     be computed for each $i \in \{1, \ldots, d\}$. Finally,
     $\enumcompfun(m)$ checks if the $\lequiv{m}$ class of $\mf{A}_i$
     is $\delta_i$. If the tree $\tree{t}$ passes this last check,
     then $\enumcompfun(m)$ stores $\tree{t}$ as
     $\tree{v}_{\sigma, \delta_1, \ldots, \delta_{\rho(\sigma)}}$. Again
     if none of the trees enumerated by $\enumcompfun(m)$ pass the
     last check, then $\enumcompfun(m)$ stores $\mathsf{null}$ for
     $\sigma, \delta_1, \ldots, \delta_{\rho(\sigma)}$.  \end{enumerate}

   In the above cases, it is clear by
   Lemma~\ref{lemma:small-tree-given-input-equivalence-classes}, that
   if $\enumcompfun(m)$ stores $\mathsf{null}$ for $\sigma, \delta_1, \ldots,
   \delta_{\rho(\sigma)}$, then there is no tree in $\mc{T}$ that can
   be taken as $\tree{s}_{\sigma, \delta_1, \ldots,
     \delta_{\rho(\sigma)}}$, resp.  $\tree{v}_{\sigma, \delta_1,
     \ldots, \delta_{\rho(\sigma)}}$.

  \vspace{2pt} \und{Stage III}: In this stage, the trees identified in
  the previous stage are used to define functions $g_{\sigma, m}$ if
  $\sigma \in \sigmarank$ and $g_{\sigma, m, i}$ if
  $\sigma \in \sigmaint \setminus \sigmarank$, that satisfy the
  composition properties mentioned in
  Definition~\ref{definition:L-good-map}, whereby these resp. can
  indeed be considered as the functions $f_{\sigma, m}$ and
  $f_{\sigma, m,  i}$ as mentioned in
  Definition~\ref{definition:L-good-map}. We show how to define
  $g_{\sigma, m}$ for $\sigma \in \sigmarank$ using
  $\tree{s}_{\sigma, \delta_1, \ldots, \delta_{\rho(\sigma)}}$ (if
  identified); analogously, for
  $\sigma \in \sigmaint \setminus \sigmarank$, the function
  $g_{m, \sigma, \rho(\sigma)}$ is defined using
  $\tree{v}_{\sigma, \delta_1, \ldots, \delta_{\rho(\sigma)}}$ and
  function $g_{\sigma, m, i}$ is defined using
  $\tree{u}_{\sigma, \delta_1, \ldots, \delta_i}$ for
  $i \in \{1, \ldots, \rho(\sigma) - 1\}$.

  Let $\sigma \in \sigmarank$ and $\delta_1, \ldots,
  \delta_{\rho(\sigma)} \in \listofmclasses$.
  \begin{itemize}

   \item If no tree $\tree{z}$ of the form
      $\tree{s}_{\sigma, \delta_1, \ldots, \delta_{\rho(\sigma)}}$ is
      identified in the previous stage \linebreak
      (i.e. $\enumcompfun(m)$ stores $\mathsf{null}$ for
      $\sigma, \delta_1, \ldots, \delta_{\rho(\sigma)}$), then
      define\linebreak $g_{\sigma,
      m}(\delta_1, \ldots, \delta_{\rho(\sigma)})
      = \delta_{\text{default}}$ where $\delta_{\text{default}}$ is
      some fixed chosen element of $\listofmclasses$.

   \item Else, let
      $\tree{z}
      = \tree{s}_{\sigma, \delta_1, \ldots, \delta_{\rho(\sigma)}}$. Identify
      $\varphi \in \listofmclasses$ such that
      $\str{\tree{z}} \models \varphi$. Let $\delta$ be the
      equivalence class represented by $\varphi$. Then define
      $g_{ \sigma, m}(\delta_1, \ldots, \delta_{\rho(\sigma)})
      = \delta$.

  \end{itemize}
  
  Observe that since $\mathsf{Str}$ is assumed to be computable and
  since model checking an $\mc{L}$ sentence on a finite structure is
  decidable, $g_{\sigma, m}$ indeed gets generated after a finite
  amount of time.  Analogously, the functions $g_{\sigma, m, i}$ also
  get generated after a finite amount of time. It is easily seen from
  the above description of $\enumcompfun(m)$, that for some computable
  function $\eta_3: \mathbb{N} \rightarrow \mathbb{N}$, the total time
  taken by $\enumcompfun(m)$ is bounded by $\eta_3(m)$.

  \vspace{3pt} We finally show that $g_{\sigma, m}$ and $g_{\sigma, m, i}$
  constructed above indeed satisfy the composition properties of
  Definition~\ref{definition:L-good-map}.

  Let $\tree{t} = (O, \lambda) \in \mc{T}$ and $a$ be an internal node
  of $\tree{t}$ such that $\lambda(a) = \sigma$ and the children of
  $a$ in $\tree{t}$ are $b_1, \ldots, b_n$. Let $\delta_i$ be the
  $\lequiv{m}$ class of $\mathsf{Str}(\tree{t}_{\ge b_i})$ for
  $i \in \{1, \ldots, n\}$.

  \begin{enumerate}

  \item Suppose $\sigma \in \sigmarank$, whereby $n
  = \rho(\sigma)$. Then $\tree{t}_{\ge a} \in \mc{T}$ since $\mc{T}$
  is representation-feasible. Consider the tree
  $\tree{s}_{\sigma, \delta_1, \ldots, \delta_{\rho(\sigma)}}$ that
  then is guaranteed to be generated by $\enumcompfun(m)$ in Stage II
  because of
  Lemma~\ref{lemma:small-tree-given-input-equivalence-classes}. By the
  composition property as mentioned in
  Definition~\ref{definition:L-good-map}, it follows that
  $\str{\tree{t}_{\ge
  a}} \lequiv{m} \str{\tree{s}_{\sigma, \delta_1, \ldots, \delta_{\rho(\sigma)}}}$,
  i.e., the $\lequiv{m}$ classes of $\str{\tree{t}_{\ge a}}$ and
  $\str{\tree{s}_{\sigma, \delta_1, \ldots, \delta_{\rho(\sigma)}}}$
  are the same. Indeed, then the $\lequiv{m}$ class of
  $\str{\tree{t}_{\ge a}}$ is $g_{\sigma,
  m}(\delta_1, \ldots, \delta_n)$, because the $\lequiv{m}$ class of
  $\str{\tree{s}_{\sigma, \delta_1, \ldots, \delta_{\rho(\sigma)}}}$
  is $g_{\sigma, m}(\delta_1, \ldots, \delta_n)$ by construction.

  \item Suppose $\sigma \in \sigmaint \setminus \sigmarank$ and $n
  < \rho(\sigma)$. By similar reasoning as above, the tree
  $\tree{u}_{\sigma, \delta_1, \ldots, \delta_n}$ (that is guaranteed
  to be generated by $\enumcompfun(m)$) is such that
  $\str{\tree{t}_{\ge
  a}} \lequiv{m} \str{\tree{u}_{\sigma, \delta_1, \ldots, \delta_n}}$. Whereby,
  the $\lequiv{m}$ class of $\str{\tree{t}_{\ge a}}$ is indeed
  $g_{m, \sigma, n}(\delta_1, \ldots, \delta_n)$.

  \item Suppose $\sigma \in \sigmaint \setminus \sigmarank$ and
  $n \ge \rho(\sigma)$. Let $d = \rho(\sigma)$ and $n = r + q \cdot (d
  - 1)$ where $1 \leq r < d$ and $q > 0$. Consider the trees
  $\tree{z}_{1, k}$ obtained from $\tree{t}_{\ge a}$ by deleting the
  subtrees of $\tree{t}_{\ge a}$ rooted at $b_{k+1}, \ldots, b_n$, for
  $k \in I = \{r + j \cdot (d - 1) \mid 0 \leq j \leq q\}$ (whereby,
  $\tree{t}_{\ge a} = \tree{z}_{1, n}$).  Since $\mc{T}$ is
  representation-feasible, $\tree{z}_{1, k} \in \mc{T}$ for each
  $k$. Let $\chi_k$ be the $\lequiv{m}$ class of $\str{\tree{z}_{1,
  k}}$.  Using
  Lemma~\ref{lemma:small-tree-given-input-equivalence-classes}, it is
  guaranteed that in Stage II, $\enumcompfun(m)$ produces the trees
  $\tree{u}_{\sigma, \delta_1, \ldots, \delta_r}$ and
  $\tree{v}_{\sigma, \chi_k, \delta_{k+1}, \ldots, \delta_{k + (d
  -1)}}$ for each $k \in I$.

    By the composition property of
    Definition~\ref{definition:L-good-map}, we see that for $k \in
    I \setminus \{n\}$, we have

    \begin{center}
    \begin{tabular}{ccl}
    $\str{\tree{z}_{1, r}} $ & $ \lequiv{m}$ &$\str{\tree{u}_{\sigma, \delta_1,
        \ldots, \delta_r}}$\\
    $\str{\tree{z}_{1, k + (d-1)}}$ & $\lequiv{m}$ & $\str{\tree{v}_{\sigma,
        \chi_k, \delta_{k+1}, \ldots, \delta_{k + (d - 1)}}}$ 
    \end{tabular}
    \end{center}

    Whereby, from the very constructions of $g_{\sigma, m, i}$ for
    $i \in \{1, \ldots, \rho(\sigma)\}$, we get for $k \in
    I \setminus \{n\}$, that

    \begin{center}
    \begin{tabular}{lclcl}
    $\lequiv{m}$ class of $\str{\tree{z}_{1, r}}$ & = & $\chi_r$ & = &
      $g_{m, \sigma, r}(\delta_1, \ldots, \delta_r)$\\
      
    $\lequiv{m}$ class of $\str{\tree{z}_{1, k + (d-1)}}$ & = &
      $\chi_{k + (d -1)}$ & = & $g_{m, \sigma, d}(\chi_k, \delta_{k+1},
      \ldots, \delta_{k + (d -1)})$
    \end{tabular}
    \end{center}

    Putting $k = r + (q -1) \cdot (d- 1)$ above, we see that $\chi_n$,
    which is the $\lequiv{m}$ class of $\str{\tree{z}_{1, n}}
    (= \str{\tree{t}_{\ge a}})$, is indeed given by $g_{\sigma, m,
    d}(g_{\sigma, m, d}(\ldots g_{\sigma, m, d}(g_{\sigma, m,
    d}(\widehat{\delta}, \delta_{r+1}, $ $ \ldots, \delta_{r+
    (d-1)}), \delta_{r+d}, \ldots, \delta_{r+2 \cdot (d-1)}) \ldots), \delta_{n
    - d + 2}, \ldots, \delta_n)$, where $\widehat{\delta} = g_{\sigma,
    m, r}(\delta_1, \ldots, \delta_r)$.

    \end{enumerate}
\end{proof}

\newcommand{\compdegred}{\ensuremath{\mathsf{Complete}\text{-}\mathsf{degree}\text{-}\mathsf{reduction}}}
\newcommand{\listofmoneclasses}{\ensuremath{\mathsf{\mc{L}[m_1]\text{-}\mathsf{classes}}}}

\subsubsection{Proof of Lemma~\ref{lemma:core-lemma-for-fpt-over-trees-abstract}(\ref{lemma:core-lemma-for-fpt-reduce})}\label{subsubsection:computational-reductions}

\begin{proof}
  (1) \und{$\reducedeg(\tree{t}, m)$}:\\

  Suppose $\tree{t} \in \mc{T}$ and $m \in \mathbb{N}$ are given as
  inputs. Let $m_0$ be a witness to the composition property of
  $\mathsf{Str}$, as mentioned in
  Definition~\ref{definition:L-good-map}, and let $m_1 =
  \text{max}\{m_0, m\}$.  The algorithm $\reducedeg(\tree{t}, m)$
  functions in various stages as described below.

  \vspace{3pt}Stage I: $\reducedeg(\tree{t}, m)$ first invokes the
  procedure $\enumcompfun(m_1)$. The latter procedure produces the
  following:
  \begin{enumerate}
  \item $\listofmoneclasses$ which is a list of $\mc{L}[m_1]$
    sentences that represent all and exactly the equivalence classes
    of the $\lequiv{m_1}$ relation over $\cl{S}$.
  \item the ``composition'' functions $f_{\sigma, m_1}$ and
    $f_{\sigma, m_1, i}$ which satisfy the composition properties
    mentioned in Definition~\ref{definition:L-good-map}.
  \end{enumerate}

  The time taken to complete this step is at most $\eta_3(m_1)$, where
  $\eta_3$ is as given by
  Lemma~\ref{lemma:core-lemma-for-fpt-over-trees-abstract}(\ref{lemma:core-lemma-for-fpt-enum}).

  \vspace{3pt} Stage II: $\reducedeg(\tree{t}, m)$ now constructs
  bottom-up in $\tree{t}$, the function $\colour: \tree{t} \rightarrow
  \ldelta{\cl{S}}{m}$ such that $\colour(a)$ is the $\lequiv{m_1}$
  class of $\str{\tree{t}_{\ge a}}$. This is done inductively as
  follows:
  \begin{itemize}
   \item Base case: We first compute $\colour(e)$ for each leaf node
     $e$ of $\tree{t}$. This can be done in constant time as explained
     below.

     Since $\sigmaleaf$ is finite and since $\mathsf{Str}$ is
     isomorphism preserving (see
     Definition~\ref{definition:L-good-map}), there is a finite
     function $\mathsf{leaf\text{-}structures}: \sigmaleaf \rightarrow
     \cl{S}$ such that for any leaf node $e$ of $\tree{t}$, if its
     label is $\sigma$, then $\str{\tree{t}_{\ge e}} \cong
     \mathsf{leaf\text{-}structures}(\sigma)$. Further, since the
     range of $\mathsf{leaf\text{-}structures}$ is finite, there
     exists a finite function $\mathsf{leaf\text{-}colour}:
     \text{Range}(\mathsf{leaf\text{-}structures}) \rightarrow
     \ldelta{\cl{S}}{m_1}$ such that for each $\mf{A}$ in the range of
     $\mathsf{leaf\text{-}structures}$, we have
     $\mathsf{leaf\text{-}colour}(\mf{A})$ is the $\lequiv{m_1}$ class
     of $\mf{A}$. Whereby, given a leaf node $e$, we have $\colour(e)
     =
     \mathsf{leaf\text{-}colour}(\mathsf{leaf\text{-}structure}(\sigma))$,
     where $\sigma$ is the label of $e$.

    \item Induction step: Assume that for an internal node $a$, if
      $b_1, \ldots, b_n$ are the children of $a$ in $\tree{t}$, then
      $\colour(b_i)$ has been computed, for $i \in \{1, \ldots,
      n\}$. Let $\sigma$ be the label of $a$ in $\tree{t}$. We have
      two cases here to compute $\colour(a)$:
      \begin{itemize}
        \item $\sigma \in \sigmarank$: Then by the composition
          property of $\mathsf{Str}$, we have that $\colour(a) =
          f_{\sigma, m_1}(\colour(b_1), \ldots, \colour(b_n))$. Since
          $f_{\sigma, m_1}$ is a finite function, $\colour(a)$ can be
          computed in constant time.
        \item $\sigma \in \sigmaint \setminus \sigmarank$: Let $n = r
          + q \cdot (d -1)$ where $d = \rho(\sigma)$ and $1 \leq r <
          d$.  Let $\xi_0 = f_{\sigma, m_1, r}(\colour(b_1), \ldots,
          \colour(b_r))$, and $\xi_{i+1} = f_{\sigma, m_1, d}(\xi_i,
          \colour(b_{r + i \cdot (d-1) + 1}),$ $ \ldots,$ $
          \colour(b_{r + (i + 1) \cdot (d - 1)}))$ for $i \in \{0,
          \ldots, q-1\}$.  Then by the composition property of
          $\mathsf{Str}$, we have that $\colour(a) = \xi_1$ if $n <
          d$, else $\colour(a) = \xi_q$. Observe that the $\xi_i$s can
          be computed in constant time, whereby the time taken to
          compute $\colour(a)$ is linear in the degree of $a$ in
          $\tree{t}$.
      \end{itemize}
   \end{itemize}

   At the end of the above process, $\colour$ gets constructed. The
   time taken for this construction is linear in the sum of the
   degrees of the nodes of $\tree{t}$, and hence linear in
   $|\tree{t}|$.

   \vspace{3pt} Stage III: $\reducedeg(\tree{t}, m)$ finally invokes
   $\compdegred(\tree{t})$ below that reduces the degrees of the nodes
   of $\tree{t}$ to under a threshold. The output of\linebreak
   $\compdegred(\tree{t})$ is the output of $\reducedeg(\tree{t}, m)$.
   The former in turn uses the degree reduction procedure
   $\degred(\tree{u}, a)$ which takes in a tree $\tree{u}$ of $\mc{T}$
   and a node $a$ of $\tree{u}$, and produces a subtree $\tree{v}$ of
   $\tree{u}$ in $\mc{T}$, containing $a$, such that (i) the degree of
   $a$ in $\tree{v}$ is at most $p$, (ii) the roots of $\tree{v}$ and
   $\tree{u}$ are the same, (iii) $\str{\tree{v}} \hookrightarrow
   \str{\tree{u}}$ and (iv) $\str{\tree{v}} \lequiv{m_1}
   \str{\tree{u}}$.

   \vspace{4pt}\und{\compdegred(\tree{t}):}
   \begin{enumerate}
     \item Initialize $\tree{z} := \tree{t}$.
     \item For $a$ ranging over the nodes of $\tree{t}$, set
       $\tree{z}$ := $\degred(\tree{z}, a)$.
     \item Return $\tree{z}$.
   \end{enumerate}

   It is clear that $\compdegred(\tree{t})$, and hence
   $\reducedeg(\tree{t}, m)$, outputs the desired subtree $\tree{s}_1$
   as required by the statement of
   Lemma~\ref{lemma:abstract-tree-lemma} (observe that $\mc{L}[m_1]$
   eqivalence implies $\mc{L}[m]$ equivalence).  We now describe
   $\degred$ below and show that the time taken by $\degred(\tree{u},
   a)$ is linear in $\Lambda_{\cl{S}, \mc{L}}(m_1) ~\times$ \text{(the
     degree of $a$ in $\tree{u}$)}. Whereby the time taken by
   $\compdegred(\tree{t})$ is linear in $\Lambda_{\cl{S}, \mc{L}}(m_1)
   \times |\tree{t}|$. It follows then that there exists a computable
   function $\eta_4: \mathbb{N} \rightarrow \mathbb{N}$ such that the
   time taken by $\reducedeg(\tree{t}, m)$ to compute $\tree{s}_1$ is
   indeed at most $\eta_4(m) \times |\tree{t}|$.

   We now complete this part of the proof by describing $\degred$ and
   showing its running time to be as mentioned above.

   \vspace{4pt}\und{$\degred(\tree{u}, a)$}:

   \vspace{2pt}\begin{enumerate}[nosep]
     \item Let $\sigma$ be the label of $a$ in $\tree{u}$. If $a$ is a
       leaf node or if $\sigma \in \sigmarank$, then return
       $\tree{u}$.
     \item Else, for each $\delta \in \listofmoneclasses$, do the
       following:
       \begin{enumerate}[nosep]
         \item[(i)] Let $\tree{x} = \tree{u}_{\ge a}$. Let $a$ have $n$
           children in $\tree{x}$, call these $a_1, \ldots, a_n$ (in
           ascending order). Let $n = r + q \cdot (d - 1)$ where $1
           \leq r < d$ and $q > 0$. Let $I = \{r + l \cdot (d - 1)
           \mid 0 \leq l \leq q\}$.
         \item[(ii)] Let $\tree{x}_{1, k}$ denote the subtree of
           $\tree{x}$ obtained by deleting the subtrees rooted at
           $a_{k+1}, \ldots, a_n$. Construct the function $g: I
           \rightarrow \ldelta{\cl{S}}{m_1}$ such that $g(k)$ is the
           $\lequiv{m_1}$ class of $\str{\tree{x}_{1, k}}$ for $k \in
           I$.
         \item[(iii)] If $\delta$ is in the range of $g$, then let $i,
           j$ be resp. the least and greatest indices such that $g(i)
           = g(j) = \delta$. Let $\tree{y}$ be the subtree of
           $\tree{x}$ obtained by deleting the subtrees rooted at
           $a_{i+1}, \ldots, a_j$. Set $\tree{x}:= \tree{y}$.
       \end{enumerate}
     \item Let $\tree{v} = \tree{u}[\tree{u}_{\ge a} \mapsto
       \tree{x}]$. Return $\tree{v}$.
    \end{enumerate}

    \vspace{4pt}Reasoning similarly as in the proof of
    Lemma~\ref{lemma:abstract-tree-lemma}(\ref{lemma:abstract-tree-lemma-degree-reduction}),
    we can verify that \linebreak$\degred(\tree{u}, a)$ indeed works
    correctly. Given that we have already computed the function
    $\colour$ in Stage II, the time taken to compute $g$ is linear in
    the degree of $a$, whereby the time taken to reduce the degree of
    node $a$ in any iteration of the loop, is linear in the degree of
    $a$. Then, the total time taken by $\degred(\tree{u}, a)$ is then
    linear in $\Lambda_{\cl{S}, \mc{L}}(m_1) \times \text{(the degree
      of $a$ in $\tree{u}$)}$.

    (Important note: Observe the function $\colour$ restricted to the
    nodes of the output $\tree{v}$ of $\degred(\tree{u}, a)$ is such
    that for any node $a$ of $\tree{v}$, the $\lequiv{m_1}$ class of
    $\str{\tree{v}_{\ge a}}$ is indeed $\colour(a)$.)

\vspace{5pt} (2) \und{$\redheight(\tree{t}, m)$:}

\vspace{3pt} Just like $\reducedeg(\tree{t}, m)$, the algorithm
$\redheight(\tree{t}, m)$ also functions in various stages as
described below. Let as before, $m_1 = \text{max}(m_0, m)$.

\vspace{3pt} Step I: We generate $\listofmoneclasses$ and the function
$\colour$ as done in $\reducedeg(\tree{t}, m)$.

Step II: We construct a 2 dimensional array $\lowest[i][j]$ where $i$
ranges over the nodes of $\tree{t}$ and $j$ ranges over
$\listofmoneclasses$, such that $\lowest[i][j]$ stores a pointer to a
lowest (i.e. closest to a leaf) node $a$ in the subtree of $\tree{t}$
rooted at $i$, such that the $\lequiv{m_1}$ class of
$\str{\tree{t}_{\ge a}}$ is $j$. In other words, $a$ is such that the
$\lequiv{m_1}$ class of $\str{\tree{t}_{\ge a}}$ is $j$, and there is
no node $b \neq a$ in $\tree{t}_{\ge a}$ such that the $\lequiv{m_1}$
class of $\str{\tree{t}_{\ge b}}$ is $j$.

We construct $\lowest$ bottom-up in $\tree{t}$ as described below.
\begin{itemize}
  \item Base case: For a leaf node $e$, since the $\lequiv{m_1}$ class
    $\delta_e$ of $\str{\tree{t}_{\ge e}}$ has already been computed
    in Step I, we simply set $\lowest[e][\delta_e]$ to store a pointer
    to $e$, and for all $\delta \in \listofmoneclasses$ such that $\delta
    \neq \delta_e$, set $\lowest[e][\delta] = \mathsf{null}$.

    The time taken to complete this step is linear in the
    number of leaf nodes of $\tree{t}$.
    
  \item Induction: Assume that for an internal node $a$ under
    consideration, for all its children $b$ in $\tree{t}$, the value
    of $\lowest[b][\delta]$ has been computed for all
    $\delta \in \listofmoneclasses$. Let $\delta_a$ be the
    $\lequiv{m_1}$ class of $\str{\tree{t}_{\ge a}}$ (that has already
    been computed as $\colour(a)$ in Step I). For
    $\delta \in \listofmoneclasses$, check if for some child $b$ of
    $a$, the value of $\lowest[b][\delta]$ is not $\mathsf{null}$. If
    there is such a child $b$, set $\lowest[a][\delta]
    = \lowest[b][\delta]$. If there is no such child, then if $\delta
    = \delta_a$, then set $\lowest[a][\delta]$ to store a pointer to
    $a$, else set $\lowest[a][\delta] = \mathsf{null}$.

    Observe that $\lowest[a][\delta]$ indeed stores a pointer to a
    lowest node $c$ in $\tree{t}_{\ge a}$ such that the $\lequiv{m_1}$
    class of $\str{\tree{t}_{\ge c}}$ is $\delta$. Also observe that
    the time taken to complete this step is linear in the degree of
    $a$ in $\tree{t}$.
\end{itemize}

It is clear that eventually $\lowest$ gets constructed in time linear
in $|\tree{t}|$.

\vspace{3pt} Step III: We now describe an algorithm
$\rainbow(\tree{x})$ that takes a subtree $\tree{x}$ of $\tree{t}$ in
$\mc{T}$ as input and outputs a subtree $\tree{y}$ of $\tree{x}$ in
$\mc{T}$ such that
\begin{enumerate}
\item $\str{\tree{x}} \hookrightarrow \str{\tree{y}}$
\item $\str{\tree{x}} \lequiv{m_1} \str{\tree{y}}$
\item in no path from the root to the leaf of $\tree{y}$ is it the
  case that there exist two distinct nodes $a$ and $b$ such that the
  $\lequiv{m_1}$ classes of $\str{\tree{y}_{\ge a}}$ and
  $\str{\tree{y}_{\ge b}}$ are the same. Whereby, the height of
  $\tree{x}$ is at most $\Lambda_{\cl{S}, \mc{L}}(m_1)$.
\item
   The input $\tree{x}$ and output $\tree{y}$ both satisfy the
   following ``colour preservation'' property $\mc{Q}(\cdot)$: for a
   subtree $\tree{s}$ of $\tree{t}$, obtained from $\tree{t}$ by
   removal of rooted subtrees and replacements with rooted subtrees,
   $\mc{Q}(\tree{s})$ is a predicate denoting that the function
   $\mathsf{Colour}$ computed for $\tree{t}$, when restricted to the
   nodes of $\tree{s}$, is such that for any node $a$ of $\tree{s}$,
   $\mathsf{Colour}(a)$ gives the $\lequiv{m}$ class of $\tree{s}_{\ge
   a}$.

\end{enumerate}

\und{$\rainbow(\tree{x})$:}
\begin{enumerate}
\item Let $\delta$ be the $\lequiv{m_1}$ class of $\troot{\tree{x}}$ (by the properties mentioned above,  $\delta = \colour(\troot{\tree{x}})$).
\item Let $a = \troot{\tree{x}}$.
\item If $\lowest[a][\delta]$ stores a pointer to some node $b$ other
  than $a$, then return $\rainbow(\tree{x}_{\ge b})$.
\item Else, let $b_1, \ldots, b_n$ be the children of $a$ in
  $\tree{x}$.
\item For $i \in \{1, \ldots, n\}$, let $\tree{y}_i =
  \rainbow(\tree{x}_{\ge b_i})$.
\item Let $\tree{z} = \tree{x}[\tree{x}_{\ge b_1} \mapsto
  \tree{y}_1]\ldots[\tree{x}_{\ge b_n} \mapsto \tree{y}_n]$ be the
  subtree of $\tree{x}$ obtained by replacing $\tree{x}_{\ge b_i}$
  with $\tree{y}_i$ for $i \in \{1, \ldots, n\}$.
\item Return $\tree{z}$.
\end{enumerate}

It is easy to check using the fact that $\mc{T}$ is closed under
replacements with rooted subtrees and
Lemma~\ref{lemma:helper-height-reduction} that $\rainbow(\tree{x})$ is
indeed correct. We also see that the number of ``top level'' recursive
calls made by $\rainbow(\tree{x})$ is linear in the degree of
$\troot{\tree{x}}$, whereby the total time taken by
$\rainbow(\tree{x})$ is indeed linear in $|\tree{x}|$.

\vspace{3pt}
Having defined $\rainbow(\tree{x})$, we describe Step III which
consists of executing the following substeps.

\begin{enumerate}
\item Let $b_1, \ldots, b_n$ be the children of $\troot{\tree{t}}$ in
  $\tree{t}$.
\item For $i \in \{1, \ldots, n\}$, let $\tree{u}_i =
  \rainbow(\tree{t}_{\ge b_i})$.
\item Let $\tree{v} = \tree{t}[\tree{t}_{\ge b_1} \mapsto
  \tree{u}_1]\ldots[\tree{t}_{\ge b_n} \mapsto \tree{u}_n]$ be the
  subtree of $\tree{t}$ obtained by replacing $\tree{t}_{\ge b_i}$
  with $\tree{u}_i$ for $i \in \{1, \ldots, n\}$. Return $\tree{v}$.
\end{enumerate}

Reasoning similarly as in the proof of
 Lemma~\ref{lemma:abstract-tree-lemma}(\ref{lemma:abstract-tree-lemma-height-reduction}),
 we observe that $\tree{v}$ above is indeed the desired subtree
 $\tree{s}_2$ of $\tree{t}$.  It is easy to see from the descriptions
 above that the time taken by $\redheight(\tree{t}, m)$ to compute
 $\tree{s}_2$ is at most $\eta_5(m) \times |\tree{t}|$, for some
 computable function $\eta_5: \mathbb{N} \rightarrow \mathbb{N}$.
\end{proof}

\section{Applications to various concrete settings}\label{section:classes-satisfying-lebsp} 

\newcommand{\unorderedtrees}{\mathsf{Unordered}\text{-}\mathsf{trees}}
\newcommand{\orderedtrees}{\mathsf{Ordered}\text{-}\mathsf{trees}}
\newcommand{\orderedrankedtrees}{\mathsf{Ordered}\text{-}\mathsf{ranked}\text{-}\mathsf{trees}}

\subsection{Regular languages of words, trees (unordered, ordered, ranked
  or partially ranked) and nested words}

Let $\Sigma$ be a finite alphabet. The notion of unordered, ordered,
ranked and partially ranked $\Sigma$-trees was already introduced in
Section~\ref{section:partially-ranked-trees}.  A $\Sigma$-tree whose
underlying poset is a linear order is called a \emph{$\Sigma$-word}.
A \emph{nested word over $\Sigma$} is a pair $(w, \leadsto)$ where $w$
is a ``usual'' $\Sigma$-word and $\leadsto$ is a binary relation
representing a ``nested matching''. Formally, if $(A, \leq)$ is the
linear order underlying $w$, then $\leadsto$ satisfies the following:
(i) for $i, j \in A$, if $i \leadsto j$, then $i \leq j$ and $i \neq
j$ (ii) for $i \in A$, there is at most one $j \in A$ such that $i
\leadsto j$ and at most one $l \in A$ such that $l \leadsto i$, and
(iii) for $i_1, i_2, j_1, j_2 \in A$, if $i_1 \leadsto j_1$ and $i_2
\leadsto j_2$, then it is not the case that $i_1 < i_2 \leq j_1 <
j_2$. (Nested words here correspond to nested words
of~\cite{alur-madhu}, that have no pending calls or pending returns.)
\begin{figure}[H]
\begin{minipage}{0.45\textwidth}
For e.g., $\nesword{w} = (abaabba, \{(2, 6), (4, 5)\})$ is a nested
word over $\{a, b\}$.  A nested $\Sigma$-word has a natural
representation using a representation-feasible tree over $\sigmaint
\cup \sigmaleaf$, where $\sigmaleaf = \Sigma \cup (\Sigma \times
\Sigma)$, and $\sigmaint = \sigmaleaf \cup \{\circ\}$. The figure
alongside shows the tree $\tree{t}$ for $\nesword{w}$.  Conversely,
every representation-feasible tree over $(\sigmaint \cup \sigmaleaf)$
represents a nested $\Sigma$-word. 
\end{minipage}\hfill
\begin{minipage}{0.5\textwidth}

\flushright
\includegraphics[width=7cm,height=3.2cm]{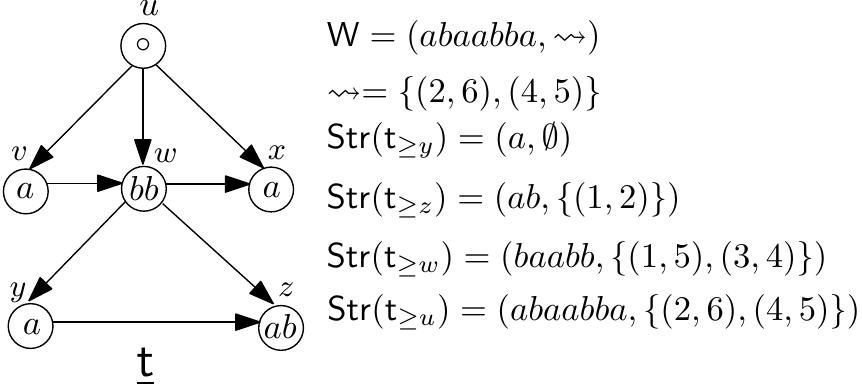}
\caption{\label{nested-word-tree-rep}Nested word as a tree}

\end{minipage} \hfill
\end{figure}

The notion of \emph{regular languages} of words, trees and nested
words is well studied. Since this notion corresponds to MSO
definability~\cite{tata,alur-madhu}, we say a class of words, trees
(of any of the aforesaid kinds) or nested words is regular if it
is the class of models of an MSO sentence.

\begin{theorem}\label{theorem:words-and-trees-and-nested-words-satisfy-lebsp}
Given finite alphabets $\Sigma, \Omega$ such that $\Omega \subseteq
\Sigma$, and a function $\rho: \Omega \rightarrow \mathbb{N}$, let
${\words}(\Sigma)$, ${\unorderedtrees}(\Sigma)$,
${\orderedtrees}(\Sigma)$, $\prankedtrees(\Sigma, \Omega, \rho)$ and
$\nestedwords{\Sigma}$ denote resp. the classes of all $\Sigma$-words,
all unordered $\Sigma$-trees, all ordered $\Sigma$-trees, all ordered
$\Sigma$-trees partially ranked by $(\Omega, \rho)$, and all nested
$\Sigma$-words.  Let $\cl{S}$ be a regular subclass of any of these
classes.  Then $\lebsp{\cl{S}}$ holds with a computable witness
function. Further, any witness function for $\lebsp{\cl{S}}$ is
necessarily non-elementary.
\end{theorem}

To present the proof idea for the above result, we need two
composition lemmas, one for unordered trees and the other for nested
words, just as we needed the composition lemma for the proof of
Proposition~\ref{prop:result-for-partially-ranked-trees}.

Towards the composition lemma for unordered trees, we introduce
terminology akin to that introduced for ordered trees in
Section~\ref{section:partially-ranked-trees}. Given unordered trees
$\tree{t}$ and $\tree{s}$, and a node $a$ of $\tree{t}$, define the
\emph{join of $\tree{s}$ to $\tree{t}$ at $a$}, denoted $\tree{t}
\cdot_a \tree{s}$, as follows: Let $\tree{s}'$ be an isomorphic copy
of $\tree{s}$ whose set of nodes is disjoint with the set of nodes of
$\tree{t}$. Then $\tree{t} \cdot_a \tree{s}$ is defined upto
isomorphism as the tree obtained by making $\tree{s}'$ as a new child
subtree of $a$ in $\tree{t}$. The composition lemma for unordered
trees is now as stated below. The proof is similar to that of
Lemma~\ref{lemma:mso-composition-lemma-for-ordered-trees}, and is
hence skipped.

\begin{lemma}[Composition lemma for unordered trees]\label{lemma:mso-composition-lemma-for-unordered-trees}
For a finite alphabet $\Omega$, let ${\tree{t}}_i, \tree{s}_i$ be
non-empty unordered $\Omega$-trees, and let $a_i$ be a node of
$\tree{t}_i$, for each $i \in \{1, 2\}$. For $m \in \mathbb{N}$,
suppose that $({\tree{t}}_1, a_1) \lequiv{m} ({\tree{t}}_2, a_2)$ and
${\tree{s}}_1 \lequiv{m} {\tree{s}}_2$. Then
$(({\tree{t}}_1 \cdot_{a_1} {\tree{s}}_1), a_1) \lequiv{m}
(({\tree{t}}_2 \cdot_{a_2} {\tree{s}}_2), a_2)$.
\end{lemma}

Towards the composition lemma for nested words, we first define the
notion of \emph{insert of a nested word $\nesword{v}$ in a nested word
$\nesword{u}$ at a given position $e$ of $\nesword{u}$}.

\begin{definition}[Insert]
Let $\nesword{u} =
(A_{\nesword{u}}, \leq_{\nesword{u}}, \lambda_{\nesword{u}}, \leadsto_{\nesword{u}})$
and $\nesword{v} =
(A_{\nesword{v}}, \leq_{\nesword{v}}, \lambda_{\nesword{v}}, \leadsto_{\nesword{v}})$
be given nested $\Sigma$-words, and let $e$ be a position in
$\nesword{u}$. The \emph{insert of $\nesword{v}$ in $\nesword{u}$ at
$e$}, denoted $\nesword{u}\uparrow_{e}\nesword{v}$, is a nested
$\Sigma$-word defined as below.
\begin{enumerate}[nosep]
\item 
If $\nesword{u}$ and $\nesword{v}$ have disjoint sets of positions,
then $\nesword{u}\uparrow_{e}\nesword{v} =
(A, \leq, \lambda, \leadsto)$ where
\begin{itemize}[nosep]
\item $A = A_{\nesword{u}} \sqcup A_{\nesword{v}}$
\item $\leq \,=\, \leq_{\nesword{u}} \cup \leq_{\nesword{v}}$ $\cup
  \,\{(i, j) \mid i \in A_{\nesword{u}}, j \in A_{\nesword{v}}, i
  \leq_{\nesword{u}} e\}$ $\cup\, \{(j, i) \mid i \in A_{\nesword{u}},
  j \in A_{\nesword{v}}, e \leq_{\nesword{u}} i, e \neq i\}$
\item $\lambda(a) = \lambda_{\nesword{u}}(a)$ if $a \in
  A_{\nesword{u}}$, else $\lambda(a) = \lambda_{\nesword{v}}(a)$
\item $\leadsto \,=\, \leadsto_{\nesword{u}} \cup \leadsto_{\nesword{v}}$
\end{itemize}
\item 
If $\nesword{u}$ and $\nesword{v}$ have overlapping sets of positions,
then let $\nesword{v}_1$ be an isomorphic copy of $\nesword{v}$ whose
set of positions is disjoint with that of $\nesword{u}$. Then
$\nesword{u}\uparrow_e\nesword{v}$ is defined upto isomorphism as
$\nesword{u}\uparrow_e\nesword{v}_1$.
\end{enumerate}
In the special case that $e$ is the last (under $\leq_{\nesword{u}}$)
position of $\nesword{u}$, we denote
$\nesword{u} \uparrow_e \nesword{v}$ as
$\nesword{u} \cdot \nesword{v}$, and call the latter as
the \emph{concatenation of $\nesword{v}$ with $\nesword{u}$}.
\end{definition}

\begin{lemma}[Composition lemma for nested words]\label{lemma:nested-words-composition-lemma}
For a finite alphabet $\Sigma$, let
$\nesword{u}_i, \nesword{v}_i \in \nestedwords{\Sigma}$, and let $e_i$
be a position in $\nesword{u}_i$ for $i \in \{1, 2\}$. Then the
following hold for each $m \in \mathbb{N}$.
\begin{enumerate}[nosep]
\item If $(\nesword{u}_1, e_1) \lequiv{m} (\nesword{u}_2, e_2)$ and
  $\nesword{v}_1 \lequiv{m} \nesword{v}_2$, then
  $(\nesword{u}_1\uparrow_{e_1}\nesword{v}_1) \lequiv{m}
  (\nesword{u}_2\uparrow_{e_2}\nesword{v}_2)$.
\item 
$\nesword{u}_1 \lequiv{m} \nesword{u}_2$ and
  $\nesword{v}_1 \lequiv{m} \nesword{v}_2$, then
  $\nesword{u}_1\cdot\nesword{v}_1 \lequiv{m} \nesword{u}_2\cdot\nesword{v}_2$.
\end{enumerate}
\end{lemma}

\begin{proof}
We give the proof for $\mc{L} = $MSO. The proof for $\mc{L} = $FO is
similar.  

The winning strategy $S$ for the duplicator in the $m$-round $\mef$
game between $\nesword{u}_1\uparrow_{e_1}\nesword{v}_1$ and
$\nesword{u}_2\uparrow_{e_2}\nesword{v}_2$ is simply the composition
of the winning strategies $S_1$, resp. $S_2$, of the duplicator in the
$m$-round $\mef$ game between $(\nesword{u}_1, e_1)$ and
$(\nesword{u}_2, e_2)$, resp.  $\nesword{v}_1$ and
$\nesword{v}_2$. Formally, $S$ is defined as follows. 
\begin{enumerate}[nosep]
\item Point move: If the spoiler picks an element of $\nesword{u}_1$,
resp. $\nesword{v}_1$, from
$\nesword{u}_1\uparrow_{e_1}\nesword{v}_1$, then the duplicator picks
the element of $\nesword{u}_2$, resp. $\nesword{v}_2$, from
$\nesword{u}_2\uparrow_{e_2}\nesword{v}_2$, that is given by the
strategy $S_1$, resp. $S_2$.  A similar choice of an element from
$\nesword{u}_1\uparrow_{e_1}\nesword{v}_1$ is made by the duplicator
if the spoiler picks an element from
$\nesword{u}_2\uparrow_{e_2}\nesword{v}_2$.
\item Set move: If the spoiler picks a set $Z$ from
$\nesword{u}_1\uparrow_{e_1}\nesword{v}_1$, then let $Z = X \sqcup Y$
where $X$ is a subset of positions of $\nesword{u}_1$ and $Y$ is a
subset of positions of $\nesword{v}_1$. Then the duplicator picks the
set $Z'$ from $\nesword{u}_2\uparrow_{e_2}\nesword{v}_2$ where $Z' =
X' \sqcup Y'$, $X'$ is the subset of positions of $\nesword{u}_2$ that
is chosen by the duplicator in response to $X$ according to strategy
$S_1$, and $Y'$ is the subset of positions of $\nesword{v}_2$ that is
chosen by the duplicator in response to $Y$ according to strategy
$S_2$.  A similar choice of a set from
$\nesword{u}_1\uparrow_{e_1}\nesword{v}_1$ is made by the duplicator
if the spoiler picks a set from
$\nesword{u}_2\uparrow_{e_2}\nesword{v}_2$.  
\end{enumerate}
It is easy to see that $S$ is a winning strategy in the $\mef$ game
between $\nesword{u}_1\uparrow_{e_1}\nesword{v}_1$ and
$\nesword{u}_2\uparrow_{e_2}\nesword{v}_2$.
\end{proof}

\begin{proof}[Proof idea for Theorem~\ref{theorem:words-and-trees-and-nested-words-satisfy-lebsp}]
We first show $\mebsp{\cl{S}}$ holds when $\cl{S}$ is exactly one of
the classes mentioned in the statement of the theorem. That
$\lebsp{\cdot}$ holds for a regular subclass follows, because (i)
$\mebsp{\cdot}$ is preserved under $\mso$ definable subclasses, and
(ii) $\mebsp{\cdot}$ implies $\febsp{\cdot}$.

Consider $\cl{S} = {\unorderedtrees}(\Sigma)$ (the other cases of
trees have been covered by
Proposition~\ref{prop:result-for-partially-ranked-trees}).  There is a
natural map $\mathsf{Str}_1 : \mc{T}_1 \rightarrow \cl{S}$, where
$\mc{T}_1$ is the class of all $(\sigmaint \cup \sigmaleaf)$-trees
that is representation-feasible for $(\sigmarank, \rho)$, $\sigmaint =
\sigmaleaf = \Sigma$, $\sigmarank = \emptyset, \rho$ is the constant
function of value 2, and $\mathsf{Str}_1$ ``forgets'' the ordering
among the children of any node of its input tree. The $\mso$
composition lemma for unordered trees, given by
Lemma~\ref{lemma:mso-composition-lemma-for-unordered-trees}, then
shows that $\mathsf{Str}_1$ is an elementary $\mso$-good tree
representation for $\cl{S}$, whereby using
Theorem~\ref{theorem:good-tree-rep-implies-lebsp-and-f.p.t.-algorithm},
we are done.

Likewise, when $\cl{S}=\nestedwords{\Sigma}$, there is a natural map
$\mathsf{Str}_2 : \mc{T}_2 \rightarrow \cl{S}$, where $\mc{T}_2$ is
the class of all $(\sigmaint \cup \sigmaleaf)$-trees that is
representation-feasible for $(\sigmarank, \rho)$, $\sigmaint =
\sigmaleaf \cup \{\circ\}, \sigmaleaf = \Sigma \cup (\Sigma \times
\Sigma)$, $\sigmarank = \emptyset$, $\rho$ is the constant function 2,
and $\mathsf{Str}_2$ is as described in the example above. Then
$\mathsf{Str}_2$ is an elementary $\mso$-good tree representation for
$\cl{S}$, due to the $\mso$ composition lemma for nested words given
by Lemma~\ref{lemma:nested-words-composition-lemma}.  We are done by
Theorem~\ref{theorem:good-tree-rep-implies-lebsp-and-f.p.t.-algorithm}
again. The non-elementariness of witness functions is due to
Theorem~\ref{theorem:good-tree-rep-implies-lebsp-and-f.p.t.-algorithm}
and the non-elementariness of the index of the $\lequiv{m}$ relation
over words~\cite{frick-grohe}.
\end{proof}

\newcommand{\labelednpartitecographs}{\ensuremath{\mathsf{Labeled}\text{-}\mathsf{n}\text{-}\mathsf{partite}\text{-}\mathsf{cographs}}}

\subsection{$n$-partite cographs}

The class of $n$-partite cographs, introduced in~\cite{shrub-depth},
can be defined upto isomorphism as the range of the map $\mathsf{Str}$
described as follows.  Let $\sigmaleaf = \left[n\right] = \{1, \ldots,
n\}$ and $\sigmaint = \{ f \mid f:\left[n\right] \times \left[n\right]
\rightarrow \{0, 1\}\}, \sigmarank = \emptyset$ and $\rho: \sigmarank
\rightarrow \mathbb{N}_{+}$ be the constant function 2.  Let $\mc{T}$
be the class of all $(\sigmaint \cup \sigmaleaf)$-trees that are
representation-feasible for $(\sigmarank, \rho)$. Consider
$\mathsf{Str}: \mc{T} \rightarrow \mathsf{Graphs}$ be defined as
follows, where $\mathsf{Graphs}$ is the class of all undirected
graphs: For $\tree{t} = (O, \lambda) \in \mc{T}$ where $O = ((A,
\leq), \lesssim)$ is an ordered unlabeled tree and $\lambda$ is the
labeling function, $\str{\tree{t}} = G = (V, E)$ is such that (i) $V$
is exactly the set of leaf nodes of $\tree{t}$ (ii) for $a, b \in V$,
if $c = a \wedge b$ is the greatest common ancestor (under $\leq$) of
$a$ and $b$ in $\tree{t}$, then $\{a, b\} \in E$ iff
$\lambda(c)(\lambda(a), \lambda(b)) = 1$.  We now have the following
result. Below, a \emph{$\Sigma$-labeled} $n$-partite cograph is a pair
$(G, \nu)$ where $G$ is an $n$-partite cograph and $\nu: V \rightarrow
\Sigma$ is a labeling function. Also, ``hereditary'' means ``closed
under substructures''.

\begin{theorem}\label{theorem:n-partite-cographs-satisfy-lebsp}
Given $n \in \mathbb{N}$ and a finite alphabet $\Sigma$, let
$\labelednpartitecographs(\Sigma)$ be the class of all
$\Sigma$-labeled $n$-partite cographs. Let $\cl{S}$ be any hereditary
subclass of this class. Then $\lebsp{\cl{S}}$ holds with a computable
witness function.  Whereby, each of the graph classes below satisfies
$\lebsp{\cdot}$ with a computable witness function. Further, the
classes with bounded parameters as mentioned below have elementary
functions witnessing $\lebsp{\cdot}$.
\begin{enumerate}[nosep]
\item Any hereditary class of $n$-partite cographs, for each $n \in
  \mathbb{N}$.
\item Any hereditary class of graphs of bounded shrub-depth.
\item Any hereditary class of graphs of bounded $\mc{SC}$-depth.
\item Any hereditary class of graphs of bounded tree-depth.
\item Any hereditary class of co-graphs.
\end{enumerate}
\end{theorem}

The proof of Theorem~\ref{theorem:n-partite-cographs-satisfy-lebsp}
again uses a composition lemma for $n$-partite cographs.  Let $\mc{T}$
be the class of all $(\sigmaint \cup \sigmaleaf)$-trees where
$\sigmaleaf = \left[n\right] \times \Sigma$, $\sigmaint = \{f \mid f:
\left[n\right] \times \left[n\right] \rightarrow \{0, 1\}\}$ and
$\sigmarank = \emptyset$. Let $\rho: \sigmaint \rightarrow
\mathbb{N}_+$ be the constant function 2.  Then $\mc{T}$ is
representation-feasible for $(\sigmarank, \rho)$. Further, there is a
natural representation map $\mathsf{Str}: \mc{T} \rightarrow
\labelednpartitecographs(\Sigma)$ exactly of the kind described above
for $n$-partite cographs, that maps a tree in $\mc{T}$ to the labeled
$n$-partite graph that it represents. The composition lemma for
$n$-partite cographs can now be stated as below.

\newcommand{\V}[1]{\ensuremath{\mathsf{V}\text{-}\mathsf{Str}(#1)}}

\begin{lemma}[Composition lemma for $n$-partite cographs]\label{lemma:composition-lemma-for-n-partite-cographs}
For $i \in \{1, 2\}$, let $(G_i, \nu_{i, 1})$ and $(H_i, \nu_{i, 2})$
be graphs in $\labelednpartitecographs(\Sigma)$. Suppose $\tree{t}_i$
and $\tree{s}_i$ are trees of $\mc{T}$ such that $\str{\tree{t}_i} =
(G_i, \nu_{i, 1})$, $\str{\tree{s}_i} = (H_i, \nu_{i, 2})$, and the
labels of $\troot{\tree{t}_i}$ and $\troot{\tree{s}_i}$ are the
same. Let $\tree{z}_i = \tree{t}_i \odot \tree{s}_i$ and
$\str{\tree{z}_i} = (Z_i, \nu_{i})$ for $i \in \{1, 2\}$. For each $m
\in \mathbb{N}$, if $(G_1, \nu_{1, 1}) \lequiv{m} (G_2, \nu_{2, 1})$
and $(H_1, \nu_{1, 2}) \lequiv{m} (H_2, \nu_{2, 2})$, then $(Z_1,
\nu_{1}) \lequiv{m} (Z_2, \nu_{2}) $.
\end{lemma}

\begin{proof}

  We prove the lemma for $\mc{L} = \mso$. A similar proof can be done
  for $\mc{L} = \fo$.

  We can assume w.l.o.g. that $\tree{t}_i$ and $\tree{s}_i$ have
  disjoint sets of nodes for $i \in \{1, 2\}$.  Let the set of
  vertices of $\mathsf{Str}(\tree{t}_i)$ and
  $\mathsf{Str}(\tree{s}_i)$ be $\V{\tree{t}_i}$ and $\V{\tree{s}_i}$
  respectively. Then the vertex set $\V{\tree{z}_i}$ of
  $\mathsf{Str}(\tree{z}_i)$ is $\V{\tree{t}_i} \sqcup \V{\tree{s}_i}$
  for $i \in \{1, 2\}$.

Let $\tbf{S}_{\tree{t}}$, resp. $\tbf{S}_{\tree{s}}$, be the strategy
of the duplicator in the $m$-round $\mef$ game between
$\mathsf{Str}(\tree{t}_1)$ and $\mathsf{Str}(\tree{t}_2)$,
resp. between $\mathsf{Str}(\tree{s}_1)$ and $\mathsf{Str}(\tree{s}_2)$. For
the $m$-round $\mef$ game between $\mathsf{Str}({\tree{z}}_1)$ and
$\mathsf{Str}({\tree{z}}_2)$, the duplicator follows the following
strategy, call it $\tbf{R}$.
\begin{itemize}[nosep]
\item Point move: If the spoiler chooses a vertex from
  $\V{\tree{t}_1}$ (resp. $\V{\tree{t}_2}$), then the duplicator
  chooses a vertex from $\V{\tree{t}_2}$ (resp. $\V{\tree{t}_1}$)
  according to $\tbf{S}_{\tree{t}}$. Else, if the spoiler chooses a
  vertex from $\V{\tree{s}_1}$ (resp. $\V{\tree{s}_2}$), then the
  duplicator chooses a vertex from $\V{\tree{s}_2}$
  (resp. $\V{\tree{s}_1}$) according to $\tbf{S}_{\tree{s}}$.

\item Set move: If the spoiler chooses a set, say $U$, from
  $\V{\tree{z}_1}$ (resp. $\V{\tree{z}_2}$), then let $X = U \cap
  \V{\tree{t}_1}$ (resp. $X = U \cap \V{\tree{t}_2}$) and $Y = U \cap
  \V{\tree{s}_1}$ (resp. $Y = U \cap \V{\tree{s}_2}$). Let $X'$ be the
  subset of $\V{\tree{t}_2}$ (resp. $\V{\tree{t}_1}$) that is picked
  according to the strategy $\tbf{S}_{\tree{t}}$ in response to the
  choice of $X$ in $\V{\tree{t}_1}$
  (resp. $\V{\tree{t}_2}$). Likewise, let $Y'$ be the subset of
  $\V{\tree{s}_2}$ (resp. $\V{\tree{s}_1}$) that is picked according
  to $\tbf{S}_{\tree{s}}$ in response to the choice of $Y$ in
  $\V{\tree{s}_1}$ (resp. $\V{\tree{s}_2}$).  Then the set $U'$ picked
  by the duplicator from $\V{\tree{z}_2}$ according to strategy
  $\tbf{R}$ is given by $U' = X' \sqcup Y'$.
\end{itemize}
We now show that $\tbf{R}$ is a winning strategy for the duplicator.

Let at the end of $m$ rounds, the vertices and sets chosen from
$\mathsf{Str}({\tree{z}}_1)$, resp. $\mathsf{Str}({\tree{z}}_2)$, be
$a_1, \ldots, a_p$ and $A_1, \ldots, A_r$, resp.  $b_1, \ldots, b_p$
and $B_1, \ldots, B_r$, where $p + r = m$. Let $A_l^1 = A_l\cap
\V{\tree{t}_1}, A_l^2 = A_l \cap \V{\tree{s}_1}, B_l^1 = B_l \cap
\V{\tree{t}_2}$ and $B_l^2 = B_l \cap \V{\tree{s}_2}$ for $l \in \{1,
\ldots, r\}$.

It is easy to see that the labels of $a_i$ and $b_i$ are the same for
all $i \in \{1, \ldots, p\}$.  Also by the description of $\tbf{R}$
given above it is easy to check for all $i \in \{1, \ldots, p\}$ that
$a_i \in \V{\tree{t}_1}$ iff $b_i \in \V{\tree{t}_2}$ and $a_i \in
\V{\tree{s}_1}$ iff $b_i \in \V{\tree{s}_2}$. Likewise, for all $l \in
\{1, \ldots, r\}$ and $i \in \{1, \ldots, p\}$, we have $a_i \in
A_l^1$ iff $b_i \in B_l^1$ and $a_i \in A_l^2$ iff $b_i \in B_l^2$,
whereby $a_i \in A_l$ iff $b_i \in B_l$.

Consider $a_i, a_j$ for $i \neq j$ and $i, j \in \{1, \ldots, p\}$.
We show below that $a_i, a_j$ are adjacent in $\str{\tree{z}_1}$ iff
$b_i, b_j$ are adjacent in $\str{\tree{z}_2}$. This would show that
$a_i \mapsto b_i$ is a partial isomorphism between $(\str{\tree{z}_1},
A_1, \ldots, A_r)$ and $(\str{\tree{z}_2}, B_1, \ldots, B_r)$
completing the proof. We have the following three cases:

\begin{enumerate}[nosep]
\item Each of $a_i$ and $a_j$ is from $\V{\tree{t}_1}$: Then by the
  description of $\tbf{R}$ above, we have that (i) $b_i$ and $b_j$ are
  both from $\V{\tree{t}_2}$ and (ii) $a_i, a_j$ are adjacent in
  $\mathsf{Str}(\tree{t}_1)$ iff $b_i, b_j$ are adjacent in
  $\mathsf{Str}(\tree{t}_2)$. Observe that $a_i, a_j$ are adjacent in
  $\mathsf{Str}(\tree{t}_1)$ iff $a_i, a_j$ are adjacent in
  $\mathsf{Str}(\tree{z}_1)$. Likewise, $b_i, b_j$ are adjacent in
  $\mathsf{Str}(\tree{t}_2)$ iff $b_i, b_j$ are adjacent in
  $\mathsf{Str}(\tree{z}_2)$. Then $a_i, a_j$ are adjacent in
  $\mathsf{Str}(\tree{z}_1)$ iff $b_i, b_j$ are adjacent in
  $\mathsf{Str}(\tree{z}_2)$.

\item Each of $a_i$ and $a_j$ is from $\V{\tree{s}_1}$: Reasoning
  similarly as in the previous case, we can show that $a_i, a_j$ are
  adjacent in $\mathsf{Str}(\tree{z}_1)$ iff $b_i, b_j$ are adjacent
  in $\mathsf{Str}(\tree{z}_2)$.

\item W.l.o.g.  $a_i \in \V{\tree{t}_1}$ and $a_j \in
  \V{\tree{s}_1}$: Then $b_i \in \V{\tree{t}_2}$ and $b_j \in
  \V{\tree{s}_2}$. Observe now that the greatest common ancestor of
  $a_i$ and $a_j$ in $\tree{z}_1$ is $\troot{\tree{z}_1}$, and the
  greatest common ancestor of $b_i$ and $b_j$ in $\tree{z}_2$ is
  $\troot{\tree{z}_2}$. Since (i) the labels of $\troot{\tree{z}_1}$
  and $\troot{\tree{z}_2}$ are the same (by assumption) and (ii) the
  label of $a_i$ (resp. $a_j$) in $\tree{z}_1$ = label of $a_i$
  (resp. $a_j$) in $\str{\tree{z}_1}$ = label of $b_i$ (resp. $b_j$)
  in $\str{\tree{z}_2}$ = label of $b_i$ (resp. $b_j$) in
  $\tree{z}_2$, it follows by the definition of an $n$-partite cograph
  that $a_i, a_j$ are adjacent in $\mathsf{Str}(\tree{z}_1)$ iff $b_i,
  b_j$ are adjacent in $\mathsf{Str}(\tree{z}_2)$.
\end{enumerate}

\end{proof}

\begin{proof}[Proof idea for
    Theorem~\ref{theorem:n-partite-cographs-satisfy-lebsp}] We first
  show the result for $\cl{S} = \labelednpartitecographs(\Sigma)$. The
  result for the various specific classes mentioned in the statement
  of the theorem follows from the fact that $\lebsp{\cdot}$ is closed
  under hereditary subclasses, and that all of the specific classes
  are hereditary subclasses of $n$-partite
  cographs~\cite{shrub-depth}. As described prior to the statement of
  Lemma~\ref{lemma:composition-lemma-for-n-partite-cographs}, there is
  a natural representation map $\mathsf{Str}: \mc{T} \rightarrow
  \cl{S}$ that is elementary.  Using the composition lemma for
  $n$-partite graphs given by
  Lemma~\ref{lemma:composition-lemma-for-n-partite-cographs}, we can
  see that $\mathsf{Str}$ is $\mc{L}$-good for $\cl{S}$, whereby we
  are done by
  Theorem~\ref{theorem:good-tree-rep-implies-lebsp-and-f.p.t.-algorithm}. That
  the graph classes with the bounded parameters as above have
  elementary witness functions follows again from
  Theorem~\ref{theorem:good-tree-rep-implies-lebsp-and-f.p.t.-algorithm}
  and elementariness of the index of the $\lequiv{m}$ relation over
  these classes (the latter follows from Theorem 3.2
  of~\cite{shrub-depth-FO-equals-MSO}).
\end{proof}

\newcommand{\ndisjointsum}{\ensuremath{n\text{-}\mathsf{disjoint}\text{-}\mathsf{sum}}}
\newcommand{\ncopy}{\ensuremath{n\text{-}\mathsf{copy}}}
\newcommand{\twocopy}{\ensuremath{2\text{-}\mathsf{copy}}}

\subsection{Classes generated using translation schemes}

We look operations on classes of structures, that are
``implementable'' using quantifier-free translation
schemes~\cite{makowsky}. Given a vocabulary $\tau$, let
$\tau_{\text{disj-un}, n}$ be the vocabulary obtained by expanding
$\tau$ with $n$ fresh unary predicates $P_1, \ldots, P_n$.  Given
$\tau$-structures $\mf{A}_1, \ldots, \mf{A}_n$ (assumed disjoint
w.l.o.g.), the \emph{$n$-disjoint sum of $\mf{A}_1, \ldots,
  \mf{A}_n$}, denoted $\bigoplus_{i = 1}^{i = n} \mf{A}_i$, is the
$\tau_{\text{disj-un}, n}$-structure obtained upto isomorphism, by
expanding the disjoint union $\bigsqcup_{i = 1}^{i = n}\mf{A}_i$ with
$P_1, \ldots, P_n$ interpreted respectively as the universe of
$\mf{A}_1, \ldots, \mf{A}_n$. Let $\cl{S}_1, \ldots, \cl{S}_n$ be
given classes of structures. A quantifier-free $(t,
\tau_{\text{disj-un}, n}, \tau, \fo)$-translation scheme $\Xi$ gives
rise to an $n$-ary operation $\mathsf{O}: \cl{S}_1 \times \cdots
\times \cl{S}_n \rightarrow \{\Xi(\bigoplus_{i = 1}^{i = n} \mf{A}_i)
\mid \mf{A}_i \in \cl{S}_i, 1 \leq i \leq n\}$ defined as
$\mathsf{O}_1(\mf{A}_1, \ldots, \mf{A}_n) = \Xi(\bigoplus_{i = 1}^{i =
  n} \mf{A}_i)$.  In this case, we say that $\mathsf{O}$ is
\emph{implementable using $\Xi$}. We say an operation is
\emph{quantifier-free}, if it is of the kind $\mathsf{O}$ just
described.

For a quantifier-free operation $\mathsf{O}$, define the
\emph{dimension} of $\mathsf{O}$ to be the minimum of the dimensions
of the quantifier-free translation schemes that implement
$\mathsf{O}$.  We say $\mathsf{O}$ is ``sum-like'' if its dimension is
one, else we say $\mathsf{O}$ is ``product-like''.  Call $\mathsf{O}$
as \emph{$\lequiv{m}$-preserving} if whenever an input of $\mathsf{O}$
is replaced with an $\mc{L}[m]$-equivalent input, the output of
$\mathsf{O}$ is replaced with an $\mc{L}[m]$-equivalent output. We say
$\mathsf{O}$ is \emph{monotone} if any input of $\mathsf{O}$ is
embeddable in the output of $\mathsf{O}$. The well-studied unary graph
operations like complementation, transpose, and the line-graph
operation, and binary operations like disjoint union and join are all
sum-like, $\lequiv{m}$-preserving and monotone. Likewise, the
well-studied Cartesian, tensor, lexicographic, and strong products are
all product-like, $\lequiv{m}$-preserving and monotone. The central
result of this section is as stated below.

\begin{theorem}\label{theorem:lebsp-closure-under-1-step-operations}
Let $\cl{S}_1, \ldots, \cl{S}_n, \cl{S}$ be classes of structures and
let $\mathsf{O}: \cl{S}_1 \times \cdots \times \cl{S}_n \rightarrow
\cl{S}$ be a surjective $n$-ary quantifier-free operation. Then the
following are true:
\vspace{2pt}\begin{enumerate}[nosep]
  \item In each of the following scenarios, it is the case that if
    $\lebsp{\cl{S}_i}$ holds (with computable/elementary witness
    functions) for each $i \in \{1, \ldots, n\}$, then
    $\lebsp{\cl{S}}$ holds as well (with computable/elementary witness
    functions): (i) $\mathsf{O}$ is sum-like (ii) $\mathsf{O}$ is
    product-like and $\mc{L} = \fo$.
  \item Suppose $\cl{S}_i$ admits an effective $\mc{L}$-good tree
    representation for each $i \in \{1, \ldots, n\}$, and $\mathsf{O}$
    is $\lequiv{m}$ preserving and monotone. Then there exists an effective
    $\mc{L}$-good tree representation $\mathsf{Str}: \mc{T}
    \rightarrow \mc{Z}$ for the class $\mc{Z} = \cl{S} \cup \bigcup_{i
      = 1}^{i = n} \cl{S}_i$.
  \item Let $\mathsf{Str}$ be as given by the previous point. Then
    there is a linear time f.p.t. algorithm for $\mathsf{MC}(\mc{L},
    \cl{S})$ that decides, for every $\mc{L}$ sentence $\varphi$ (the
    parameter), if a given structure $\mf{A}$ in $\cl{S}$ satisfies
    $\varphi$, provided a tree representation of $\mf{A}$ under
    $\mathsf{Str}$ is given.
   \end{enumerate}
\end{theorem}

Towards the proof of
Theorem~\ref{theorem:lebsp-closure-under-1-step-operations}, we first
present the following two auxiliary results. Below,
$\ndisjointsum(\cl{S}_1, \ldots, \cl{S}_n) = \{ \bigoplus_{i = 1}^{i =
  n} \mf{A}_i \mid \mf{A}_i \in \cl{S}_i, 1 \leq i \leq n \}$. Also,
we say a quantifier-free translation scheme is \emph{scalar} if its
dimension is one.

\begin{lemma}\label{lemma:lebsp-pres-under-n-disj-sum-and-n-copy}
Let $\cl{S}, \cl{S}_1, \ldots, \cl{S}_n$ be classes of structures for
$n \ge 1$.  If $\lebsp{\cl{S}_i}$ is true for each $i \in \{1, \ldots,
n\}$, then so is $\lebsp{\ndisjointsum(\cl{S}_1, \ldots, \cl{S}_n)}$.
Further, if there is a computable/elementary witness function for
$\lebsp{\cl{S}_i}$ for each $i \in \{1, \ldots, n\}$, then there is a
computable/elementary witness function for
$\lebsp{\ndisjointsum(\cl{S}_1, \ldots, \cl{S}_n)}$ as well.
\end{lemma}

\begin{prop}\label{prop:ebsp-gebsp-and-transductions}
Let $\cl{S}$ be class of $\tau$-structures, and let $\Xi = (\xi,
(\xi_R)_{R \in \nu})$ be a quantifier-free $(t, \tau, \nu,
\fo)$-translation scheme.  Then the following hold for each $k \in
\mathbb{N}$.
\begin{enumerate}[nosep]
\item If ~$\febsp{\cl{S}}$ is true, then so is
  $\febsp{\Xi(\cl{S})}$.\label{prop:ebsp-pres-under-transductions}
\item If ~$\Xi$ is scalar and $\mebsp{\cl{S}}$ is true, then so
  is $\mebsp{\Xi(\cl{S})}$.
  \label{prop:gebspmso-pres-under-transductions}
\end{enumerate}
In each of the implications above, a computable/elementary witness
function for the antecedent implies a computable/elementary witness
function for the consequent.
\end{prop}

\begin{proof}[Proof of Theorem~\ref{theorem:lebsp-closure-under-1-step-operations}]
  (1): Follows easily from
  Lemma~\ref{lemma:lebsp-pres-under-n-disj-sum-and-n-copy} and
  Proposition~\ref{prop:ebsp-gebsp-and-transductions}.

  (2): Let $\mathsf{Str}_i: \mc{T}_i \rightarrow \cl{S}_i$ be an
  effective $\mc{L}$-good tree representation for $\cl{S}_i$ for
  $1\leq i \leq n$, where $\mc{T}_i$ is a class of trees over
  $(\sigmaint^i \cup \sigmaleaf^i)$ that is representation feasible
  for $(\sigmarank^i, \rho_i)$.

  Let $O$ be a new label that is not in $(\sigmaint^i \cup
  \sigmaleaf^i)$ for any $i \in \{1, \ldots, n\}$. Define $\sigmaint,
  \sigmaleaf, \sigmarank$ and $\rho: \sigmarank \rightarrow
  \mathbb{N}_+$ as follows:
  \begin{itemize}
    \item $\sigmaint = \{O\} \cup \bigcup_{i = 1}^{i =n} \sigmaint^i$
    \item $\sigmaleaf = \bigcup_{i = 1}^{i =n} \sigmaleaf^i$
    \item $\sigmarank = \{O\} \cup \bigcup_{i = 1}^{i =n}
      \sigmarank^i$
    \item $\rho = \{ (O, n)\} \cup \bigcup_{i = 1}^{i = n} \rho_i$.
  \end{itemize}
  Let $\widehat{\mc{T}}$ be the class of all trees over $(\sigmaint
  \cup \sigmaleaf)$ obtained by taking $\tree{t}_i \in \mc{T}_i$ for
  $1 \leq i \leq n$, and making $\tree{t}_1, \ldots, \tree{t}_n$ as
  child subtrees (and in that order) of a new root node whose label is
  $O$.  Let $\mc{T} = \widehat{\mc{T}} \cup \bigcup_{i = 1}^{i = n}
  \mc{T}_i$. Verify that $\mc{T}$ is indeed representation feasible
  for $(\sigmarank, \rho)$.

  Let $\mathsf{Str}: \mc{T} \rightarrow \mc{Z}$ be such that for
  $\tree{t} \in \mc{T}$, if $\tree{t} \in \mc{T}_i$, then
  $\str{\tree{t}} = \mathsf{Str}_i(\tree{t})$. Else, let $a_1, \ldots,
  a_n$ be the children of the root of $\tree{t}$. Clearly then
  $\tree{t}_{\ge a_i} \in \mc{T}_i$ by construction of $\mc{T}$. Then
  define $\str{\tree{t}} = \mathsf{O}(\mathsf{Str}_1(\tree{t}_{\ge
    a_1}), \ldots, \mathsf{Str}_n(\tree{t}_{\ge a_n}))$.

  Using the fact that $\mathsf{O}$ is $\lequiv{m}$-preserving and
  monotone, and using
  Lemma~\ref{lemma:qt-free-transductions-preserve-substructure-prop}
  that we prove below, it is easy to verify that $\mathsf{Str}$ is
  indeed an effective $\mc{L}$-good representation map for $\cl{Z}$.

  (3): Since $\mathsf{Str}$ is effective and $\mc{L}$-good for
  $\cl{Z}$, by
  Theorem~\ref{theorem:good-tree-rep-implies-lebsp-and-f.p.t.-algorithm},
  there is a linear time f.p.t. algorithm for $\mathsf{MC}(\mc{L},
  \cl{Z})$ that decides, for every $\mc{L}$ sentence $\varphi$, if a
  given structure $\mf{A}$ in $\cl{Z}$ satisfies $\varphi$, provided
  that a tree representation of $\mf{A}$ under $\mathsf{Str}$. Clearly
  the same algorithm is also f.p.t.  for $\mathsf{MC}(\mc{L},
  \cl{S})$.
\end{proof}

The remainder of this section is devoted to proving
Lemma~\ref{lemma:lebsp-pres-under-n-disj-sum-and-n-copy} and
Proposition~\ref{prop:ebsp-gebsp-and-transductions}.

Towards the proof of
Lemma~\ref{lemma:lebsp-pres-under-n-disj-sum-and-n-copy}, we present
the following simple facts about $n$-disjoint sum. We skip the proof.

\begin{lemma}\label{lemma:facts-about-disj-unions}
Let $\mf{A}_i$ and $\mf{B}_i$ be $\tau$-structures for $i \in \{1,
\ldots, n\}$. Let $m \in \mathbb{N}$. Then the following are true.
\begin{enumerate}[nosep]
\item If $\mf{B}_i \hookrightarrow \mf{A}_i$ for $i \in \{1, \ldots,
  n\}$, then $(\bigoplus_{i = 1}^{i = n} \mf{B}_i) \hookrightarrow
  (\bigoplus_{i = 1}^{i = n} \mf{A}_i)$.
\item If $\mf{B}_i \lequiv{m} \mf{A}_i$ for $i \in \{1, \ldots, n\}$,
  then $(\bigoplus_{i = 1}^{i = n} \mf{B}_i) \lequiv{m} (\bigoplus_{i
    = 1}^{i = n} \mf{A}_i)$.
\end{enumerate}
\end{lemma}

\begin{proof}[Proof of Lemma~\ref{lemma:lebsp-pres-under-n-disj-sum-and-n-copy}]
Consider a structure $\mf{A} = (\bigoplus_{i = 1}^{i = n} \mf{A}_i)
\in \ndisjointsum(\cl{S}_1, \ldots, \cl{S}_n)$. Let $m \in
\mathbb{N}$. Since $\lebsp{\cl{S}_i}$ is true, there exists $\mf{B}_i$
such that\linebreak $\lebspcond(\cl{S}_i, \mf{A}_i, \mf{B}_i, m,
\lwitfn{\cl{S}_i}(m))$ holds where $\lwitfn{\cl{S}_i}(m)$ is a witness
function for $\lebsp{\cl{S}_i}$.  Then $\mf{B}_i \subseteq \mf{A}_i$
and $\mf{B}_i \lequiv{m} \mf{A}_i$. Then by
Lemma~\ref{lemma:facts-about-disj-unions}, we have that (i)
$\bigoplus_{i = 1}^{i = n} \mf{B}_i \hookrightarrow \bigoplus_{i =
  1}^{i = n} \mf{A}_i$, and (ii) $\bigoplus_{i = 1}^{i = n} \mf{B}_i
\lequiv{m} \bigoplus_{i = 1}^{i = n} \mf{A}_i$. Observe that
$(\bigoplus_{i = 1}^{i = n} \mf{B}_i) \in \ndisjointsum(\cl{S}_1,
\ldots, \cl{S}_n)$, and that $|(\bigoplus_{i = 1}^{i = n} \mf{B}_i)|
\leq \theta(m) = \sum_{i = 0}^{i = n} \lwitfn{\cl{S}_i}(m)$. Taking
$\mf{B}$ to be the substructure of $\mf{A}$ that is isomorphic to
$(\bigoplus_{i = 1}^{i = n} \mf{B}_i)$, we see that \linebreak
$\lebspcond(\ndisjointsum(\cl{S}_1, \ldots, \cl{S}_n), $ $ \mf{A},
\mf{B}, m, \theta)$ is true with witness function $\theta$. Whereby
$\lebsp{\ndisjointsum(\cl{S}_1, \ldots, \cl{S}_n)}$ is true. It is
easy to see that if $\lwitfn{\cl{S}_i}(m)$ is computable/elementary
for each $i \in \{1, \ldots, n\}$, then so is $\theta$.
\end{proof}

We now proceed to proving
Proposition~\ref{prop:ebsp-gebsp-and-transductions}. We use the
following known facts about translation schemes~\cite{makowsky}.  To
present these facts, we recall from Section~\ref{section:background}
that one can associate with a $(t, \tau, \nu, \mc{L})$-translation
scheme $\Xi$, two partial maps: (i) $\Xi^*$ from $\tau$-structures to
$\nu$-structures (ii) $\Xi^\sharp$ from $\mc{L}(\nu)$ formulae to
$\mc{L}(\tau)$ formulae. See~\cite{makowsky} for the definitions of
these.  For the ease of readability, we abuse notation slightly and
use $\Xi$ to denote both $\Xi^*$ and $\Xi^\sharp$. We now have the
following results from literature.

\begin{prop}\label{prop:relating-transductions-applications-to-structures-and-formulae}
Let $\Xi$ be either a $(t, \tau, \nu, \fo)$-translation scheme for
$t \ge 1$, or a $(t, \tau, \nu, \mso)$-translation scheme with $t =
1$.  Then for every $\mc{L}(\nu)$ formula $\varphi(x_1, \ldots,
x_n)$ where $n \ge 0$, for every $\tau$-structure $\mf{A}$ and for
every $n$-tuple $(\bar{a}_1, \ldots, \bar{a}_n)$ from $\Xi(\mf{A})$,
the following holds.
\[
\begin{array}{lrll}
& (\Xi(\mf{A}), \bar{a}_1, \ldots, \bar{a}_n) & \models & \varphi(x_1,
  \ldots, x_n)\\ 
\mbox{iff} & (\mf{A}, \bar{a}_1, \ldots, \bar{a}_n) & \models &
\Xi(\varphi)(\bar{x}_1, \ldots, \bar{x}_n)\\
\end{array}
\]
where $\bar{x}_i = (x_{i, 1}, \ldots, x_{i, t})$ for each $i \in \{1,
\ldots, n\}$.
\end{prop}

\begin{lemma}\label{corollary:transferring-m-equivalence-across-transductions}
Let $\Xi$ be a quantifier-free $(t, \tau, \nu, \fo)$-translation
scheme.  Let $m, r \in \mathbb{N}$ be such that $r = t \cdot m$.
Suppose $\mf{A}$ and $\mf{B}$ are $\tau$-structures.
\begin{enumerate}[nosep]
\item If $\mf{A} \fequiv{r} \mf{B}$, then $\Xi(\mf{A}) \fequiv{m}
  \Xi(\mf{B})$.
\item If $\mf{A} \mequiv{m} \mf{B}$, then $\Xi(\mf{A}) \mequiv{m}
  \Xi(\mf{B})$, when $\Xi$ is scalar.
\end{enumerate}
\end{lemma}

Towards the proof of
Proposition~\ref{prop:ebsp-gebsp-and-transductions}, we first observe
the following result that shows that quantifier-free translation
schemes preserve the substructure relation between any two structures
of $\cl{S}$. 

\begin{lemma}\label{lemma:qt-free-transductions-preserve-substructure-prop}
Let $\cl{S}$ be a given class of finite structures. Let $\Xi = (\xi,
(\xi_R)_{R \in \nu})$ be a quantifier-free $(t, \tau, \nu,
\fo)$-translation scheme. Let $\mf{A}$ and $\mf{B}$ be given
structures from $\cl{S}$.  If $\mf{B} \subseteq \mf{A}$, then
$\Xi(\mf{B}) \subseteq \Xi(\mf{A})$.
\end{lemma}
\begin{proof}
Consider any element of $\Xi(\mf{B})$; it is a $t$-tuple $\bar{b}$ of
$\mf{B}$ such that $(\mf{B}, \bar{b}) \models \xi(\bar{x})$. Since
$\xi(\bar{x})$ is quantifier-free, it is preserved under extensions
over $\cl{S}$. Whereby $(\mf{A}, \bar{b}) \models \xi(\bar{x})$; then
$\bar{b}$ is an element of $\Xi(\mf{A})$. Since $\bar{b}$ is an
arbitrary element of $\Xi(\mf{B})$, we have $\mathsf{U}_{\Xi(\mf{B})}
\subseteq \mathsf{U}_{\Xi(\mf{A})}$. 

Consider a relation symbol $R \in \nu$ of arity say $n$. Let
$\bar{d}_1, \ldots, \bar{d}_n$ be elements of $\Xi(\mf{B})$. Then we
have the following. Below $\bar{x}_i = (x_{i, 1}, \ldots, x_{i, t})$
for each $i \in \{1, \ldots, n\}$.
\[
\begin{array}{lllll}
& (\Xi(\mf{B}), \bar{d}_1, \ldots, \bar{d}_n) & \models & R(x_1,
  \ldots, x_n) & \\

\mbox{iff} & (\mf{B}, \bar{d}_1, \ldots, \bar{d}_n) & \models &
\Xi(R)(\bar{x}_1, \ldots, \bar{x}_n) & \mbox{(by
  Proposition~\ref{prop:relating-transductions-applications-to-structures-and-formulae})}\\

\mbox{iff} & (\mf{B}, \bar{d}_1, \ldots, \bar{d}_n) & \models &
\bigwedge_{i = 1}^{i = n}\xi(\bar{x}_i) ~\wedge~ \xi_R(\bar{x}_1,
\ldots, \bar{x}_n) & \mbox{(by defn. of $\Xi(R)$; see~\cite{makowsky})} \\
\end{array}
\]
Now since (i) each of $\xi$ and $\xi_R$ is quantifier-free, (ii) a
finite conjunction of quantifier-free formulae is a
quantifier-free formula, and (iii) a quantifier-free formula is
preserved under substructures as well as preserved under extensions
over any class, we have that
\[
\begin{array}{lllll}
& (\mf{B}, \bar{d}_1, \ldots, \bar{d}_n) & \models & \bigwedge_{i =
    1}^{i = n}\xi(\bar{x}_i) ~\wedge~ \xi_R(\bar{x}_1, \ldots,
  \bar{x}_n) &  \\

\mbox{iff} & (\mf{A}, \bar{d}_1, \ldots, \bar{d}_n) & \models &
\bigwedge_{i = 1}^{i = n}\xi(\bar{x}_i) ~\wedge~ \xi_R(\bar{x}_1,
\ldots, \bar{x}_n) & \\

\mbox{iff} & (\mf{A}, \bar{d}_1, \ldots, \bar{d}_n) & \models &
\Xi(R)(\bar{x}_1, \ldots, \bar{x}_n) & \mbox{(by definition of
  $\Xi(R)$)} \\

\mbox{iff} & (\Xi(\mf{A}), \bar{d}_1, \ldots, \bar{d}_n) & \models & R(x_1,
  \ldots, x_n) & \mbox{(by
  Proposition~\ref{prop:relating-transductions-applications-to-structures-and-formulae})}\\
\end{array}
\]
Since $R$ is an arbitrary relation symbol of $\nu$, we have that
$\Xi(\mf{B}) \subseteq \Xi(\mf{A})$.
\end{proof}

\begin{proof}[Proof of Proposition~\ref{prop:ebsp-gebsp-and-transductions}]

We show the proof for part 1. The proof for part 2 is similar.

Consider a structure $\Xi(\mf{A}) \in \Xi(\cl{S})$ for some structure
$\mf{A} \in \cl{S}$. Let $m \in \mathbb{N}$. Since $\febsp{\cl{S}}$ is
true, there exists a witness function $\fwitfn{\cl{S}}:\mathbb{N}
\rightarrow \mathbb{N}$ and a structure $\mf{B}$ such that if $r = t
\cdot m$, then $\febspcond(\cl{S}, \mf{A}, \mf{B}, r,
\fwitfn{\cl{S}})$ is true. That is, (i) $\mf{B} \in \cl{S}$, (ii)
$\mf{B} \subseteq \mf{A}$ (iii) $|\mf{B}| \leq \fwitfn{\cl{S}}(r)$ and
(iv) $\mf{B} \fequiv{r} \mf{A}$.

We now show that there exists a function
$\fwitfn{\Xi(\cl{S})}:\mathbb{N} \rightarrow \mathbb{N}$ such
that\linebreak $\febspcond(\Xi(\cl{S}), \Xi(\mf{A}), $ $\Xi(\mf{B}),
m, \fwitfn{\Xi(\cl{S})})$ is true. This would show \linebreak
$\febsp{\Xi(\cl{S})}$ is true.
\begin{enumerate}
\item $\Xi(\mf{B}) \in \Xi(\cl{S})$: Obvious from the definition of
  $\Xi(\cl{S})$ and the fact that $\mf{B} \in \cl{S}$.

\item $\Xi(\mf{B}) \subseteq \Xi(\mf{A})$: Follows from
  Lemma~\ref{lemma:qt-free-transductions-preserve-substructure-prop}.

\item $\Xi(\mf{B}) \fequiv{m} \Xi(\mf{A})$: Since $\mf{B} \fequiv{r}
  \mf{A}$, it follows from
  Lemma~\ref{corollary:transferring-m-equivalence-across-transductions},
  that $\Xi(\mf{B}) \fequiv{m} \Xi(\mf{A})$.

\item The existence of a function $\fwitfn{\Xi(\cl{S})}: \mathbb{N}
  \rightarrow \mathbb{N}$ such that $|\Xi(\mf{B})| \leq
  \fwitfn{\Xi(\cl{S})}(m)$: Define $\fwitfn{\Xi(\cl{S})}:\mathbb{N}
  \rightarrow \mathbb{N}$ as $\fwitfn{\Xi(\cl{S})}(m) =
  (\fwitfn{\cl{S}}(t \cdot m))^t$. Since $|\mf{B}| \leq
  \fwitfn{\cl{S}}(t \cdot m)$, we have that $|\Xi(\mf{B})| \leq
  \fwitfn{\Xi(\cl{S})}(m)$.
\end{enumerate}
It is clear that if $\fwitfn{\cl{S}}$ is computable/elementary, then
so is $\fwitfn{\Xi(\cl{S})}$.
\end{proof}

\begin{discussion}
  Theorems~\ref{theorem:words-and-trees-and-nested-words-satisfy-lebsp},
  ~\ref{theorem:n-partite-cographs-satisfy-lebsp},
  ~\ref{theorem:lebsp-closure-under-1-step-operations} and
  ~\ref{theorem:good-tree-rep-implies-lebsp-and-f.p.t.-algorithm}
  jointly show that the various posets and graph classes described in
  this section admit linear time f.p.t. algorithms for
  $\mathsf{MC}(\mc{L}, \cdot)$, provided an $\mc{L}$-good tree
  representation of the input structure is given.  In the case of
  words, the various kinds of trees, nested words, the class of
  cographs, we can indeed even construct the $\mc{L}$-good tree
  representation in quadratic time from a standard presentation of
  structures in these classes (this is easy to see for the first three
  kinds of classes; for the case of cographs,
  see~\cite{cograph-1981-paper}). Whereby, these classes admit
  quadratic time f.p.t. algorithms for $\mathsf{MC}(\mc{L}, \cdot)$.
  The quadratic time is because we have assumed our
  tree-representations to be poset trees, which are required to be
  transitive.  The graph theoretic directed trees underlying the poset
  trees (which are the Hasse diagrams of the poset trees) are actually
  constructible in linear time for each of the cases of words, the
  various kinds of trees, nested words and the class of cographs. One
  can see that the techniques that we use to get linear time
  f.p.t. algorithms for the aforesaid classes, given $\mc{L}$-good
  tree representations for structures in these classes, can be adapted
  to work even when the structures in these classes are represented
  using the graph theoretic (Hasse diagram) tree representations just
  mentioned. This indeed then enables getting linear time
  f.p.t. algorithms $\mathsf{MC}(\mc{L}, \cdot)$ for the case of
  words, various kinds of trees, nested words, and cographs, thereby
  matching known f.p.t. results concerning these
  classes~\cite{frick-grohe,alur-madhu,shrub-depth-FO-equals-MSO}. Going
  further, to the best of our knowledge, the f.p.t. results for
  $n$-partite cographs and those for classes generated using trees of
  quantifier-free operations, that are entailed by Theorems
  ~\ref{theorem:n-partite-cographs-satisfy-lebsp},
  ~\ref{theorem:lebsp-closure-under-1-step-operations} and
  ~\ref{theorem:good-tree-rep-implies-lebsp-and-f.p.t.-algorithm}, are
  new.  Our proofs can then be seen as giving a different and unified
  technique to show existing f.p.t. results, in addition to giving new
  results. We mention however that if the dependence on the parameter
  in our f.p.t. algorithms is also considered, then our results (which
  give only computable parameter dependence) are weaker than those
  in~\cite{shrub-depth-FO-equals-MSO} which show that for classes of
  bounded tree-depth/$\mc{SC}$-depth/shrub-depth, there are linear
  time f.p.t. algorithms for $\mc{L}$ model checking, that have
  elementary parameter dependence.
\end{discussion}

\newcommand{\scale}[1]{\langle #1 \rangle}
\newcommand{\fracwitfn}[1]{\ensuremath{\theta_{(#1, \mc{L})}}}

\vspace{-7pt}\section{Logical fractals}\label{section:logical-fractal}

We define a strengthening of the notion of $\reflebsp$ -- instead of
asserting ``logical self-similarity'' just at ``small scales'', we
assert the same ``for all scales'' for a suitable notion of scale. To
present the formal definition, we say a function $f: \mathbb{N}_+
\rightarrow \mathbb{N}_+$ is a \emph{scale function} if it is strictly
increasing. The \emph{$i^{\text{th}}$ scale}, denoted $\scale{i}_{f}$,
is defined as the interval $[1, f(1)] = \{ j \mid 1 \leq j \leq
f(i)\}$ if $i = 1$, and $\left[f(i-1) + 1, f(i)\right] = \{ j \mid
f(i-1) + 1 \leq j \leq f(i)\}$ if $i > 1$.

\begin{definition}[Logical fractal]\label{definition:logical-fractal}
  Given a class $\mc{S}$ of structures and a logic $\mc{L}$ that is
  $\fo$ or $\mso$, we say $\cl{S}$ is an \emph{$\mc{L}$-fractal}, if
  there exists a function $\fracwitfn{\cl{S}}: \mathbb{N}_+^2
  \rightarrow \mathbb{N}_+$ such that (i) $\fracwitfn{\cl{S}}(m)$ is a
  scale function for all $m \in \mathbb{N}$, and (ii) for each
  structure $\mf{A}$ of $\cl{S}$ and each $m \in \mathbb{N}$, if $f$
  is the function $\fracwitfn{\cl{S}}(m)$ and $|\mf{A}| \in
  \scale{i}_{f}$ for some $i \in \mathbb{N}$, then for all $j < i$,
  there exists a substructure $\mf{B}$ of $\mf{A}$ in $\cl{S}$, such
  that $|\mf{B}| \in \scale{j}_{f}$ and $\mf{B} \lequiv{m} \mf{A}$. We
  say $\fracwitfn{\cl{S}}$ is a \emph{witness} to the $\mc{L}$-fractal
  property of $\cl{S}$.
\end{definition}

Towards the central result of this section, we first show the
following result.

\begin{lemma}\label{lemma:logical-fractal-and-tree-reps}
Let $\cl{S}$ be a class of structures that admits an $\mc{L}$-good
tree representation $\mathsf{Str}: \mc{T} \rightarrow \cl{S}$. Then
there exists a strictly increasing computable function $\eta:
\mathbb{N} \rightarrow \mathbb{N}$ such that for each $m \in
\mathbb{N}$ and for each tree $\tree{t} \in \mc{T}$ of size $>
\eta(m)$, there exists a proper subtree $\tree{s}$ of $\tree{t}$ in
$\mc{T}$ such that (i) $\str{\tree{s}} \hookrightarrow
\str{\tree{t}}$, (ii) $\str{\tree{s}} \lequiv{m} \str{\tree{t}}$, and
(iii) $|\tree{t}| - |\tree{s}| \leq \eta(m)$.
\end{lemma}

\newcommand{\boundedred}{\ensuremath{\mathsf{Bounded}\text{-}\mathsf{reduction}}}
\begin{proof}

  The proof is very much along the lines of the proof of
  Lemma~\ref{lemma:abstract-tree-lemma}. 

  Let $\mathsf{Str}: \mc{T} \rightarrow \cl{S}$ be an $\mc{L}$-good
  tree representation for $\cl{S}$, where $\mc{T}$ is a class of
  $(\sigmaleaf \cup \sigmaint)$-trees that is representation feasible
  for $(\sigmarank, \rho)$.  Let $m_0 \in \mathbb{N}$ be a witness to
  the composition property of $\mathsf{Str}$, as mentioned in
  Definition~\ref{definition:L-good-map}.  Let $\eta_1, \eta_2:
  \mathbb{N} \rightarrow \mathbb{N}$ be defined as follows: For $l \in
  \mathbb{N}$, $\eta_1(l) = \text{max}\{\rho(\sigma) \mid \sigma \in
  \sigmaint \} \times \Lambda_{\cl{S}, \mc{L}}(\text{max}\{l, m_0\})$
  and $\eta_2(l) = 1 + \Lambda_{\cl{S}, \mc{L}}(\text{max}\{l,
  m_0\})$. For $l \in \mathbb{N}$, define $\eta(l) = \Lambda_{\cl{S},
    \mc{L}}(\text{max}\{l, m_0\}) \cdot \eta_1(l)^{(\eta_2(l) + 1)} +
  \eta_1(l)^{(\eta_2(l) + 2)}$. Clearly $\eta$ is strictly increasing
  and computable.

  Let $b$ be a node of $\tree{t}$ with the properties mentioned below.
  \begin{itemize}[nosep]
    \item $b$ is a ``closest-to-a-leaf'' node of $\tree{t}$ whose
      degree $> \eta_1(m)$. In other words, every node in
      $\tree{t}_{\ge b}$, other than $b$, has degree $\leq \eta_1(m)$,
      while $b$ has degree $> \eta_1(m)$.
    \item
      for each child $c$ of $b$ in $\tree{t}$, the subtree
      $\tree{t}_{\ge c}$ has height $\leq \eta_2(m)$.
  \end{itemize}

  \vspace{2pt}We have the following two cases. Let $m_1 = \text{max}(m_0, m)$.

\vspace{3pt}1. The node $b$ exists: We then perform a ``degree
reduction'' just as in
Lemma~\ref{lemma:abstract-tree-lemma}(\ref{lemma:abstract-tree-lemma-degree-reduction}).

Let $\sigma$ be the label of $\troot{\tree{t}_{\ge b}}$. Since degree
of $b$ is $> \eta_1(m)$, we have $\sigma \in \sigmaint \setminus
\sigmarank$. Let $\tree{z} = \tree{t}_{\ge b}$ and let $a_1, \ldots,
a_n$ be the (ascending) sequence of children of $b$ in $\tree{t}$. For
$d = \rho(\sigma)$, let $n = r + q \cdot (d - 1)$ for $1 \leq r < d$
and $q > 1$.
  
For $k \in I = \{r + j \cdot (d - 1) \mid 0 \leq j \leq q\}$, let
$\tree{x}_{1, k}$, resp. $\tree{y}_{k+1, n}$, be the subtree of
$\tree{z}$ obtained from $\tree{z}$ by deleting the subtrees rooted at
$a_{k+1}, \ldots, a_n$, resp.  deleting the subtrees rooted at $a_1,
a_2, \ldots, a_k$. Then $\tree{z} = \tree{x}_{1, n} = \tree{x}_{1, k}
\odot \tree{y}_{k+1, n}$ for all $k \in I$.  Define $g: I \rightarrow
\ldelta{\cl{S}}{m_1}$ such that $g(k)$ is the $\lequiv{m_1}$ class of
$\str{\tree{x}_{1, k}}$ for $k \in I$.

Since $n > \eta_1(m)$, there exist $i, j \in I$ such that $i < j$, $j
- i \leq \Lambda_{\cl{S}, \mc{L}}(m_1)$ and $g(i) = g(j)$,
i.e. $\str{\tree{x}_{1, i}} \lequiv{m_1} \str{\tree{x}_{1, j}}$.  Let
$\tree{z}_1$ be the subtree of $\tree{z}$ obtained by deleting the
subtrees of $\tree{z}$ that are rooted at $a_{i+1}, \ldots, a_j$. By a
similar reasoning as in the proof of
Lemma~\ref{lemma:abstract-tree-lemma}(\ref{lemma:abstract-tree-lemma-degree-reduction}),
we see that if $\tree{s} = \tree{t} \left[ \tree{z} \mapsto \tree{z}_1
  \right]$, then $\tree{s}$ is a proper subtree of $\tree{t}$ in
$\mc{T}$ such that $\str{\tree{s}} \hookrightarrow \str{\tree{t}}$ and
$\str{\tree{s}} \lequiv{m_1} \str{\tree{t}}$, whereby $\str{\tree{s}}
\lequiv{m} \str{\tree{t}}$ (since $m_1 \ge m$). Finally, since for
each $l \in \{i + 1, \ldots, j\}$, the subtree of $\tree{z}$ rooted at
$a_l$ has degree $\leq \eta_1(m)$ and height $\leq \eta_2(m)$, we have
that $|\tree{t}| - |\tree{s}|$ is at most $(j - i) \cdot
\eta_1(m)^{(\eta_2(m)+1)} \leq \eta(m)$.

\vspace{3pt}2.  The node $b$ does not exist: Then there exists some node $c$ of
$\tree{t}$ such that $\tree{t}_{\ge c}$ has degree $\leq \eta_1(m)$
and height $\eta_2(m) + 1$. This can be seen as follows. Either there
is no node in $\tree{t}$ of degree $> \eta_1(m)$ in which case the
size of $\tree{t}$ being $> \eta(m)$ (as assumed in the statement of
the lemma) itself implies the existence of node $c$ as
aformentioned. Else there is a node in $\tree{t}$ of degree $>
\eta_1(m)$, whereby there is a closest-to-a-leaf such node, call it
$d$.  Then some child $d_1$ of $d$ must be such that $\tree{t}_{\ge
  d_1}$ has height $> \eta_2(m)$ (for otherwise $d$ can be taken to be
node $b$ which we have assumed does not exist). Then the
aforementioned node $c$ can be found in $\tree{t}_{\ge d_1}$.

We now perform a ``height reduction'' as in
Lemma~\ref{lemma:abstract-tree-lemma}(\ref{lemma:abstract-tree-lemma-height-reduction}).
Let $A$ be the set of nodes appearing on a path of length $\eta_2(m) +
1$ from $c$ to some leaf of $\tree{t}$.  Consider the function $h : A
\rightarrow \ldelta{\cl{S}}{m_1}$ such that for each $a \in A$, $h(a)
= \delta$ where $\delta$ is the $\lequiv{m_1}$ class of
$\str{\tree{t}_{\ge a}}$. Since $|A| > \eta_2(m)$, there exist
distinct nodes $a, e \in A$ such that $a$ is an ancestor of $e$ in
$\tree{t}$ and $h(a) = h(e)$.  Let $\tree{s} =
\tree{t}\left[\tree{t}_{\ge a} \mapsto \tree{t}_{\ge e} \right]$. By a
similar reasoning as in the proof of
Lemma~\ref{lemma:abstract-tree-lemma}(\ref{lemma:abstract-tree-lemma-height-reduction}),
we can see that (i) $\tree{s}$ is a subtree of $\tree{t}$ in $\mc{T}$,
(ii) $\str{\tree{s}} \hookrightarrow \str{\tree{t}}$ and (iii)
$\str{\tree{s}} \lequiv{m_1} \str{\tree{t}}$, whereby $\str{\tree{s}}
\lequiv{m} \str{\tree{t}}$ (since $m_1 \ge m$). Since $a$ is a
descendent of $c$, and since $\tree{t}_{\ge c}$ has height $\eta_2(m)
+ 1$ and degree $\leq \eta_1(m)$, the height and degree of
$\tree{t}_{\ge a}$ are resp. at most $\eta_2(m) + 1$ and $\eta_1(m)$,
whereby the size of $\tree{t}_{\ge a}$ is $\leq
\eta_1(m)^{(\eta_2(m)+2)} \leq \eta(m)$. Clearly then $|\tree{t}| -
|\tree{s}| \leq \eta(m)$.
\end{proof}

Call a representation map $\mathsf{Str}: \mc{T} \rightarrow \cl{S}$ as
\emph{$\mc{L}$-great} if (i) it is $\mc{L}$-good for $\cl{S}$, and
(ii) there is a strictly increasing function $\beta: \mathbb{N}
\rightarrow \mathbb{N}$ such that for every $\tree{t}, \tree{s} \in
\mc{T}$, if $|(|\tree{t}| - |\tree{s}|)| \leq n$, then
$|(|\str{\tree{t}}| - |\str{\tree{s}}|)| \leq \beta(n)$. In such a
case, we say $\cl{S}$ \emph{admits} an $\mc{L}$-great tree
representation. The central result of this section can now be stated
as below.

\begin{prop}\label{prop:greatness-and-fractals}
If $\cl{S}$ admits an $\mc{L}$-great tree representation, then
$\cl{S}$ is an $\mc{L}$-fractal.
\end{prop}

\begin{proof}
  Let $\mathsf{Str}: \mc{T} \rightarrow \cl{S}$ be an $\mc{L}$-great
  representation for $\cl{S}$.  Then $\mathsf{Str}$ is $\mc{L}$-good,
  whereby by Lemma~\ref{lemma:logical-fractal-and-tree-reps}, there
  exists a function $\eta$ satisfying the properties mentioned in the
  lemma. For $m \in \mathbb{N}$, define $f(m) =
  \text{max}\{|\str{\tree{t}}| \mid |\tree{t}| \leq \eta(m)\}$. Since
  $\mathsf{Str}$ is $\mc{L}$-great, there exists a function $\beta:
  \mathbb{N} \rightarrow \mathbb{N}$ satisfying the properties
  mentioned in the definition of $\mc{L}$-greatness. Then define the
  function $\lwitfn{\cl{S}}:\mathbb{N}_+^2 \rightarrow \mathbb{N}_+$
  as $\lwitfn{\cl{S}}(m)(n) = f(m) + (n - 1) \cdot \beta(\eta(m))$.
  It is easily seen that $\lwitfn{\cl{S}}(m)$ is a scale
  function. Consider $\mf{A} \in \cl{S}$ and $m \in \mathbb{N}$, and
  suppose that $|\mf{A}| \in \scale{i}_g$ where $g$ is the function
  $\lwitfn{\cl{S}}(m)$ and $i > 1$. To show that for $j < i$, there
  exists a substructure $\mf{B}$ of $\mf{A}$ in $\cl{S}$ such that
  $|\mf{B}| \in \scale{j}_g$ and $\mf{B} \lequiv{m} \mf{A}$, we
  observe that it suffices to show the same simply for $j = i - 1$.
  Let $\tree{t} \in \mc{T}$ be such that $\str{\tree{t}} = \mf{A}$. By
  Lemma~\ref{lemma:logical-fractal-and-tree-reps}, there exists a
  subtree $\tree{s}$ of $\tree{t}$ in $\mc{T}$ such that
  $\str{\tree{s}} \hookrightarrow \str{\tree{t}}$, $\str{\tree{s}}
  \lequiv{m} \str{\tree{t}}$ and $|\tree{t}| - |\tree{s}| \leq
  \eta(m)$. Since $\mathsf{Str}$ is $\mc{L}$-great, it follows that
  $|\str{\tree{t}}| - |\str{\tree{s}}| \leq \beta(\eta(m))$. Whereby,
  either $|\str{\tree{s}}| \in \scale{i-1}_g$ or $|\str{\tree{s}}| \in
  \scale{i}_g$. If the former holds, then taking $\mf{B} =
  \str{\tree{s}}$, we are done. If the latter holds, then we apply
  Lemma~\ref{lemma:logical-fractal-and-tree-reps} recursively to
  $\tree{s}$ till eventually we get a subtree $\tree{x}$ of $\tree{t}$
  in $\mc{T}$ such that $\str{\tree{x}} \hookrightarrow
  \str{\tree{t}}$, $\str{\tree{x}} \lequiv{m} \str{\tree{t}}$ and
  $|\str{\tree{x}}| \in \scale{i-1}_g$. Then taking $\mc{B} =
  \str{\tree{x}}$, we are done.
\end{proof}

Indeed, the diverse spectrum of posets and graphs, including those
constructed using several quantifier-free operations, as seen in
Section~\ref{section:classes-satisfying-lebsp}, admit $\mc{L}$-great
tree representations, whereby they are all logical fractals. The
logical fractal property thus appears to be a natural property that
arises in a variety of interesting settings of computer science.

\vspace{-5pt}\section{Conclusion}\label{section:conclusion}

We presented a natural finitary analogue of the well-studied {\dls}
property from classical model theory, denoted $\reflebsp$, and showed
that this property is enjoyed by various classes of interest in
computer science, whereby all these classes can be seen to admit a
natural finitary version of the {\dlsfull}. The aforesaid classes
further admit linear time f.p.t. algorithms for $\fo$ and $\mso$ model
checking, when the structures in the classes are presented using their
natural tree representations. Finally, the aforesaid classes possess a
fractal like property, one based on logic. These observations open up
several interesting directions for future work. We mention below two
such directions that we find challenging:
\begin{enumerate}[nosep]
\item
Under what conditions on a class of structures, is it the case that
the index of the $\lequiv{m}$ relation over the class is an elementary
function of $m$? Investigating this question for classes that admit
elementary $\mc{L}$-good tree representations might yield insights for
getting linear time f.p.t. algorithms for $\fo$ and $\mso$ model
checking over these classes, that have elementary parameter
dependence.
\item
The $\reflebsp$ classes (resp. logical fractals) we have identified
are structurally defined. This motivates asking the converse, and
hence the following: Is there a structural characterization of
$\reflebsp$ (resp. of logical fractals)? We believe that an answer to
this question, even under reasonable assumptions, would yield new
classes that are well-behaved from both the logical and the
algorithmic perspectives.
\end{enumerate}

\vspace{5pt}

{{\tbf{Acknowledgments:}}} I express my deepest gratitude to Bharat 
Adsul for various insightful discussions and critical feedback that have 
helped in preparing this article.

\vspace{-5pt}
\bibliography{refs}

\begin{thebibliography}{10}

\bibitem{gurevich-ajtai}
Miklos Ajtai and Yuri Gurevich.
\newblock Monotone versus positive.
\newblock {\em J. ACM}, 34(4):1004--1015, October 1987.

\bibitem{ajtai-gurevich94}
Miklos Ajtai and Yuri Gurevich.
\newblock {Datalog vs first-order logic}.
\newblock {\em J. Comput. Sys. Sci.}, 49(3):562 -- 588, 1994.

\bibitem{gurevich-alechina}
Natasha Alechina and Yuri Gurevich.
\newblock Syntax vs. semantics on finite structures.
\newblock In {\em Structures in Logic and Computer Science. A Selection of
  Essays in Honor of A. Ehrenfeucht}, pages 14--33. Springer-Verlag, 1997.

\bibitem{alur-madhu}
Rajeev Alur and Parthasarathy Madhusudan.
\newblock Adding nesting structure to words.
\newblock {\em J. {ACM}}, 56(3), 2009.

\bibitem{dawar-pres-under-ext}
Albert Atserias, Anuj Dawar, and Martin Grohe.
\newblock Preservation under extensions on well-behaved finite structures.
\newblock {\em SIAM J. Comput.}, 38(4):1364--1381, 2008.

\bibitem{dawar-hom}
Albert Atserias, Anuj Dawar, and Phokion~G. Kolaitis.
\newblock On preservation under homomorphisms and unions of conjunctive
  queries.
\newblock {\em J. ACM}, 53(2):208--237, 2006.

\bibitem{fractal-2}
Michael Barnsley.
\newblock {\em Fractals Everywhere}.
\newblock Academic Press Professional, Inc., 1988.

\bibitem{tata}
Hubert Comon, Max Dauchet, Remi Gilleron, Christof L\"oding, Florent
  Jacquemard, Denis Lugiez, Sophie Tison, and Marc Tommasi.
\newblock Tree automata techniques and applications.
\newblock Available at: \url{http://www.grappa.univ-lille3.fr/tata}, 2007.
\newblock release October 12, 2007.

\bibitem{cograph-1981-paper}
D.~G. Corneil, H.~Lerchs, and L.~S. Burlingham.
\newblock Complement reducible graphs.
\newblock {\em Discrete Applied Mathematics}, 3(3):163 -- 174, 1981.

\bibitem{tree-depth-FO-equals-MSO}
Michael Elberfeld, Martin Grohe, and Till Tantau.
\newblock Where first-order and monadic second-order logic coincide.
\newblock In {\em {LICS} 2012, Croatia, June 25-28, 2012}, pages 265--274,
  2012.

\bibitem{frick-grohe}
Markus Frick and Martin Grohe.
\newblock The complexity of first-order and monadic second-order logic
  revisited.
\newblock {\em Ann. Pure Appl. Logic}, 130(1-3):3--31, 2004.

\bibitem{shrub-depth-FO-equals-MSO}
Jakub Gajarsky and Petr Hlinen{\'{y}}.
\newblock Kernelizing {MSO} properties of trees of fixed height, and some
  consequences.
\newblock {\em Log. Meth. Comp. Sci.}, 11(19):1--26, 2015.

\bibitem{shrub-depth}
Robert Ganian, Petr Hlinen{\'{y}}, Jaroslav Ne\v{s}et\v{r}il, Jan
  Obdrz{\'{a}}lek, Patrice~Ossona de~Mendez, and Reshma Ramadurai.
\newblock When trees grow low: Shrubs and fast {MSO1}.
\newblock In {\em {MFCS} 2012, Bratislava, Slovakia, August 27-31, 2012}, pages
  419--430, 2012.

\bibitem{gradel-rosen}
Erich Gr{\"{a}}del and Eric Rosen.
\newblock On preservation theorems for two-variable logic.
\newblock {\em Math. Log. Quart.}, 45:315--325, 1999.

\bibitem{grohe-dls}
Martin Grohe.
\newblock Some remarks on finite {L{\"{o}}wenheim-Skolem} theorems.
\newblock {\em Math. Log. Q.}, 42:569--571, 1996.

\bibitem{model-theoretic-methods}
Martin Grohe and Johann~A. Makowsky.
\newblock {\em Model Theoretic Methods in Finite Combinatorics}, volume 558 of
  {\em Contemporary Mathematics}.
\newblock AMS, 2011.

\bibitem{gurevich84}
Yuri Gurevich.
\newblock {Toward logic tailored for computational complexity}.
\newblock In Michael M.~Richter et~al., editor, {\em {Computation and Proof
  Theory: Proceedings of the Logic Colloquium held in Aachen, July 18 - 23,
  1983, Part II}}, pages 175 -- 216. Springer-Verlag, 1984.

\bibitem{nicole-lmcs-15}
Frederik Harwath, Lucas Heimberg, and Nicole Schweikardt.
\newblock Preservation and decomposition theorems for bounded degree
  structures.
\newblock {\em Log. Meth. Comp. Sci.}, 11(4), 2015.

\bibitem{libkin}
Leonid Libkin.
\newblock {\em Elements of Finite Model Theory}.
\newblock Springer-Verlag, 2004.

\bibitem{lindstrom}
{Per} {Lindstr{\"o}m}.
\newblock A characterization of elementary logic.
\newblock In {S{\"o}ren} {Halld{\'e}n}, editor, {\em Modality, Morality and
  Other Problems of Sense and Nonsense}, pages 189--191. Lund,Gleerup, 1973.

\bibitem{lowenheim}
Leopold L{\"o}wenheim.
\newblock {\"U}ber m{\"o}glichkeiten im relativkalk{\"u}l.
\newblock {\em Mathematische Annalen}, 76(4):447--470, 1915.

\bibitem{makowsky}
Johann~A. Makowsky.
\newblock Algorithmic uses of the {Feferman-Vaught} theorem.
\newblock {\em Ann. Pure Appl. Logic}, 126(1-3):159--213, 2004.

\bibitem{maltsev}
Anatoly~I. Maltsev.
\newblock {Untersuchungen aus dem Gebiete der mathematischen Logik}.
\newblock {\em Matematicheskii Sbornik}, n.s.(1):323--336, 1936.

\bibitem{tree-depth}
Jaroslav Ne\v{s}et\v{r}il and Patrice~Ossona de~Mendez.
\newblock Tree-depth, subgraph coloring and homomorphism bounds.
\newblock {\em Eur. J. Comb.}, 27(6):1022--1041, 2006.

\bibitem{rosen-thesis}
Eric Rosen.
\newblock {\em Finite model theory and finite variable logics}.
\newblock PhD thesis, University of Pennsylvania, 1995.

\bibitem{rosen}
Eric Rosen.
\newblock Some aspects of model theory and finite structures.
\newblock {\em Bull. Symbolic Logic}, 8(3):380--403, 2002.

\bibitem{rosen-weinstein}
Eric Rosen and Scott Weinstein.
\newblock Preservation theorems in finite model theory.
\newblock In {\em Logical and Computational Complexity. Selected Papers. Logic
  and Computational Complexity, International Workshop {LCC} '94, Indianapolis,
  Indiana, USA, October 13-16, 1994}, pages 480--502, 1994.

\bibitem{rossman-hom}
Benjamin Rossman.
\newblock Homomorphism preservation theorems.
\newblock {\em J. ACM}, 55(3):15:1--15:53, 2008.

\bibitem{arxiv-self-thesis}
Abhisekh Sankaran.
\newblock A generalization of the {{\L}o{\'{s}}-Tarski} preservation theorem.
\newblock {\em CoRR}, abs/1609.06297, 2016.
\newblock URL: \url{http://arxiv.org/abs/1609.06297}.

\bibitem{abhisekh-mfcs}
Abhisekh Sankaran, Bharat Adsul, and Supratik Chakraborty.
\newblock A generalization of the {{\L}o{\'{s}}-Tarski} preservation theorem
  over classes of finite structures.
\newblock In {\em {MFCS} 2014, Budapest, Hungary, August 25-29, 2014, Part
  {I}}, pages 474--485, 2014.

\bibitem{abhisekh-apal}
Abhisekh Sankaran, Bharat Adsul, and Supratik Chakraborty.
\newblock A generalization of the {{\L}o{\'{s}}-Tarski} preservation theorem.
\newblock {\em Ann. Pure Appl. Logic}, 167(3):189--210, 2016.

\bibitem{stolboushkin}
Alexei~P. Stolboushkin.
\newblock Finitely monotone properties.
\newblock In {\em Proceedings of the 10th Annual IEEE Symposium on Logic in
  Computer Science, LICS 1995, San Diego, USA, June 26-29, 1995}, pages
  324--330. IEEE Computer Society, 1995.

\bibitem{vaananen-dls}
Jouko V{\"a}{\"a}n{\"a}nen.
\newblock Pseudo-finite model theory.
\newblock {\em Mat. Contemp}, 24(8th):169--183, 2003.

\end{thebibliography}
\appendix
\section{Proof of composition lemma for ordered trees}\label{section:appendix:composition-lemma-for-partially-ranked-trees}

We prove the composition lemma for ordered trees as given by
Lemma~\ref{lemma:mso-composition-lemma-for-ordered-trees}.

\begin{proof}[Proof of Lemma~\ref{lemma:mso-composition-lemma-for-ordered-trees}]

We present the result for $\mc{L} = \mso$. A similar proof can be done
for $\mc{L} = \fo$.  Without loss of generality, we assume
$\tree{t}_i$ and $\tree{s}_i$ have disjoint sets of nodes for $i \in
\{1. 2\}$. We show the result for part
(\ref{lemma:mso-composition-lemma-for-ordered-trees:part-1}) of
Lemma~\ref{lemma:mso-composition-lemma-for-ordered-trees}. The other
parts can be proved similarly. Let $\tree{z}_i = ({\tree{t}}_i
\cdot^{\rightarrow}_{a_i} {\tree{s}}_i)$ for $i \in \{1, 2\}$.

Let $\beta_1$ be the winning strategy of the duplicator in the $m$
round $\mef$ game between $(\tree{t}_1, a_1)$ and $(\tree{t}_2,
a_2)$. Let $\beta_2$ be the winning strategy of the duplicator in the
$m$ round $\mef$ game between $\tree{s}_1$ and $\tree{s}_2$. Observe
that since $m \ge 2$, we can assume that $\beta_2$ is such that if in
any round, the spoiler picks $\troot{\tree{s}_1}$
(resp. $\troot{\tree{s}_2}$), then $\beta_2$ will require the
duplicator to pick $\troot{\tree{s}_2}$
(resp. $\troot{\tree{s}_1}$). We use this observation later on.  The
strategy $\alpha$ of the duplicator in the $m$-round $\mef$ game
between $(\tree{z}_1, a_1)$ and $(\tree{z}_2, a_2)$ is defined as
follows:
\begin{enumerate}[nosep]
\item Point move: (i) If the spoiler picks an element of $\tree{t}_1$
  (resp. $\tree{t}_2$), the duplicator picks the element of
  $\tree{t}_2$ (resp. $\tree{t}_1$) given by $\beta_1$.  (ii) If the
  spoiler picks an element of $\tree{s}_1$ (resp. $\tree{s}_2$), the
  duplicator picks the element of $\tree{s}_2$ (resp. $\tree{s}_1$)
  given by $\beta_2$.
\item Set move: If the spoiler picks a set $X$ from $\tree{z}_1$, then
  let $X = Y_1 \sqcup Y_2$ where $Y_1$ is a set of elements of
  $\tree{t}_1$ and $Y_2$ is a set of elements of $\tree{s}_1$. Let
  $Y_1'$ and $Y_2'$ be the sets of elements of $\tree{t}_2$ and
  $\tree{s}_2$ respectively, chosen according to strategies $\beta_1$
  and $\beta_2$. Then in the game between $(\tree{z}_1, a_1)$ and
  $(\tree{z}_2, a_2)$, the duplicator responds with the set $X' = Y_1'
  \cup Y_2'$. A similar choice of set is made by the duplicator from
  $\tree{z}_1$ when the spoiler chooses a set from $\tree{z}_2$.
\end{enumerate}

We now show that the strategy $\alpha$ is winning for the duplicator
in the $m$-round $\mef$ game between $(\tree{z}_1, a_1)$ and 
$(\tree{z}_2, a_2)$.

Let at the end of $m$ rounds, the vertices and sets chosen from
$\tree{z}_1$, resp. $\tree{z}_2$, be $e_1, \ldots, e_p$ and $E_1,
\ldots, E_r$, resp. $f_1, \ldots, f_p$ and $F_1, \ldots, F_r$, where
$p + r = m$. For $l \in \{1, \ldots, r\}$, let $E_l^t$, resp. $E_l^s$
be the intersection of $E_l$ with the nodes of $\tree{t}_1$,
resp. nodes of $\tree{s}_1$, and likewise, let $F_l^t$, resp. $F_l^s$
be the intersection of $F_l$ with the nodes of $\tree{t}_2$,
resp. nodes of $\tree{s}_2$.

Firstly, it is straightforward to verify that the labels of $e_i$ and
$f_i$ are the same for all $i \in \{1, \ldots, p\}$, and that for
$l \in \{1, \ldots, r\}$, $e_i$ is in $E_l^s$, resp. $E_l^t$, iff
$f_i$ is in $F_l^s$, resp. $F_l^t$, whereby $e_i \in E_l$ iff $f_i \in
F_l$.  For $1 \leq i, j \leq p$, if $e_i$ and $e_j$ both belong to
$\tree{t}_1$ or both belong to $\tree{s}_1$, then it is clear from the
strategy $\alpha$ described above, that $f_i$ and $f_j$ both belong
resp. to $\tree{t}_2$ or both belong to $\tree{s}_2$. It is easy to
verify from the description of $\alpha$ that for every binary relation
(namely, the ancestor-descendent-order $\leq$, and the
ordering-on-the-children-order $\lesssim$), the pair $(e_i, e_j)$ is
in the binary relation in $\tree{z}_1$ iff $(f_i, f_j)$ is in that
binary relation in $\tree{z}_2$.  Consider the case when without loss
of generality, $e_1 \in \tree{t}_1$ and $e_2 \in \tree{s}_1$. Then
$f_1 \in \tree{t}_2$ and $f_2 \in \tree{s}_2$. We have the following
cases. Assume that the ordered tree underlying $\tree{z}_i$ is
$((A_i, \leq_i), \lesssim_i)$ for $i \in \{1, 2\}$.
\begin{enumerate}[nosep]
\item $e_1 \lesssim_1 a_1$ and $e_2 = \troot{\tree{s}_1}$: Then we see
  that $f_1 \lesssim_2 a_2$ and $f_2 = \troot{\tree{s}_2$}. Observe that
  $f_2$ must be $\troot{\tree{s}_2}$ by the property of $\beta_2$
  stated at the outset. Whereby $e_1 \lesssim_1 e_2$ and $f_1
  \lesssim_2 f_2$. Likewise $e_1 \not\leq_1 e_2$, $e_2 \not\leq_1 e_1$ and
  $f_1 \not\leq_2 f_2$, $f_2 \not\leq_2 f_1$.
\item $e_1 \lesssim_1 a_1$ and $e_2 \neq \troot{\tree{s}_1}$: Then we
  see that $f_1 \lesssim_2 a_2$ and $f_2 \neq \troot{\tree{s}_2$}
  (again by the property of $\beta_2$ stated at the outset). Whereby
  $e_1 \not\lesssim_1 e_2$, $e_2 \not\lesssim_1 e_1$ and $f_1
  \not\lesssim_2 f_2$, $f_2 \not\lesssim_2 f_1$. Likewise, $e_1
  \not\leq_1 e_2$, $e_2 \not\leq_1 e_1$ and $f_1 \not\leq_2 f_2$, $f_2
  \not\leq_2 f_1$.
\item $a_1 \lesssim_1 e_1$, $a_1 \neq e_1$ and $e_2 =
  \troot{\tree{s}_1}$: Then we see that $a_2 \lesssim_2 f_1$, $a_2
  \neq f_1$ and $f_2 = \troot{\tree{s}_2$}. Observe that $f_2$ must be
  $\troot{\tree{s}_2}$ by the property of $\beta_2$ stated at the
  outset. Whereby $e_2 \lesssim_1 e_1$ and $f_2 \lesssim_2
  f_1$. Likewise $e_1 \not\leq_1 e_2$, $e_2 \not\leq_1 e_1$ and $f_1
  \not\leq_2 f_2$, $f_2 \not\leq_2 f_1$.
\item $a_1 \lesssim_1 e_1$, $a_1 \neq e_1$ and $e_2 \neq
  \troot{\tree{s}_1}$: Then we see that $a_2 \lesssim_2 f_1$, $a_2
  \neq f_1$ and $f_2 \neq \troot{\tree{s}_2$} (again by the property
  of $\beta_2$ stated at the outset). Whereby $e_1 \not\lesssim_1
  e_2$, $e_2 \not\lesssim_1 e_1$ and $f_1 \not\lesssim_2 f_2$, $f_2
  \not\lesssim_2 f_1$. Likewise, $e_1 \not\leq_1 e_2$, $e_2 \not\leq_1
  e_1$ and $f_1 \not\leq_2 f_2$, $f_2 \not\leq_2 f_1$.
\item $e_1 \neq a_1, e_1 \leq_1 a_1$: Then $f_1 \neq a_1, f_1 \leq_2
  a_2$. Whereby $e_1 \leq_1 e_2$ and $f_1 \leq_2 f_2$. This is because
  $e_1 \leq_1 c_1$ and $f_1 \leq_2 c_2$ where $c_1$ and $c_2$ are
  resp. the parents of $a_1$ and $a_2$ in $\tree{z}_1$ and
  $\tree{z}_2$. Also $e_1 \not\lesssim_1 e_2$, $e_2 \not\lesssim_1
  e_1$ and $f_1 \not\lesssim_2 f_2$, $f_2 \not\lesssim_2 f_1$.
\item $e_1$ and $e_2$ are not related by $\leq_1$ or $\lesssim_1$: Then
  $f_1$ and $f_2$ are also not related by $\leq_2$ or $\lesssim_2$.
\end{enumerate}
In all cases, we have that the pair $(e_i, e_j)$ is in $\leq_1$
(resp. $\lesssim_1$) iff $(f_i, f_j)$ is in $\leq_2$
(resp. $\lesssim_2$).
\end{proof}

\end{document}